\documentclass[12pt]{article}
\usepackage{setspace}
\usepackage{amsmath, amsfonts, amssymb, amsthm}
\usepackage{bm}
\usepackage{tikz}
\usepackage{float}
\usepackage{algorithm}
\usepackage{algpseudocode}
\usepackage{url}
\usepackage{subcaption}
\usepackage{natbib}
\usepackage[normalem]{ulem}
\usepackage[margin=1in]{geometry}
\usepackage{booktabs}
\usepackage{multirow}
\usepackage{calc}
\usepackage{titlesec}
\usepackage{mathtools}
\mathtoolsset{showonlyrefs}

\newcommand{\norm}[1]{\left\lVert#1\right\rVert}
\newcommand{\Diag}{\mathrm{Diag}}
\newcommand{\conv}{\mathrm{conv}}
\newcommand{\relint}{\mathrm{relint}}
\newcommand{\tr}{\mathrm{tr}}
\newcommand{\cone}{\mathrm{cone}}
\newcommand{\aff}{\mathrm{aff}}
\newcommand{\supp}{\mathrm{supp}}
\newcommand{\diag}{\mathrm{diag}}
\usepackage[normalem]{ulem}

\newcommand{\R}{\mathbb{R}}
\newcommand{\cD}{\mathcal{D}}
\newcommand{\cF}{\mathcal{F}}
\newcommand{\cV}{\mathcal{V}}
\newcommand{\cN}{\mathcal{N}}
\newcommand{\inner}[2]{\langle #1, #2\rangle}
\DeclareMathOperator*{\argmax}{argmax}
\DeclareMathOperator*{\argmin}{argmin}
\makeatletter
\newtheoremstyle{plaindot}{}{}{\itshape}{}{\bfseries}{}{ }{\thmname{#1}\thmnumber{ #2}.\thmnote{ (#3)}}
\makeatother
\theoremstyle{plaindot}

\newtheorem{definition}{Definition}
\newtheorem{lemma}{Lemma}
\newtheorem{proposition}{Proposition}

\newtheorem{assumption}{Assumption}

\titleformat{\section}{\normalfont\normalsize\bfseries}{\thesection.}{1em}{}
\titleformat{\subsection}{\normalfont\normalsize\bfseries}{\thesubsection.}{1em}{}
\titleformat{\subsubsection}{\normalfont\small\bfseries}{\thesubsubsection.}{1em}{}
\titlespacing*{\section}{0pt}{3.5ex plus 1ex minus .2ex}{2.3ex plus .2ex}
\titlespacing*{\subsection}{0pt}{3.25ex plus 1ex minus .2ex}{1.5ex plus .2ex}
\titlespacing*{\subsubsection}{0pt}{3.25ex plus 1ex minus .2ex}{1.5ex plus .2ex}

\newcommand{\blfootnote}[1]{\begingroup\renewcommand\thefootnote{}\footnotetext{#1}\endgroup}
\newcommand{\tabnotes}[1]{\par\vspace{4pt}\begin{minipage}{0.95\linewidth}\footnotesize\emph{Notes:} #1\end{minipage}}
\makeatletter
\renewcommand\fs@ruled{\def\@fs@cfont{\bfseries}\let\@fs@capt\floatc@ruled
  \def\@fs@pre{\hrule height.8pt depth0pt \kern2pt}%
  \def\@fs@post{}%
  \def\@fs@mid{\kern2pt\hrule\kern2pt}%
  \let\@fs@iftopcapt\iftrue}
\makeatother
\newcommand{\algnotes}[1]{\par\kern2pt\hrule\relax\tabnotes{#1}}

\usepackage{xcolor}

\newcommand{\sgn}{\mathrm{sgn}}

\title{Synthetic Control:~An Exact Geometric Formulation}
\author{Alberto Abadie\\ MIT
  \and Yucheng Shang\\  MIT}
\date{\today}

\begin{document}

\maketitle
\thispagestyle{empty}
\blfootnote{\hspace*{-0.298in} Alberto Abadie: Department of Economics, Massachusetts Institute of Technology, abadie@mit.edu. Yucheng Shang: Department of Economics, Massachusetts Institute of Technology, ycshang@mit.edu. We used Claude (Anthropic), ChatGPT (OpenAI), and Refine.ink to assist with the preparation of this manuscript. We gratefully acknowledge support from the Office of Naval Research under grant N00014-24-1-2687. Replication files  are available at \url{https://github.com/lychemon/synth-AS-replication}. A Stata command \texttt{synthx}, which implements the estimator for general use, is available at \url{https://github.com/lychemon/synthx_package}.}

\vspace{-2em}
\begin{abstract}
\noindent The synthetic control method matches a treated unit to a convex combination of donor units, with the relative importance of each predictor determined by a weight matrix $V$. In standard practice, $V$ is selected through a bi-level optimization problem that is nonconvex and lacks convergence guarantees. This article derives an exact geometric solution for synthetic control estimators. For any positive definite $V$, a synthetic control is given by the projection of the treated unit's predictors onto the convex hull of the donor predictors under a $V$-weighted norm. As $V$ varies over its candidate class, these projections trace out the set of attainable synthetic controls. We characterize this set. When the predictor vector of the treated unit lies outside the convex hull of the donor predictors, the attainable set under general positive definite $V$ is the union of the facets of the donor hull that are visible from the treated unit. For diagonal $V$, a sign-compatibility test identifies the attainable set.  When the treated unit lies inside the hull, exact matches exist. We select the one that minimizes the average of the pairwise matching discrepancies between the treated unit and each of the donors contributing to the synthetic control. 
These characterizations replace bi-level optimization with a fast enumeration method that is guaranteed to deliver the global optimum over the attainable set. Moreover, the Karush-Kuhn-Tucker conditions for the selection of donor weights characterize the set of  positive definite predictor weight matrices that rationalize any given set of donor weights. Applications to the California tobacco control program and to the economic effects of German reunification illustrate the advantages of the proposed procedure over existing practice.
\end{abstract}

%------------------------------------------------------------------
\section{Introduction}\label{sec:intro}
%------------------------------------------------------------------

The synthetic control method \citep{Abadie2003Basque,Abadie2010California,Abadie2015Germany} estimates a counterfactual for a treated unit as a weighted combination of untreated donor units. In the standard formulation, the donor weights $\bm{w}$ minimize the distance between the pre-treatment predictors of the synthetic control and those of the treated unit. A predictor weight matrix $V$, chosen either by the researcher or through a data-driven procedure, defines the metric that determines the relative importance of each predictor in constructing the donor weights. 

Standard implementations choose a diagonal $V$ through a heuristic bi-level procedure \citep{Abadie2003Basque}. For each candidate predictor weight matrix, an inner loop selects donor weights to minimize the distance between the predictors of the treated unit and those of the synthetic control. An outer loop then searches for predictor weights that minimize pre-treatment prediction error on a validation set.

Although this procedure produces reasonable solutions in many applications, the optimization can be fragile. The bi-level problem is nonconvex, and its objective function is only piecewise smooth. It is, moreover, typically discontinuous at the boundary of the parameter space, where zero predictor weights reduce the effective dimension of the predictors. Bi-level gradient-based solvers may therefore converge to local minima or to solutions that discard information on the user-provided predictors.

This article develops a geometric framework for synthetic control estimators. For any positive definite predictor weight matrix $V$, the donor weights are determined by the $V$-weighted projection of the treated unit’s predictor vector onto the convex hull of the donor predictor vectors. As the predictor weights $V$ vary over a prespecified candidate set, these projections trace out a subset of the convex hull.

From this geometric perspective, the search over $V$ regularizes the choice of donor weights. Instead of searching over the full simplex of donor weights, the procedure searches only over the weights that some candidate $V$ rationalizes, the \emph{attainable set}. The geometric framework developed here characterizes this set. 

When the candidate set for $V$ contains all positive definite matrices, the attainable projections are exactly the points on the {\em visible facets} of the donor convex hull, that is, the facets oriented toward the predictor vector of the treated
unit. 
This characterization allows us to replace the bi-level
optimization problem with a single-level problem that
enumerates the visible facets and minimizes the validation loss over the donor
weights supported on each one. The minimum over this finite set is the global
minimum. When the candidate set for $V$ consists instead of positive diagonal matrices, a sign-compatibility test determines which points on the visible faces are attainable. In an application to the effects of a tobacco control program in California \citep{Abadie2010California}, the donor hull has $26{,}282$ faces. Visibility narrows the candidate set to $3{,}839$ faces, and the sign-compatibility condition for diagonal $V$ leaves only the $226$ attainable faces, fewer than one percent of the total (see Table \ref{tab:faces} in Section \ref{sec:numerical}).

This article makes four contributions. First, it characterizes the set of attainable donor weights under both diagonal and general positive definite predictor weight matrices. Second, it replaces a computationally unstable nested optimization with a combinatorial enumeration procedure guaranteed to find the global optimum. Third, it clarifies how the set of attainable donor weights differs between diagonal and positive definite predictor weight matrices.
Fourth, the Karush-Kuhn-Tucker conditions of the inner optimization characterize the set of positive definite predictor weight matrices that rationalize any set of donor weights. Two applications illustrate the advantages of the methods proposed in this article: an evaluation of a tobacco control program in California \citep{Abadie2010California} and an analysis of the effects of German reunification \citep{Abadie2015Germany}.

This article relates to several strands in the synthetic control literature. \citet{BECKER20181} document the fragility of bi-level optimization procedures for synthetic controls and propose solutions to this problem (such as discarding the donors that cannot receive positive weight under any positive definite $V$, which corresponds to our visibility restriction). Their approach, however, remains within the framework of numerical search over values of diagonal $V$ and offers no guarantees of global optimality. 
Also within the bi-level formulation, \citet{Vives2022} adds a penalty on the number of predictors to stabilize the outer problem. \citet{Pouliot2026} derive the degrees of freedom of synthetic control estimators and propose information criteria for selecting $V$. \citet{Malo2023bilevel} study a bi-level program in which ties in the inner problem are resolved in favor of the validation criterion. As a result, their program favors extremely sparse choices for the diagonal values in $V$. Zeros in the diagonal of $V$ enlarge the set of solutions to the inner problem. The tie-breaking rule then selects donor weights that lower the validation loss. However, the resulting synthetic controls need not be projections of the treated unit onto the donor hull under a positive definite $V$ or limits of such projections. In contrast, this article optimizes only over attainable projections and their limits. It also characterizes the entire solution set for $V$, extending the methods beyond the diagonal class to general positive definite weights.

A second group of estimators avoids the bi-level structure by specifying the predictor weight matrix $V$ in advance, so the donor weights solve a single convex program. \citet{Doudchenko2016} fit pre-treatment outcomes directly and relax the constraints on the donor weights. In settings with poor pre-treatment fit, \citet{FermanPinto2021} demean the outcomes and study the asymptotic bias of the resulting estimator.  
\citet{BenMichael2021} correct the synthetic control fit with an outcome model. \citet{Kellogg2021} combine synthetic control weights with nearest-neighbor matching to trade off interpolation and extrapolation biases. \citet{Arkhangelsky2021} combine unit and time weights in a weighted difference-in-differences regression. \citet{abadie2021penalized} add a penalty on the pairwise discrepancies between the treated unit and each donor, which makes the weights unique and sparse.  \citet{Cattaneo2021} and \citet{chernozhukov2021conformal} develop uncertainty quantification for this class of estimators \citep[see also][]{chernozhukov2026ttest, Cattaneo2025}. \citet{Gunsilius2023} takes the predictors to be the entire distributions of pre-treatment outcomes and matches the quantile function of the treated unit to a weighted average of the donors' quantile functions. In all these cases, the metric on the predictors is not estimated as part of the donor weight selection. This article instead retains the search over $V$ and proposes a procedure with a global optimality guarantee.

The rest of the article is organized as follows. Section~\ref{sec:setup} reviews the synthetic control framework and formalizes the predictor weight problem. Section~\ref{sec:predictor_weights} examines the common practice of restricting the predictor weights to a diagonal matrix and describes the boundary problem that motivates the geometric analysis of this article. Section~\ref{sec:geometric} develops our geometric formulation of synthetic controls for the usual case when the treated unit is outside the convex hull of the available donor units in the space of the predictors.
Section~\ref{sec:interior_results} extends the analysis to the case in which the treated unit lies inside the donor hull. 
Section~\ref{sec:numerical} illustrates the approach with an application to the California tobacco control program. Section~\ref{sec:temporal_split} considers a leakage-free validation design that separates the predictor and validation windows in time, illustrated with the German reunification application. Section~\ref{sec:conclusion} concludes.

\paragraph{Notation.} For a vector $\bm{x} \in \mathbb{R}^K$, $x_k$ denotes its $k$-th coordinate, $\bm{x}'$ its transpose, $\norm{\bm{x}}$ its Euclidean norm, and $\sgn(\bm{x})$ its componentwise sign, with values in $\{-1, 0, +1\}^K$. For a symmetric matrix $V$, $V \succ 0$ ($V \succeq 0$) indicates that $V$ is positive (semi)definite, and $\norm{\bm{x}}_V = (\bm{x}'V\bm{x})^{1/2}$ for $V \succeq 0$. $\mathbb{S}^K_{++}$ is the set of symmetric positive definite $K \times K$ matrices, $\mathbb{S}^K_{+}$ the set of symmetric positive semidefinite $K \times K$ matrices, $\Diag(\bm{d})$ the diagonal matrix with the entries of $\bm{d}$ on its diagonal, $\norm{\bm{d}}_0$ the number of nonzero entries of $\bm{d}$, $\diag(M)$ the vector of diagonal entries of a square matrix $M$, and $\tr(M)$ its trace. $\Delta^m$ denotes the unit simplex, the set of $\bm{x} \in \mathbb{R}^m$ with non-negative coordinates summing to one. $\mathbf 1$ denotes a vector of ones. $\mathbb{R}^K_+$ and $\mathbb{R}^K_{++}$ denote the vectors in $\mathbb{R}^K$ with non-negative and with strictly positive coordinates. The coordinates of a donor weight vector $\bm{w}$ are indexed by $j = 2, \ldots, J+1$. $\supp(\bm{w}) = \{j : w_j \neq 0\}$ is the support of $\bm{w}$. For a set $A \subseteq \mathbb{R}^K$, $\conv(A)$, $\cone(A)$, and $\aff(A)$ denote its convex, conical, and affine hulls, and $\mathrm{Int}(A)$, $\relint(A)$, $\overline A$, and $\partial A$ its interior, relative interior, closure, and boundary. A \emph{face} of a polytope is its intersection with a supporting hyperplane. \emph{Facets}, \emph{ridges}, and \emph{vertices} are the faces of dimension one less than the polytope, two less, and zero, respectively.

%------------------------------------------------------------------
\section{Synthetic Control and Predictor Weights}\label{sec:setup}
%------------------------------------------------------------------

The synthetic control method estimates the causal effect of an intervention by constructing a weighted average of untreated units that approximates the characteristics of the treated unit. This section reviews the standard setup and introduces the predictor weight problem that motivates our geometric formulation.

\subsection{Setup}\label{sec:setup_notation}

Consider a panel of $J+1$ units observed over $T$ periods. Unit 1 is exposed to some intervention beginning in period $T_0 + 1$, while the other $J$ units (the \emph{donor pool}) remain untreated throughout. 

For each unit $j$ and time $t$, the potential outcome without the intervention is $Y_{jt}^N$. For the treated unit and $t > T_0$, the potential outcome under the intervention is $Y_{1t}^I$. The observed outcome data take the form
\begin{align*}
 \begin{pmatrix}
        Y_{11}^N & Y_{12}^N & \cdots & \textcolor{blue}{Y_{1(T_0+1)}^{I}} & \cdots & \textcolor{blue}{Y_{1T}^{I}}\\
        Y_{21}^N & Y_{22}^N & \cdots & Y_{2(T_0+1)}^N & \cdots & Y_{2T}^N \\
        \vdots & \vdots & \ddots & \vdots & \ddots & \vdots \\
        Y_{(J+1)1}^N & Y_{(J+1)2}^N & \cdots & Y_{(J+1)(T_0+1)}^N & \cdots & Y_{(J+1)T}^N
    \end{pmatrix}.
\end{align*}
The parameters of interest are the treatment effects on the treated unit at times $t> T_0$, 
\begin{equation}\label{eq:estimand}
    \tau_{1t} = Y_{1t}^I - Y_{1t}^N, \quad t = T_0+1, \ldots, T.
\end{equation}
Because $Y_{1t}^I$ is observed for $t > T_0$, estimation of $\tau_{1t}$ comes down to estimation of the counterfactual outcome $Y_{1t}^N$ for $t > T_0$.

A $K$-vector of predictors $X_j = (X_{j1}, \ldots, X_{jK})'$ summarizes unit $j$'s characteristics relevant for forecasting untreated outcomes. These predictors may include pre-treatment outcome values, covariates, or transformations thereof. The $K \times J$ matrix $X_0 = (X_2, \ldots, X_{J+1})$ collects the predictor vectors of all donor units. We make the following assumption about the predictor vectors.

\begin{assumption}[General position]\label{ass:gp}
(i) $X_1 \notin \conv(\{X_2, \ldots, X_{J+1}\})$.  (ii) For every nonempty subset $S$ of predictors, any $|S|+1$ or fewer donor predictor vectors, restricted to the predictors in $S$, are affinely independent. (iii) $J\geq K+1$.
\end{assumption}

Part (i) states that the predictor vector of the treated unit cannot be reproduced exactly by a convex combination of donor predictor vectors, which is the case of primary interest in applications \citep{Abadie2021using}. Section \ref{sec:interior_results} relaxes this assumption. 
Part (ii) implies that every face of the convex hull of the donor predictors 
is a simplex containing no donor points other than its vertices. As a result, each boundary point of the hull has a unique representation as a convex combination of units in the donor pool. 
Parts (ii) and (iii) imply that the convex hull of the donor predictor vectors is full-dimensional in $\R^K$.

A selected subset of pre-treatment outcomes $Z_j = (Y_{jt_1}, \ldots, Y_{jt_L})'$, where $\{t_1, \ldots, t_L\} \subseteq \{1, \ldots, T_0\}$, serves as a validation sample. The $L \times J$ matrix $Z_0 = (Z_2, \ldots, Z_{J+1})$ stacks these validation outcomes across donors, and $Z_1 = (Y_{1t_1}, \ldots, Y_{1t_L})'$ collects the same outcomes for the treated unit.

\paragraph{The Synthetic Control Estimator}

The synthetic control estimator approximates the missing counterfactual $Y_{1t}^N$ by a convex combination of donor outcomes, 
\begin{equation}\label{eq:estimator}
    \widehat{\tau}_{1t} = Y_{1t} - \sum_{j=2}^{J+1} w_j Y_{jt}, \quad t = T_0+1, \ldots, T,
\end{equation}
where the donor weights belong to the unit simplex
\begin{equation}\label{eq:simplex}
    \Delta^J = \left\{ \bm{w} \in \mathbb{R}^J : w_j \geq 0 \text{ for all } j = 2, \ldots, J+1, \;\; \sum_{j=2}^{J+1} w_j = 1 \right\}.
\end{equation}
The nonnegativity and summing-to-one constraints prevent extrapolation and ensure the synthetic control is a weighted average of untreated units.

\subsection{The Predictor Weight}\label{sec:predictor_weight}

Because the donor weights $\bm{w}$ are selected to minimize the discrepancy between the predictor vectors of the treated unit and those of the synthetic control, synthetic control estimation requires specifying the relative importance assigned to each predictor. Given a $K \times K$ positive definite matrix $V$, a synthetic control estimator selects donor weights to solve:
\begin{align}\label{eq:inner}
    \mathcal{W}(V) &=\argmin_{\bm{w} \in \Delta^J} \norm{X_1 - X_0 \bm{w}}_V^2\nonumber\\ &= \argmin_{\bm{w} \in \Delta^J} (X_1 - X_0\bm{w})' V (X_1 - X_0\bm{w}).
\end{align}
Under Assumption~\ref{ass:gp}, $\mathcal{W}(V)$ is a singleton for every positive definite $V$, and we then write $\bm{w}(V)$ for its unique element. For diagonal $V$, a large entry $V_{kk}$ favors close matching on the $k$-th predictor while a small entry allows for a greater discrepancy. 

The choice of $V$ can reflect a subjective assessment of the predictors' relative importance, but often $V$ is selected from the data through the bi-level procedure,
\begin{equation}\label{eq:outer}
 V \in \argmin_{V \in \mathcal{V}}\min_{\bm{w} \in \mathcal{W}(V)} \norm{Z_1 - Z_0 \bm{w}}^2,
\end{equation}
where $\mathcal{V}$ is a prespecified set of candidate weighting matrices (commonly, diagonal with non-negative entries that sum to one), and $\mathcal{W}(V)$ is the set of solutions to the inner optimization problem in equation \eqref{eq:inner}.  
The validation error $\norm{Z_1 - Z_0 \bm{w}}^2$ measures the quality of the synthetic control fit on the validation set.

The inner optimization~\eqref{eq:inner} projects the treated unit's predictor vector $X_1$ onto the convex hull $\conv(\{X_2, \ldots, X_{J+1}\})$ using the norm induced by $V$. Figure~\ref{fig:scm} depicts the geometry of the problem in two dimensions. The shaded region represents the convex hull of the donor predictor values. The treated unit's predictor $X_1$ lies outside this hull, and its projection onto the hull is $X_0\bm{w}$. Different choices of $V$ project $X_1$ onto different points of the boundary of the convex hull. From this standpoint, the predictor weight problem reduces to selecting, among all projections attainable through candidate $V$, the one that best predicts the validation outcomes.

\begin{figure}[t]
    \centering
    \caption{Geometry of the synthetic control problem}
    \label{fig:scm}
    \begin{tikzpicture}[scale=1.5, thick]

    \coordinate (X_6) at (-0.2,0.1);
    \coordinate (X_2) at (1.5,-1.5);
    \coordinate (X_3) at (1,0);
    \coordinate (X_4) at (-0.5,1);
    \coordinate (X_5) at (0.8,-0.3);

    \filldraw[fill=blue!10, draw=blue!60, thick]
        (X_6) -- (X_2) -- (X_3) -- (X_4) -- cycle;

    \draw[black] (-1,-2) -- (-1, 1.5);
    \draw[black] (-1,-2) -- (2,-2);

    \filldraw[blue] (X_6) circle (1.2pt);
    \filldraw[blue] (X_2) circle (1.2pt);
    \filldraw[blue] (X_3) circle (1.2pt);
    \filldraw[blue] (X_4) circle (1.2pt);
    \filldraw[blue] (X_5) circle (1.2pt);

    \foreach \p/\pos in {X_6/below left, X_2/below right, X_3/below right, X_4/above left}
      \node[\pos] at (\p) {$\p$};

    \coordinate (P) at (1.2,0.9);
    \filldraw[red] (P) ++(0,2.2pt) -- ++(2.2pt,-2.2pt) -- ++(-2.2pt,-2.2pt) -- ++(-2.2pt,2.2pt) -- cycle;
    \node[right] at (P) {$X_1$};

    \coordinate (Y_1) at (0.72, 0.18);
    \node[below left] at (Y_1) {\footnotesize $X_0\bm{w}$};
    \filldraw[orange] (Y_1) ++(-1.5pt,-1.5pt) rectangle ++(3pt,3pt);

    \draw[dashed, gray] (P) -- (Y_1);

    \filldraw[blue] (X_5) circle (1.2pt);
    \node[below] at (X_5) {$X_5$};

    \end{tikzpicture}
    \tabnotes{The shaded region is the convex hull of the donor predictor values $X_2, \ldots, X_6$. $X_0\bm{w}$ is the $V$-weighted projection of $X_1$ onto the hull.}
\end{figure}
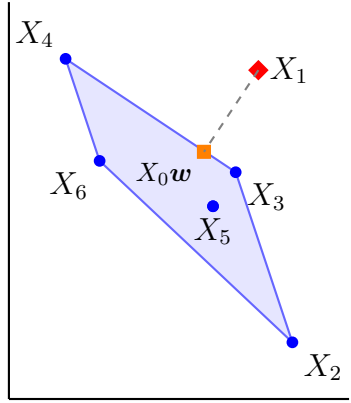

The next section examines the case of diagonal $V$ and discusses the difficulties that arise when the outer optimization drives $V$ to the boundary of the parameter space.

%------------------------------------------------------------------
\section{Bi-level Optimization with Diagonal $\boldsymbol{V}$}\label{sec:predictor_weights}
%------------------------------------------------------------------

\subsection{Diagonal $\boldsymbol{V}$ as Axis Rescaling}\label{sec:diagonal}

A common candidate set for predictor weights consists of all diagonal matrices with non-negative entries that sum to one:
\begin{equation}\label{eq:diagonal_feasible}
    \mathcal{V}_\Delta = \left\{  \Diag(\bm{d}) : \bm{d} \in \Delta^K \right\},
\end{equation}
with $\Delta^K = \{\bm{d} \in \mathbb{R}^K : d_k \geq 0, \, \sum_{k=1}^K d_k = 1\}$.
For diagonal $V$, the $V$-weighted norm is
\begin{align}\label{eq:diagonal_norm}
    \norm{X_1 - X_0\bm{w}}_V^2 &= (X_1 - X_0\bm{w})' V (X_1 - X_0\bm{w})\nonumber\\&= \sum_{k=1}^K d_k (X_{1k} - X_{0k}\bm{w})^2\nonumber\\&= \norm{\sqrt{V}(X_1 - X_0\bm{w})}^2,
\end{align}
where $\sqrt{V} = \Diag(\sqrt{d_1}, \ldots, \sqrt{d_K})$. The transformation $\sqrt{V}$ rescales the coordinate axes. A large $d_k$ stretches the $k$-th axis, magnifying discrepancies in the $k$-th predictor and forcing the inner optimization to match that predictor closely. A small $d_k$ compresses the $k$-th axis, diminishing its influence on the donor weights. 

The outer problem~\eqref{eq:outer} searches over the $(K-1)$-simplex of diagonal weights to minimize validation error.

\subsection{The Boundary Problem}\label{sec:boundary}

The solution to the inner problem can change discontinuously when the outer optimization sets one or more diagonal entries of $V$ to zero. A zero $d_k$ eliminates the $k$-th predictor from the inner optimization, reducing the dimension of the predictors. This dimension reduction removes an entire coordinate axis, projecting the problem into a lower-dimensional subspace where the geometry of the convex hull can change substantially.

Figure~\ref{fig:projection} illustrates this problem. Both panels show the same donor predictors $\{X_2,\dots,X_6\}$ and their convex hull, along with the treated unit $X_1$ and its $V$-weighted projection $X_0\bm{w}$ onto that hull. The open markers along the bottom and left axes represent the one-dimensional collapses of $X_1, \ldots, X_6$ obtained by sending a diagonal entry to zero. The choice $d_2 = 0$ retains only the horizontal coordinate and projects every point onto the bottom axis, while $d_1 = 0$ retains only the vertical coordinate and projects onto the left axis. The two panels differ only in the location of $X_1$. In panel~\ref{fig:projection-a}, $X_1$ lies beyond the donors in both coordinates, so its one-dimensional projections fall onto the extreme donors $X_2$ and $X_4$. The donor weights $\bm{w}(V)$ vary continuously with $V$ as one of its diagonal entries reaches zero. By contrast, in panel~\ref{fig:projection-b} $X_1$ lies within the donors' one-dimensional ranges. For every $V \succ 0$ the projection falls on the edge $\overline{X_3 X_4}$ or $\overline{X_2 X_3}$, leaving other units with zero weights. However, when one diagonal entry of $V$ is set to zero, $X_1$ lands in the interior of the donors' range, where synthetic controls can fit it exactly (though not uniquely). Hence, $\mathcal{W}(V)$ can expand discontinuously as a diagonal entry reaches zero. Figure \ref{fig:feasible} depicts the set of attainable projections that result from the predictor configurations in Figures~\ref{fig:projection-a} and~\ref{fig:projection-b}.

\begin{figure}[t]
    \centering
    \caption{Dimension collapse under boundary predictor weights}
    \label{fig:projection}
    \begin{subfigure}[b]{0.48\textwidth}
    \centering
    \begin{tikzpicture}[scale=1.5, thick]

    \coordinate (X_6) at (-0.2,0.1);
    \coordinate (X_2) at (1.5,-1.5);
    \coordinate (X_3) at (1,0);
    \coordinate (X_4) at (-0.5,1);
    \coordinate (X_5) at (0.8,-0.3);

    \coordinate (XX_6) at (-0.2,-2);
    \coordinate (XX_2) at (1.5,-2);
    \coordinate (XX_3) at (1,-2);
    \coordinate (XX_4) at (-0.5,-2);
    \coordinate (XX_5) at (0.8,-2);

    \coordinate (XY_6) at (-1,0.1);
    \coordinate (XY_2) at (-1,-1.5);
    \coordinate (XY_3) at (-1,0);
    \coordinate (XY_4) at (-1,1);
    \coordinate (XY_5) at (-1,-0.3);

    \filldraw[fill=blue!10, draw=blue!60, thick]
        (X_6) -- (X_2) -- (X_3) -- (X_4) -- cycle;

    \filldraw[blue] (X_6) circle (1.2pt);
    \filldraw[blue] (X_2) circle (1.2pt);
    \filldraw[blue] (X_3) circle (1.2pt);
    \filldraw[blue] (X_4) circle (1.2pt);
    \filldraw[blue] (X_5) circle (1.2pt);

    \draw[blue] (XX_6) circle (1.2pt);
    \draw[blue] (XX_2) circle (1.2pt);
    \draw[blue] (XX_3) circle (1.2pt);
    \draw[blue] (XX_4) circle (1.2pt);
    \draw[blue] (XX_5) circle (1.2pt);

    \draw[blue] (XY_6) circle (1.2pt);
    \draw[blue] (XY_2) circle (1.2pt);
    \draw[blue] (XY_3) circle (1.2pt);
    \draw[blue] (XY_4) circle (1.2pt);
    \draw[blue] (XY_5) circle (1.2pt);

    \draw[black] (-1,-2) -- (-1, 1.9);
    \draw[black] (-1,-2) -- (2,-2);

    \foreach \p/\pos in {X_6/below left, X_2/below right, X_3/below right, X_4/above left}
      \node[\pos] at (\p) {$\p$};

    % X_1 sits on the Euclidean normal to the edge X_3X_4 through X_0 w = (0.72,0.18): (0.72,0.18) + (1,1.5)
    \coordinate (P) at (1.72,1.68);
    \coordinate (PX) at (1.72,-2);
    \coordinate (PY) at (-1,1.68);

    % X_1 drawn as a filled diamond, its axis projections as open diamonds (shape-coded for black-and-white legibility)
    \filldraw[red] (P) ++(0,2.2pt) -- ++(2.2pt,-2.2pt) -- ++(-2.2pt,-2.2pt) -- ++(-2.2pt,2.2pt) -- cycle;
    \draw[red] (PX) ++(0,2.2pt) -- ++(2.2pt,-2.2pt) -- ++(-2.2pt,-2.2pt) -- ++(-2.2pt,2.2pt) -- cycle;
    \draw[red] (PY) ++(0,2.2pt) -- ++(2.2pt,-2.2pt) -- ++(-2.2pt,-2.2pt) -- ++(-2.2pt,2.2pt) -- cycle;
    \node[right] at (P) {$X_1$};

    \coordinate (Y_1) at (0.72, 0.18);
    \node[below left] at (Y_1) {\footnotesize $X_0\bm{w}$};
    % projection X_0 w drawn as a filled square
    \filldraw[orange] (Y_1) ++(-1.5pt,-1.5pt) rectangle ++(3pt,3pt);

    \draw[dashed, gray] (P) -- (Y_1);
    \draw[dashed, gray] (P) -- (PX);
    \draw[dashed, gray] (P) -- (PY);

    \filldraw[blue] (X_5) circle (1.2pt);
    \node[below] at (X_5) {$X_5$};

    \end{tikzpicture}
    \caption{}
    \label{fig:projection-a}
    \end{subfigure}
    \hspace*{-1cm}
    \begin{subfigure}[b]{0.48\textwidth}
    \centering
    \begin{tikzpicture}[scale=1.5, thick]

    \coordinate (X_6) at (-0.2,0.1);
    \coordinate (X_2) at (1.5,-1.5);
    \coordinate (X_3) at (1,0);
    \coordinate (X_4) at (-0.5,1);
    \coordinate (X_5) at (0.8,-0.3);

    \coordinate (XX_6) at (-0.2,-2);
    \coordinate (XX_2) at (1.5,-2);
    \coordinate (XX_3) at (1,-2);
    \coordinate (XX_4) at (-0.5,-2);
    \coordinate (XX_5) at (0.8,-2);

    \coordinate (XY_6) at (-1,0.1);
    \coordinate (XY_2) at (-1,-1.5);
    \coordinate (XY_3) at (-1,0);
    \coordinate (XY_4) at (-1,1);
    \coordinate (XY_5) at (-1,-0.3);

    \filldraw[fill=blue!10, draw=blue!60, thick]
        (X_6) -- (X_2) -- (X_3) -- (X_4) -- cycle;

    \filldraw[blue] (X_6) circle (1.2pt);
    \filldraw[blue] (X_2) circle (1.2pt);
    \filldraw[blue] (X_3) circle (1.2pt);
    \filldraw[blue] (X_4) circle (1.2pt);
    \filldraw[blue] (X_5) circle (1.2pt);

    \draw[blue] (XX_6) circle (1.2pt);
    \draw[blue] (XX_2) circle (1.2pt);
    \draw[blue] (XX_3) circle (1.2pt);
    \draw[blue] (XX_4) circle (1.2pt);
    \draw[blue] (XX_5) circle (1.2pt);

    \draw[blue] (XY_6) circle (1.2pt);
    \draw[blue] (XY_2) circle (1.2pt);
    \draw[blue] (XY_3) circle (1.2pt);
    \draw[blue] (XY_4) circle (1.2pt);
    \draw[blue] (XY_5) circle (1.2pt);

    \draw[black] (-1,-2) -- (-1, 1.9);
    \draw[black] (-1,-2) -- (2,-2);

    \foreach \p/\pos in {X_6/below left, X_2/below right, X_3/below right, X_4/above left}
      \node[\pos] at (\p) {$\p$};

    \coordinate (P) at (1.2,0.9);
    \coordinate (PX) at (1.2,-2);
    \coordinate (PY) at (-1,0.9);

    % X_1 drawn as a filled diamond, its axis projections as open diamonds (shape-coded for black-and-white legibility)
    \filldraw[red] (P) ++(0,2.2pt) -- ++(2.2pt,-2.2pt) -- ++(-2.2pt,-2.2pt) -- ++(-2.2pt,2.2pt) -- cycle;
    \draw[red] (PX) ++(0,2.2pt) -- ++(2.2pt,-2.2pt) -- ++(-2.2pt,-2.2pt) -- ++(-2.2pt,2.2pt) -- cycle;
    \draw[red] (PY) ++(0,2.2pt) -- ++(2.2pt,-2.2pt) -- ++(-2.2pt,-2.2pt) -- ++(-2.2pt,2.2pt) -- cycle;
    \node[right] at (P) {$X_1$};

    \coordinate (Y_1) at (0.72, 0.18);
    \node[below left] at (Y_1) {\footnotesize $X_0\bm{w}$};
    % projection X_0 w drawn as a filled square
    \filldraw[orange] (Y_1) ++(-1.5pt,-1.5pt) rectangle ++(3pt,3pt);

    \draw[dashed, gray] (P) -- (Y_1);
    \draw[dashed, gray] (P) -- (PX);
    \draw[dashed, gray] (P) -- (PY);

    \filldraw[blue] (X_5) circle (1.2pt);
    \node[below] at (X_5) {$X_5$};

    \end{tikzpicture}
    \caption{}
    \label{fig:projection-b}
    \end{subfigure}
    \tabnotes{Both panels show the same donor predictor vectors, their convex hull, the treated unit, and its projection onto the hull. The open markers on the axes are the one-dimensional projections of the donors and the treated unit on the surviving axis when one of the diagonal entries of $V$ is set to zero. 
    }
\end{figure}
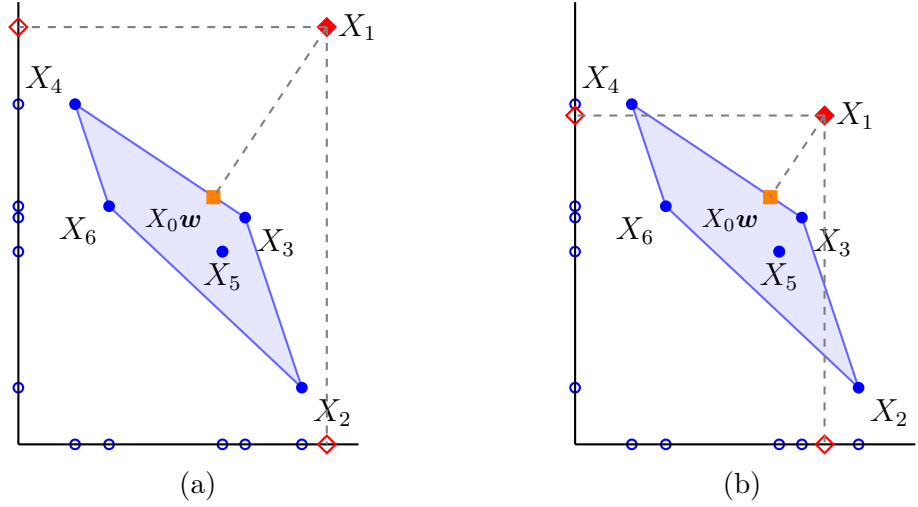

\begin{figure}[t]
    \centering
    \caption{Attainable projections under diagonal predictor weights}
    \label{fig:feasible}
    \begin{subfigure}[b]{0.48\textwidth}
    \centering
    \begin{tikzpicture}[font=\fontfamily{qag}\selectfont,scale=1.5,thick]

    % Draw the polygon
    \coordinate (X_6) at (-0.2,0.1);
    \coordinate (X_2) at (1.5,-1.5);
    \coordinate (X_3) at (1,0);
    \coordinate (X_4) at (-0.5,1);
    \coordinate (X_5) at (0.8,-0.3);

    \filldraw[fill=blue!10, draw=blue!60, thick]
        (X_6) -- (X_2) -- (X_3) -- (X_4) -- cycle;

    \draw[black] (-1,-2) -- (-1, 1.9);
    \draw[black] (-1,-2) -- (2,-2);

    % Highlight the visible facets X_2 X_3 and X_3 X_4
    \draw[purple, ultra thick] (X_2) -- (X_3);
    \draw[purple, ultra thick] (X_3) -- (X_4);

    % Draw and label vertices
    \foreach \p/\pos in {X_6/below left, X_2/below right, X_3/below right, X_4/above left}
    \node[\pos] at (\p) {$\p$};
    \node[below] at (X_5) {$X_5$};

    % Draw an external point (same X_1 as Figure 2(a): on the Euclidean normal through (0.72,0.18))
    \coordinate (P) at (1.72,1.68);
    \node[right] at (P) {$X_1$};

    % Projections of the donors and X_1 onto the two axes
    \coordinate (XX_6) at (-0.2,-2);
    \coordinate (XX_2) at (1.5,-2);
    \coordinate (XX_3) at (1,-2);
    \coordinate (XX_4) at (-0.5,-2);
    \coordinate (XX_5) at (0.8,-2);

    \coordinate (XY_6) at (-1,0.1);
    \coordinate (XY_2) at (-1,-1.5);
    \coordinate (XY_3) at (-1,0);
    \coordinate (XY_4) at (-1,1);
    \coordinate (XY_5) at (-1,-0.3);

    \coordinate (PX) at (1.72,-2);
    \coordinate (PY) at (-1,1.68);

    \draw[dashed, gray] (P) -- (PX);
    \draw[dashed, gray] (P) -- (PY);

    \draw[blue] (XX_6) circle (1.2pt);
    \draw[blue] (XX_2) circle (1.2pt);
    \draw[blue] (XX_3) circle (1.2pt);
    \draw[blue] (XX_4) circle (1.2pt);
    \draw[blue] (XX_5) circle (1.2pt);

    \draw[blue] (XY_6) circle (1.2pt);
    \draw[blue] (XY_2) circle (1.2pt);
    \draw[blue] (XY_3) circle (1.2pt);
    \draw[blue] (XY_4) circle (1.2pt);
    \draw[blue] (XY_5) circle (1.2pt);

    % X_1 axis projections drawn as open diamonds
    \draw[red] (PX) ++(0,2.2pt) -- ++(2.2pt,-2.2pt) -- ++(-2.2pt,-2.2pt) -- ++(-2.2pt,2.2pt) -- cycle;
    \draw[red] (PY) ++(0,2.2pt) -- ++(2.2pt,-2.2pt) -- ++(-2.2pt,-2.2pt) -- ++(-2.2pt,2.2pt) -- cycle;

    \filldraw[blue] (X_6) circle (1.2pt);
    \filldraw[blue] (X_2) circle (1.2pt);
    \filldraw[blue] (X_3) circle (1.2pt);
    \filldraw[blue] (X_4) circle (1.2pt);
    \filldraw[blue] (X_5) circle (1.2pt);
    % X_1 drawn as a filled diamond (shape-coded for black-and-white legibility)
    \filldraw[red] (P) ++(0,2.2pt) -- ++(2.2pt,-2.2pt) -- ++(-2.2pt,-2.2pt) -- ++(-2.2pt,2.2pt) -- cycle;

    \end{tikzpicture}
    \caption{}
    \label{fig:feasible-a}
    \end{subfigure}
    \hspace*{-1cm}
    \begin{subfigure}[b]{0.48\textwidth}
    \centering
    \begin{tikzpicture}[font=\fontfamily{qag}\selectfont,scale=1.5,thick]

    % Draw the polygon
    \coordinate (X_6) at (-0.2,0.1);
    \coordinate (X_2) at (1.5,-1.5);
    \coordinate (X_3) at (1,0);
    \coordinate (X_4) at (-0.5,1);
    \coordinate (X_5) at (0.8,-0.3);
    
    \filldraw[fill=blue!10, draw=blue!60, thick] 
        (X_6) -- (X_2) -- (X_3) -- (X_4) -- cycle;
    
    % Draw and label vertices
    \foreach \p/\pos in {X_6/below left, X_2/below right, X_3/below right, X_4/above left}
    \node[\pos] at (\p) {$\p$};
    \node[below] at (X_5) {$X_5$};
    
    % Draw an external point
    \coordinate (P) at (1.2,0.9);
    \node[right] at (P) {$X_1$};

    \draw[black] (-1,-2) -- (-1, 1.9);
    \draw[black] (-1,-2) -- (2,-2);

    \coordinate (Z_2) at (1.2,-0.6);
    \coordinate (Z_4) at (-0.36,0.9);

    \draw[dashed, gray] (P) -- (Z_2);
    \draw[dashed, gray] (P) -- (Z_4);

    \draw[purple, ultra thick] (X_3) -- (Z_4);
    \draw[purple, ultra thick] (X_3) -- (Z_2);
    \draw[purple, ultra thick] (-0.46,0.9) -- (Z_4);
    \draw[purple, ultra thick] (1.2,-1.23) -- (Z_2);

    % Projections of the donors and X_1 onto the two axes
    \coordinate (XX_6) at (-0.2,-2);
    \coordinate (XX_2) at (1.5,-2);
    \coordinate (XX_3) at (1,-2);
    \coordinate (XX_4) at (-0.5,-2);
    \coordinate (XX_5) at (0.8,-2);

    \coordinate (XY_6) at (-1,0.1);
    \coordinate (XY_2) at (-1,-1.5);
    \coordinate (XY_3) at (-1,0);
    \coordinate (XY_4) at (-1,1);
    \coordinate (XY_5) at (-1,-0.3);

    \coordinate (PX) at (1.2,-2);
    \coordinate (PY) at (-1,0.9);

    \draw[dashed, gray] (P) -- (PX);
    \draw[dashed, gray] (P) -- (PY);

    \draw[blue] (XX_6) circle (1.2pt);
    \draw[blue] (XX_2) circle (1.2pt);
    \draw[blue] (XX_3) circle (1.2pt);
    \draw[blue] (XX_4) circle (1.2pt);
    \draw[blue] (XX_5) circle (1.2pt);

    \draw[blue] (XY_6) circle (1.2pt);
    \draw[blue] (XY_2) circle (1.2pt);
    \draw[blue] (XY_3) circle (1.2pt);
    \draw[blue] (XY_4) circle (1.2pt);
    \draw[blue] (XY_5) circle (1.2pt);

    % X_1 axis projections drawn as open diamonds
    \draw[red] (PX) ++(0,2.2pt) -- ++(2.2pt,-2.2pt) -- ++(-2.2pt,-2.2pt) -- ++(-2.2pt,2.2pt) -- cycle;
    \draw[red] (PY) ++(0,2.2pt) -- ++(2.2pt,-2.2pt) -- ++(-2.2pt,-2.2pt) -- ++(-2.2pt,2.2pt) -- cycle;

    \filldraw[blue] (X_6) circle (1.2pt);
    \filldraw[blue] (X_2) circle (1.2pt);
    \filldraw[blue] (X_3) circle (1.2pt);
    \filldraw[blue] (X_4) circle (1.2pt);
    \filldraw[blue] (X_5) circle (1.2pt);
    % X_1 drawn as a filled diamond (shape-coded for black-and-white legibility)
    \filldraw[red] (P) ++(0,2.2pt) -- ++(2.2pt,-2.2pt) -- ++(-2.2pt,-2.2pt) -- ++(-2.2pt,2.2pt) -- cycle;
    \filldraw[blue] (X_5) circle (1.2pt);

    \end{tikzpicture}
    \caption{}
    \label{fig:feasible-b}
    \end{subfigure}
    \tabnotes{The thick segments indicate the projections of the treated unit attainable under diagonal $V$, for the two configurations of Figure~\ref{fig:projection}. In panel (b), the short thick segments inside the hull are the solutions with one of the diagonal entries of $V$ set to zero.}
\end{figure}
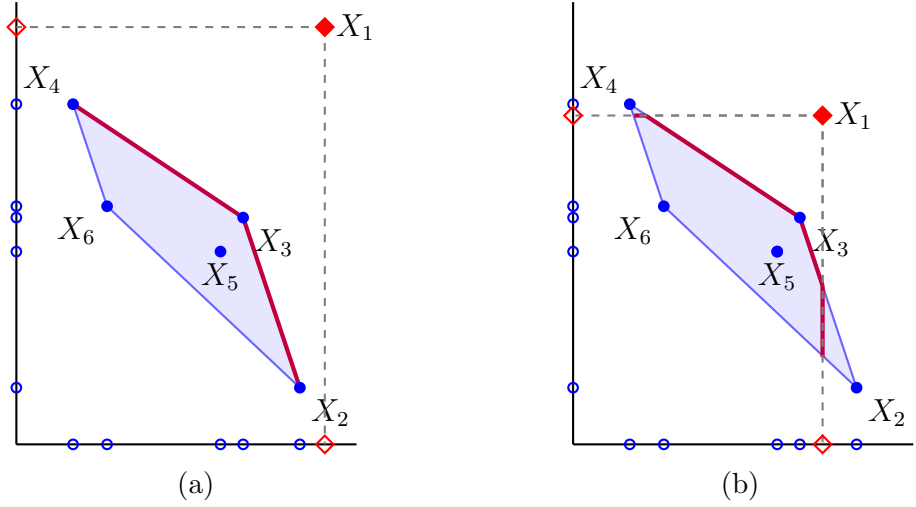

This discontinuity creates a problem for the outer optimization. Zero entries in the diagonal of $V$ can expand the set of synthetic controls that minimize the inner objective. The outer optimization selects among them the one with the smallest validation error, which need not be the projection of the treated unit onto the donor hull under a positive definite $V$ or a limit of such projections. Moreover, the resulting synthetic control ignores predictors with zero weight entirely.

%------------------------------------------------------------------
\section{A Geometric Reformulation}\label{sec:geometric}
%------------------------------------------------------------------

This section recasts the predictor weight problem, shifting the focus to the set of donor weights attainable under a class of $V$ matrices. Section~\ref{sec:attainable} defines this attainable set and expresses the resulting optimization problem as regularized selection of  donor weights. Section~\ref{sec:characterization} develops a general criterion for membership in the attainable set, drawing on the convex-geometric notions of visibility and normal cones. Sections~\ref{sec:diag_application}-\ref{sec:spd_application} apply the criterion to the two canonical candidate sets: diagonal predictor weights and general symmetric positive definite predictor weights. Section~\ref{sec:kkt} shows how to recover a rationalizing $V$ for any attainable set of donor weights.

\subsection{Single-Level Optimization Over the Attainable Set of Donor Weights}\label{sec:attainable}

Fix a candidate set $\mathcal{V}$ of predictor weight matrices (e.g.,~diagonal or positive definite matrices) and define
\begin{equation}\label{eq:W_full}
    \mathcal{W}(\mathcal{V}) = \bigcup_{V \in \mathcal{V}} \mathcal{W}(V),
\end{equation}
which is the set of donor weights that arise as inner-optimizers for some $V \in \mathcal{V}$. We call $\mathcal{W}(\mathcal{V})$ the \emph{attainable set} of donor weights under $\mathcal{V}$. The bi-level problem~\eqref{eq:outer} is then equivalent to the single-level optimization
\begin{equation}\label{eq:single_level}
    \min_{\bm{w} \in \mathcal{W}(\mathcal{V})} \norm{Z_1 - Z_0\bm{w}}^2.
\end{equation}
This equivalence helps interpret and address problems with current optimizers for synthetic controls. The outer search over $V$ encodes a restriction of the donor weight vector to the subset $\mathcal{W}(\mathcal{V}) \subseteq \Delta^J$.

Minimization of the validation loss $\norm{Z_1 - Z_0\bm{w}}^2$ over the full simplex $\Delta^J$ risks overfitting the pre-treatment validation outcomes. The estimator in \eqref{eq:single_level} addresses the risk of overfitting by constraining $\bm{w}$ to lie in $\mathcal{W}(\mathcal{V})$, the smaller subset of donor weights that can be rationalized by some candidate predictor weight matrix. Differences across candidate sets $\mathcal{V}$ translate into differences in the geometry of the regularized set $\mathcal{W}(\mathcal{V})$. In this sense, $\mathcal V$ specifies a form of regularization.

We focus on three natural candidate sets in what follows. The set of diagonal matrices $\mathcal{V}_\Delta$
includes boundary points where one or more diagonal entries vanish. This is the closed class over which the standard bi-level procedure of Section~\ref{sec:predictor_weights} optimizes. The open analog,
\begin{equation}\label{eq:V_plus}
    \mathcal{V}_+ = \{\Diag(\bm{d}) : \bm{d} \in \Delta^K_+\}, \qquad \Delta^K_+ = \{\bm{d} \in \mathbb{R}^K : d_k > 0, \, \textstyle\sum_k d_k = 1\},
\end{equation}
restricts to strictly positive diagonal entries. Every $V \in \mathcal{V}_+$ is positive definite, so under Assumption~\ref{ass:gp} the inner problem has a unique solution. $\mathcal{V}_\Delta$ is the closure of $\mathcal{V}_+$, that is, $\mathcal{V}_\Delta = \overline{\mathcal{V}_+}$. Finally, the set of symmetric positive definite matrices
\begin{equation}\label{eq:V_spd}
    \mathbb{S}^K_{++} = \{V \in \mathbb{R}^{K \times K} : V = V', \; V \succ 0\}
\end{equation}
extends the candidate set from diagonal rescalings to arbitrary full-rank linear transformations of the predictor space.

With $\mathcal V = \mathcal V_\Delta$, the bi-level program~\eqref{eq:outer} selects, among the donor weights attainable for some $V$ in the class, those that best fit the validation outcomes. The correspondence $V \mapsto \mathcal{W}(V)$ can jump and become set-valued at the boundary of $\mathcal{V}_\Delta$, where a vanishing diagonal entry deletes a predictor and enlarges the set of donor weights available to the validation criterion. The search over $\mathcal{V}_\Delta$ therefore rewards rank deficiency and can return donor weights that match the treated unit on a single predictor and ignore the others \citep{Malo2023bilevel}.

A minimum over the open class $\mathcal{V}_+$ need not exist. The bi-level procedure closes the class of metrics, optimizing over $\mathcal V_\Delta$, at the cost of the boundary problem. We instead compute
\begin{equation}\label{eq:W0}
    \overline{\mathcal{W}(\mathcal{V}_+)} = \overline{\bigcup_{V \in \mathcal{V}_+} \mathcal{W}(V)},
\end{equation}
the closure of the attainable sets under strictly positive diagonal matrices, and solve the single-level problem
\begin{equation}\label{eq:single_level_W0}
    \min_{\bm{w} \in \overline{\mathcal{W}(\mathcal{V}_+)}} \norm{Z_1 - Z_0\bm{w}}^2.
\end{equation}
A solution to \eqref{eq:single_level_W0} that ignores a predictor is a limit of inner solutions under a sequence of positive diagonal matrices $V$ whose weight on that predictor converges to zero.

The standard bi-level procedure searches over the closed class $\overline{\mathcal{V}_+}$ in the outer problem, so it selects the donor weights from $\mathcal{W}(\overline{\mathcal{V}_+}) = \bigcup_{V \in \overline{\mathcal{V}_+}} \mathcal{W}(V)$. For diagonal $V$, the procedure proposed in this article instead selects the donor weights from $\overline{\mathcal{W}(\mathcal{V}_+)}$, the closure of $\mathcal{W}(\mathcal{V}_+)$. The closure of the metrics and the closure of the weights agree when $V \mapsto \mathcal{W}(V)$ is continuous, but the boundary problem of Section~\ref{sec:boundary} is a failure of continuity. For a set $\mathcal{W}$ of donor weights, let $X_0\mathcal{W} = \{X_0 \bm{w} : \bm{w} \in \mathcal{W}\}$. The set $X_0\overline{\mathcal{W}(\mathcal{V}_+)}$ retains only the limits of projections under strictly positive diagonal $V$. In contrast, $X_0\mathcal{W}(\overline{\mathcal{V}_+})$ may include synthetic controls that are not such limits, as in Figure~\ref{fig:feasible-b}.

 The next proposition shows that the closed class $\mathcal{V}_\Delta$ admits extra donor weights relative to $\overline{\mathcal{W}(\mathcal{V}_{+})}$. The additional weights come from the matrices with a single positive diagonal entry, $\mathcal{V}_{1} = \{\Diag(\bm{d}) : \bm{d} \in \Delta^K,\; \|\bm{d}\|_0 = 1\}$. The increment, $X_0\mathcal{W}(\mathcal{V}_1)$, is the set of synthetic
controls that either match $X_1$ exactly on a single predictor $k$, leaving
the remaining predictors unrestricted, or, when $X_{1k}$ falls outside the
range of the donor values of predictor $k$, assign all weight to the donor
with the most extreme value of that predictor on the side of $X_{1k}$.

\begin{proposition}\label{prop:general_diag}
$X_0\mathcal{W}(\mathcal{V}_{\Delta}) = X_{0}\overline{\mathcal{W}(\mathcal{V}_{+})} \cup X_{0}\mathcal{W}(\mathcal{V}_{1})$.
\end{proposition}
All proofs are in the appendix.

A formulation analogous to \eqref{eq:single_level_W0} for general positive definite matrices, $V$, is
\begin{equation}\label{eq:single_level_Wplus}
    \min_{\bm{w} \in \mathcal{W}(\mathbb{S}^K_{++})} \norm{Z_1 - Z_0\bm{w}}^2, 
\end{equation}
where 
\begin{equation}\label{eq:W_plus}
    \mathcal{W}(\mathbb{S}^K_{++}) = \bigcup_{V \in \mathbb{S}^K_{++}} \mathcal{W}(V). 
\end{equation}
Proposition ~\ref{prop:visible_facets} in Section \ref{sec:spd_application} shows that $\mathcal{W}(\mathbb{S}^K_{++})$ is a closed set.

Because $\mathcal{V}_+ \subset \mathbb{S}^K_{++}$, the inclusion $\overline{\mathcal{W}(\mathcal{V}_+)} \subseteq \mathcal{W}(\mathbb{S}^K_{++})$ holds, and the general positive definite formulation can only enlarge the regularized set of attainable donor weights. The remainder of Section~\ref{sec:geometric} characterizes both $\overline{\mathcal{W}(\mathcal{V}_+)}$ and $\mathcal{W}(\mathbb{S}^K_{++})$ in terms of the facial structure of the donor convex hull.

\subsection{Characterization of the Attainable Set}\label{sec:characterization}

A donor weight vector $\bm{w}$ belongs to the attainable set if and only if $X_0\bm{w}$ can be obtained as a $V$-weighted projection of $X_1$ on $\mathcal P =\conv(\{X_2, \ldots, X_{J+1}\})$ for some $V \in \mathcal{V}$. 
Geometric visibility restricts which portions of the boundary of the donor hull can receive a projection of the treated point $X_1$ under a positive definite matrix $V$.

\begin{definition}[Visibility]\label{def:visibility}
    For a convex set $\mathcal{P} \subset \mathbb{R}^K$ and a point $X_1 \notin \mathcal{P}$, a point $Y \in \mathcal{P}$ is \emph{visible from $X_1$} if the open line segment connecting $X_1$ and $Y$ does not intersect $\mathcal{P}$:
    \[
        \left\{ \lambda X_1 + (1-\lambda) Y : \lambda \in (0,1) \right\} \cap \mathcal{P} = \varnothing.
    \]
    A face $F$ of a polytope $\mathcal{P}$ is \emph{visible from $X_1$} if every point in $F$ is visible from $X_1$.
\end{definition}

A facet $F$ is visible from $X_1$ when $X_1$ belongs to the outer side of the supporting hyperplane of $F$.

\begin{lemma}\label{lem:visibility_facet}
    Let $F$ be a facet of $\mathcal{P}$ contained in the hyperplane $\{Y : \bm{a}'Y + b = 0\}$, where $\bm{a}$ is an outward normal vector of $\mathcal{P}$ at $F$. Then $F$ is visible from $X_1$ if and only if $\bm{a}'X_1 + b > 0$.
\end{lemma}

Visibility restricts where attainable projections can land. A complete characterization of attainability requires the notion of normal cone, which describes projections onto convex sets.

\begin{definition}[Normal Cone]\label{def:normal_cone}
    For a closed convex set $\mathcal{C} \subset \mathbb{R}^K$ and a point $Y \in  \mathcal{C}$, the \emph{normal cone} of $\mathcal{C}$ at $Y$ is
    \[
        N_{\mathcal{C}}(Y) = \left\{ \bm{v} \in \mathbb{R}^K : \langle \bm{v}, Z - Y \rangle \leq 0 \text{ for all } Z \in \mathcal{C} \right\}.
    \]
\end{definition}

For a full-dimensional polytope, the normal cone at any point is generated by the outward
normal vectors of the facets whose supporting hyperplanes pass through that
point. Concretely, let
\[
    \mathcal{P} = \{Y : \bm{a}_i'Y + b_i \le 0,\ i = 1,\dots,m\},
\]
where $\bm{a}_i$ is an
outward-pointing normal. For $Y \in \mathcal{P}$, let
\(
    I(Y) = \{i : \bm{a}_i'Y + b_i = 0\}
\)
denote the set of constraints active at $Y$. Then
\begin{equation}\label{eq:normal_cone_face}
    N_{\mathcal{P}}(Y) = \cone\!\left(\{\bm{a}_i\}_{i \in I(Y)}\right)
    = \Big\{ \textstyle\sum_{i \in I(Y)} \lambda_i \bm{a}_i : \lambda_i \ge 0 \Big\}.
\end{equation}
In particular, $N_{\mathcal{P}}(Y) = \{0\}$ when $Y \in \mathrm{Int}(\mathcal{P})$, and $N_{\mathcal{P}}(Y)$ is constant for all $Y$ on the relative interior of a face $F$.

The following standard projection criterion links the projection map to the normal cone \citep[Proposition 6.47]{BauschkeCombettes2017}.

\begin{lemma}\label{lem:projection_normal}
    Let $\pi_{\mathcal{C}}(X)$ denote the Euclidean projection of $X$ onto a closed convex set $\mathcal{C}$. Then $Y = \pi_{\mathcal{C}}(X)$ if and only if $X - Y \in N_{\mathcal{C}}(Y)$.
\end{lemma}

A full-rank matrix $T$ of $\mathbb{R}^K$ transforms normal cones as follows. Every point in $T\mathcal{P}$ is $TZ$ for a unique $Z \in \mathcal{P}$, so $\bm{v} \in N_{T\mathcal{P}}(TY)$ if and only if $\inner{\bm{v}}{TZ - TY} = \inner{T'\bm{v}}{Z - Y} \leq 0$ for all $Z \in \mathcal{P}$, that is, if and only if $T'\bm{v} \in N_{\mathcal{P}}(Y)$. Therefore,
\begin{equation}\label{eq:normal_cone_transform}
    N_{T\mathcal{P}}(TY) = (T')^{-1} N_{\mathcal{P}}(Y).
\end{equation}
Let $V^{1/2}$ be the positive definite square root of $V$. A point $Y \in \mathcal{P}$ is the $V$-weighted projection of $X_1$ onto $\mathcal{P}$ (that is, $Y = X_0\bm{w}(V)$) if and only if $V^{1/2}Y$ is the Euclidean projection of $V^{1/2}X_1$ onto $V^{1/2}\mathcal{P}$. By Lemma~\ref{lem:projection_normal} and~\eqref{eq:normal_cone_transform},  this holds if and only if $V^{1/2}(X_1 - Y) \in V^{-1/2}N_{\mathcal{P}}(Y)$, or equivalently, if and only if
\begin{equation}\label{eq:master}
    V(X_1 - Y) \in N_{\mathcal{P}}(Y).
\end{equation}

The next lemma shows that, for any positive definite matrix $V$, the projection of $X_1$ onto $\mathcal{P}$ with respect to $\norm{\cdot}_{V}$ lands on a facet that is visible from $X_1$.

\begin{lemma}\label{lem:visible_projection}
    For any $V \in \mathbb{S}^K_{++}$, the point $X_0\bm{w}(V)$ lies on a facet of $\mathcal{P}$ that is visible from $X_1$.
\end{lemma}

Lemma~\ref{lem:visible_projection} restricts attainable projections to visible facets. Every vertex of a visible facet is visible from $X_1$, which allows us to identify the visible donors and discard the rest, as in \citet{BECKER20181}. In Lemma~\ref{lem:visible_projection}, however, visibility is only a necessary condition for attainability. The next two subsections determine which points of the visible facets are attainable under diagonal and under general positive definite candidate sets.

\subsection{Application to Positive Diagonal $\boldsymbol{V}$}\label{sec:diag_application}

We next show that, for the diagonal class $V = \Diag(\bm{d})$ with $\bm{d} > 0$, \eqref{eq:master} yields a sign-compatibility test for attainability. Let $F$ be a visible face of $\mathcal{P}$, let $Y \in \relint(F)$, and $I = I(Y)$ be the set of facets containing $F$. By~\eqref{eq:master} and~\eqref{eq:normal_cone_face}, $Y = X_0\bm{w}(V)$ if and only if
\begin{equation}\label{eq:diagonal_criterion}
    \Diag(\bm{d})(X_1 - Y) \in \cone\left(\{\bm{a}_i\}_{i \in I}\right).
\end{equation}
Each $d_k > 0$ rescales the $k$-th component of $X_1 - Y$ and preserves its sign. Therefore, there exists a positive diagonal $V$ that attains $Y$ if and only if the normal cone contains a vector with the same signs as $X_1 - Y$. This amounts to finding non-negative coefficients $\{\alpha_i\}_{i \in I}$ and positive scalars $\{d_k\}_{k=1}^K$ such that
\begin{equation}\label{eq:diag_lp}
    d_k(X_{1k} - Y_k) = \sum_{i \in I} \alpha_i a_{ik}, \qquad k = 1, \ldots, K.
\end{equation}

Because the signs of $X_1 - Y$ can change as $Y$ moves within $\relint(F)$, a face may contain attainable projections on only part of its relative interior. For a facet, the signs of $X_1 - Y$ must match those of its normal. For a face of lower dimension, the normal cone may contain vectors with several sign patterns, and each defines a separate attainable region. The donor weights of these regions form finitely many pieces of $\mathcal{W}(\mathcal{V}_+)$, so the closure in~\eqref{eq:W0} equals the union of the closures of the pieces.

Minimization of validation loss over the closure of each piece is a convex quadratic program. The global minimum of~\eqref{eq:single_level_W0} is the smallest of these piece minima. 

\subsection{Application to General Positive Definite $\boldsymbol{V}$}\label{sec:spd_application}

So far, we have restricted attention to diagonal weighting matrices $V \in \mathcal{V}_+$, which rescale the predictors but leave the coordinate axes fixed. A natural extension is to allow arbitrary positive definite weights $V \in \mathbb{S}^K_{++}$. General positive definite $V$ allows the inner optimization to match linear combinations of predictors, not just individual predictors along their original axes. The attainable set under $\mathbb{S}^K_{++}$ has a simple geometric characterization.

\begin{proposition}\label{prop:visible_facets}
    Under Assumption~\ref{ass:gp}, $\mathcal{W}(\mathbb{S}^K_{++})$ is closed and equal to the set of $\bm{w} \in \Delta^J$ such that $X_0\bm{w}$ lies on a facet of $\mathcal{P}$ visible from $X_1$. 
\end{proposition}
Recall that, under diagonal $V$, every projection lies on a visible facet, but not every point of a visible facet is attainable. Proposition~\ref{prop:visible_facets} says that, when $V$ ranges over $\mathbb{S}^K_{++}$, every point on a visible facet is an attainable projection. 

The global minimum of~\eqref{eq:single_level_Wplus} is therefore the smallest of the minima of the validation loss over the visible facets. 

Because the union of the visible facets is closed, $\overline{\mathcal{W}(\mathcal{V}_+)} \subseteq \mathcal{W}(\mathbb{S}^K_{++})$. The enlarged attainable set allows a closer match to the validation outcomes but increases the risk of overfitting.

\subsection{Recovering $\boldsymbol{V}$ from Optimal Weights}\label{sec:kkt}

The procedures of Section~\ref{sec:diag_application} and \ref{sec:spd_application} identify optimal donor weights $\bm{w}^*$ over classes of predictor weight matrices without choosing a $V$ in the class. Applied researchers may also wish to report $V$ alongside the donor weights. We say that a matrix $V$  rationalizes $\bm{w}^*$ if $\bm{w}^* = \bm{w}(V)$. This subsection describes the Karush-Kuhn-Tucker (KKT) constraints any rationalizing $V$ must satisfy, establishes the convexity of the resulting feasible set, and discusses selection criteria.

\subsubsection{KKT Constraints for the Inner Optimization}\label{sec:kkt_constraints}

The inner optimization~\eqref{eq:inner} is a convex quadratic program over the simplex $\Delta^J$. Its Karush-Kuhn-Tucker conditions provide necessary and sufficient conditions for optimality. For $\bm{w}^* \in \Delta^J$ to minimize $\norm{X_1 - X_0\bm{w}}_V^2$ over $\Delta^J$, there must exist $\bm{\lambda} \in \mathbb{R}^J_+$ and $\mu \in \mathbb{R}$ satisfying
\begin{align}
    X_0' V (X_0\bm{w}^* - X_1) - \bm{\lambda} + \mu \mathbf{1} &= 0, \label{eq:kkt_stationarity}\\
    \bm{\lambda} \odot \bm{w}^* &= 0, \label{eq:kkt_complementarity}\\
    \bm{\lambda} &\geq 0, \label{eq:kkt_feasibility}
\end{align}
where $\odot$ denotes the Hadamard product. Divide the donors into the \emph{active} donors in the support $\supp(\bm{w}^*) = \{j : w_j^* > 0\}$ and the \emph{inactive} donors outside it,
and fix a reference active donor $j_0 \in \supp(\bm{w}^*)$. For any other active donor, $j \in \supp(\bm{w}^*) \setminus \{j_0\}$, we have
\begin{equation}\label{eq:kkt_active_diff}
    (X_j - X_{j_0})' V (X_0\bm{w}^* - X_1) = 0.
\end{equation}
Equation \eqref{eq:kkt_active_diff} imposes $|\supp(\bm{w}^*)| - 1$ independent constraints. (The functionals $V \mapsto (X_j - X_{j_0})'V(X_0\bm{w}^* - X_1)$ are linearly independent because the vectors $\{X_j - X_{j_0}\}_{j \in \supp(\bm{w}^*)\setminus\{j_0\}}$ are linearly independent by Assumption~\ref{ass:gp}(ii).) For each inactive donor $j \notin \supp(\bm{w}^*)$, we have
\begin{equation}\label{eq:kkt_inactive_reduced}
    (X_j - X_{j_0})' V (X_0\bm{w}^* - X_1) \geq 0.
\end{equation}

Each constraint in \eqref{eq:kkt_active_diff} and \eqref{eq:kkt_inactive_reduced} is linear in the entries of $V$. The set of positive semidefinite matrices, $\mathbb{S}^K_+$, is closed and convex. Therefore, the feasible set
\begin{equation}\label{eq:V_feasible}
    \mathcal{F}(\bm{w}^*) = \left\{ V \in \mathbb{S}^K_{+} : \eqref{eq:kkt_active_diff} \text{ and } \eqref{eq:kkt_inactive_reduced} \text{ hold} \right\}
\end{equation}
is a convex set. As a result, any convex objective can be minimized over $\mathcal{F}(\bm{w}^*)$ by standard optimization methods. For every $V$ in $\mathcal{F}(\bm{w}^*)$, $\bm w^*$ belongs to the set of solutions of the inner problem. Under Assumption~\ref{ass:gp}, $\bm w^*$ is the unique solution for every positive definite member of $\mathcal{F}(\bm{w}^*)$, so $\mathcal{F}(\bm{w}^*)\cap\mathbb{S}^K_{++}$ is the set of matrices that rationalize $\bm w^*$.

\subsubsection{Selection Criteria}\label{sec:selection}

The feasible set $\mathcal{F}(\bm{w}^*)$ generically contains infinitely many
positive definite matrices, all of which rationalize the same donor weights. After trace
normalization, the equalities~\eqref{eq:kkt_active_diff} leave an identified
set of dimension $K(K+1)/2 - |\supp(\bm{w}^*)|$ in the full positive
definite class, and $K - |\supp(\bm{w}^*)|$ in the diagonal class. As a result, a unique
normalized rationalizer exists generically only for diagonal $V$ when the
donor weights are supported on every vertex of a facet.

The selection of a particular $V$ requires an additional criterion that expresses a preference over the set of predictor weightings. When the goal is to save
the predictor weights for future use, as in \citet{Abadie2015Germany}, it is
usually preferable to choose a $V$ that gives positive weight to every predictor. We consider two criteria: log-determinant maximization and the analytic center.

\paragraph{Log-determinant maximization.}
The first criterion maximizes the log-determinant of $V$:
\begin{equation}\label{eq:V_logdet}
    \max_{V \in \mathcal{F}(\bm{w}^*)} \log\det(V) \quad \text{s.t.} \quad \tr(V) = 1.
\end{equation}

The function $\log\det$ is concave on $\mathbb{S}^K_{++}$ \citep[Section 3.1.5]{BoydVandenberghe2004}, so the maximization in~\eqref{eq:V_logdet} is a convex problem. As any eigenvalue approaches zero, $\log\det(V) \to -\infty$, so the maximizer keeps $V$ positive definite whenever feasible. The unconstrained maximizer is $V = I/K$.

\begin{proposition}\label{prop:dense_V}
If $\bm{w}^* \in \mathcal{W}(\mathbb{S}^K_{++})$, then the log-determinant problem~\eqref{eq:V_logdet} has a unique and positive definite solution.
\end{proposition}

\paragraph{Analytic center.} The second criterion selects the \emph{analytic center} of $\mathcal{F}(\bm{w}^*)$, the matrix that maximizes the sum of the logarithmic barriers of the inequality constraints that define $\mathcal{F}(\bm{w}^*)$, subject to the trace normalization. 

Let $s_{j}(V) = (X_j - X_{j_0})'V(X_0\bm w^* - X_1)$ be the slack of the inequality \eqref{eq:kkt_inactive_reduced}, $j \notin \supp(\bm w^*)$. The analytic center is pushed away from every boundary of $\mathcal{F}(\bm{w}^*)$, including $s_{j}(V) = 0$ for $j \notin \supp(\bm w^*)$, where a small perturbation of $V$ makes $\bm w^*$ suboptimal. The next proposition establishes existence and uniqueness of the maximizer.

\begin{proposition}
\label{prop:extcenter}
Under Assumption \ref{ass:gp}, for every $\bm{w}^* \in \mathcal{W}(\mathbb{S}^K_{++})$, the program
\[
V_c = \argmax \left\{ \log\det V + \sum_{j \notin \supp(\bm{w}^*)} \log s_j(V) \ : \ V \in \mathcal{F}(\bm{w}^*), \ \tr (V) = 1 \right\}
\]
has a unique solution $V_c \succ 0$ satisfying
$s_j(V_c)>0$ for every $j \notin \supp(\bm{w}^*)$.
\end{proposition}

For $\bm w^* \in \mathcal W(\mathbb{S}^K_{++})$, a positive definite rationalizer $V$ exists, and the log-determinant criterion and the analytic center each select one. We recommend reporting the analytic center or the log-determinant maximizer $V^*$ as the default rationalizing matrix, summarized by its trace-normalized diagonal $V^*_{kk}/\tr(V^*)$. For diagonal $V^*$, the share $V^*_{kk}/\tr(V^*)$ equals $d_k$. For general $V^*$, the diagonal entry $V^*_{kk}$ is the penalty for a discrepancy of one unit in predictor $k$ when the other predictors are exactly matched.

For $\bm w^* \in \overline{\mathcal W(\mathcal V_+)}$, a strictly positive diagonal rationalizer may fail to exist, because $\bm{w}^*$ may be attainable only as a limit of solutions under strictly positive diagonal $V$. In that case, sparsity in the diagonal of $V$ is forced by the geometry, in contrast with the incidental sparsity induced by the optimizer, as in Figure~\ref{fig:feasible-b}.

When sparsity is forced, we apply the log-determinant and analytic center criteria over the retained predictors only. If the retained predictors are matched exactly, every diagonal matrix supported on them satisfies the KKT conditions, both criteria reduce to maximizing the sum of $\log d_k$ over the retained predictors, and the solution gives equal weight to every retained predictor.

Alternatively, a strictly positive diagonal matrix can be selected at an arbitrarily small cost in validation fit. Because $\bm{w}^*$ is a limit of solutions under strictly positive diagonal $V$, for any $\varepsilon > 0$ the predictor weights that maximize $\sum_{k} \log d_k$ subject to a validation loss within $\varepsilon$ of its minimum are strictly positive. A floor $d_k \geq \delta$ for a small $\delta > 0$ has the same effect, and it also removes the boundary discontinuity of Section~\ref{sec:boundary}.

%------------------------------------------------------------------
\section{The Interior Case}\label{sec:interior_results}
%------------------------------------------------------------------

The preceding sections assume that $X_1$ lies outside the donor hull
$\mathcal P$. When $X_1$ lies in the interior of the donor hull, multiple sets of donor weights may achieve an exact match in the predictors. For every $V\succ0$, the inner solution set
\begin{equation}
\mathcal{W}(V)=\cF
:=\{\bm w\in\Delta^J:X_0\bm w=X_1\}
\label{eq:int_exact}
\end{equation}
does not depend on $V$, and the bi-level problem~\eqref{eq:outer} reduces to minimizing validation loss over $\cF$. 
As in \citet{abadie2021penalized}, we add a vanishing penalty
on the pairwise matching discrepancies between $X_1$ and the donor predictor values, $X_j$. 
As in the exterior case, the outer optimization reduces to a search over a finite set of candidate simplices.

\subsection{Selection by a Vanishing Penalty}
\label{sec:int_penalty}

The selection penalty in \citet{abadie2021penalized} is the weighted average
of the squared distances between the treated unit and its contributing
donors. Among exact matches, this penalty favors combinations of donors
that are individually close to the treated unit. For the interior case, as the penalty parameter
goes to zero, the estimator selects the exact match that
minimizes the penalty.

Our formulation uses the same metric $V$ to measure both the aggregate
matching discrepancy and the pairwise penalty:
\begin{equation}
\bm w(\lambda,V)\in\argmin_{\bm w\in\Delta^J}
\left\{\norm{X_1-X_0\bm w}_V^2+\lambda g_V(\bm w)\right\},
\quad
 g_V(\bm w)=\sum_{j=2}^{J+1}w_j\norm{X_1-X_j}_V^2,
\label{eq:penalized}
\end{equation}
where $\lambda>0$. At $V=I/K$, the objective is proportional to that of
\citet{abadie2021penalized} and yields the same minimizers. For any fixed
$V\succ0$, a limit of $\bm w(\lambda,V)$ as
$\lambda\downarrow 0$ solves
\begin{equation}
\min_{\bm w\in\cF} g_V(\bm w).
\label{eq:int_limit}
\end{equation}
Under the general position conditions of Section~\ref{sec:int_cells}, the minimizer of~\eqref{eq:int_limit} is unique for all $V$ outside a measure-zero subset of $\mathbb{S}^K_{++}$ (Proposition~\ref{prop:coverage}), and $\bm w(\lambda,V)$ then converges to it as $\lambda\downarrow0$. We write $\bm w(0^+,V)$ for the limit. 

In this limit, the set of exact matches does not depend on $V$, but the penalty gives predictor weights a role in choosing among exact matches.

To characterize this choice, let $\ell_V(x)=x'Vx$. Because $X_0\bm w=X_1$ for $\bm w\in\cF$, the expansion of the squared distances $\norm{X_1-X_j}_V^2$ in $g_V$ gives
\begin{equation}
g_V(\bm w)=\sum_{j=2}^{J+1}w_j\ell_V(X_j)-\ell_V(X_1),
\label{eq:int_height}
\end{equation}
 so~\eqref{eq:int_limit} is a linear program in
the donor weights. To interpret this program geometrically, lift each
donor $X_j$ to the point 
$(X_j,\ell_V(X_j))\in\R^{K+1}$. Define the \emph{lifted donor hull} as the convex hull of the lifted points. For any $\bm w\in \mathcal F$,  $(X_1,\ \sum_j w_j\ell_V(X_j))$ lies in the lifted donor hull directly above $X_1$ at height $\sum_j w_j\ell_V(X_j)$. The penalty $g_V(\bm w)$ is
the vertical distance between this point and the lifted treated point
$(X_1,\ell_V(X_1))$. 
Minimization of $g_V(\bm{w})$ over $\mathcal F$ amounts to finding the lowest point of the lifted donor hull directly above $X_1$ \citep{rajan1994optimality, abadie2021penalized}. As $V$ varies, the heights $\ell_V(X_j)$ change, which may change the facet of the lifted donor hull directly above $X_1$ and the selected donors for the synthetic control.  

\subsection{Candidate Simplices and Selecting Metrics}
\label{sec:int_cells}

This section applies the lifting procedure of Section~\ref{sec:int_penalty} to characterize the synthetic control weights in the interior case as a function of $V$. The candidate donor sets are the sets of $K+1$ donors whose simplex contains $X_1$. Each such simplex yields one weight vector equal to the barycentric coordinates of $X_1$ in the simplex. 

Throughout this section, we replace Assumption~\ref{ass:gp} by the following
conditions, with $\cD=\{2,\ldots,J+1\}$ denoting the donor pool.
\begin{assumption}[Interior general position]
\label{ass:gp_int}
(i) $X_1\in\mathcal P$;
(ii) $X_1\notin\aff(\{X_j:j\in B\})$ for every $B\subseteq\cD$ with
$|B|\leq K$;
(iii) every $K+1$ donor predictor vectors are affinely independent;
(iv) for every $K+2$ donors, their lifts $(X_j,\ell_V(X_j))$ are affinely
independent for at least one positive definite $V$ and for at least one
strictly positive diagonal $V$.
\end{assumption}
Together, parts (i)-(ii) place $X_1$ in the strict interior of the donor
hull and require every exact match to use at least $K+1$ donors, so $J\geq K+1$. Part (iii)
ensures that each simplex is nondegenerate. Part (iv) rules out
affine dependencies among lifted donors that persist across all metrics
in either class. 
When $X_1$ belongs to the convex hull of the donor predictors, parts (ii)-(iv) hold for almost every
predictor configuration. 

Under Assumption~\ref{ass:gp_int}, for every $V$ outside a measure-zero subset of $\mathbb{S}^K_{++}$, the solution of~\eqref{eq:int_limit} is unique and assigns positive weight to exactly $K+1$ donors (Proposition~\ref{prop:coverage}). It is therefore enough to examine the subsets $A\subseteq\cD$ of $K+1$ donors whose simplex $\sigma_A=\conv\{X_i:i\in A\}$ contains $X_1$. We call such a subset $A$ an \emph{admissible cell}, because it allows exact interpolation of $X_1$. Each admissible cell $A$ determines a unique candidate weight vector $\bm w_A$, which does not depend on $V$.

\begin{figure}[t]
    \centering
    \caption{Lifted donors and the selection gap ($K=2$)}
    \label{fig:lift}
        \begin{tikzpicture}[scale=0.9, thick]
    \filldraw[fill=gray!8, draw=gray!50] (-5.200,-1.260) -- (2.400,-1.260) -- (5.500,0.910) -- (-2.100,0.910) -- cycle;
    \draw[gray!70, ->] (-5.200,-1.260) -- (-5.200,7.540) node[above] {\footnotesize $\ell_V$};
    \filldraw[fill=blue!10, draw=blue!60] (-2.600,-0.420) -- (0.900,-0.490) -- (-0.125,0.647) -- cycle;
    \draw[gray!70, semithick, dashed] plot[smooth] coordinates {(2.166,7.177) (1.986,7.084) (1.550,6.990) (1.114,7.084) (0.933,7.177)};
    \draw[gray!70, semithick] plot[smooth] coordinates {(0.933,7.177) (0.794,7.268) (0.677,7.357) (0.573,7.444) (0.479,7.530) (0.392,7.614) (0.311,7.696) (0.235,7.776) (0.163,7.854) (0.094,7.931) (0.028,8.005) (-0.035,8.078) (-0.096,8.148) (-0.155,8.217) (-0.213,8.283) (-0.269,8.348) (-0.323,8.411) (-0.376,8.471) (-0.427,8.529) (-0.478,8.586) (-0.527,8.640) (-0.576,8.692) (-0.623,8.741) (-0.670,8.789) (-0.715,8.834) (-0.760,8.877) (-0.805,8.918) (-0.848,8.957) (-0.891,8.993) (-0.933,9.028) (-0.975,9.059) (-1.016,9.089) (-1.056,9.116) (-1.096,9.141) (-1.136,9.164) (-1.175,9.184) (-1.213,9.202) (-1.251,9.217) (-1.289,9.230) (-1.326,9.241) (-1.363,9.250) (-1.400,9.256) (-1.436,9.259) (-1.472,9.260) (-1.507,9.259) (-1.542,9.255) (-1.577,9.249) (-1.611,9.241) (-1.646,9.230) (-1.680,9.216) (-1.713,9.200) (-1.747,9.182) (-1.780,9.161) (-1.813,9.138) (-1.846,9.112) (-1.878,9.084) (-1.911,9.054) (-1.943,9.021) (-1.975,8.985) (-2.007,8.947) (-2.039,8.907) (-2.071,8.864) (-2.102,8.819) (-2.134,8.771) (-2.166,8.721) (-2.197,8.668) (-2.229,8.613) (-2.261,8.555) (-2.292,8.495) (-2.324,8.432) (-2.357,8.366) (-2.389,8.299) (-2.422,8.228) (-2.455,8.155) (-2.488,8.079) (-2.522,8.001) (-2.557,7.920) (-2.592,7.836) (-2.629,7.749) (-2.667,7.659) (-2.706,7.566) (-2.747,7.469) (-2.790,7.368) (-2.836,7.263) (-2.887,7.153) (-2.943,7.037) (-3.009,6.912) (-3.095,6.771) (-3.300,6.545) (-3.505,6.485)};
    \draw[gray!70, semithick, dashed] plot[smooth] coordinates {(-3.505,6.485) (-3.589,6.507) (-3.652,6.541) (-3.705,6.580) (-3.751,6.623) (-3.792,6.667) (-3.829,6.711) (-3.862,6.756) (-3.893,6.800) (-3.921,6.845) (-3.947,6.888) (-3.971,6.930) (-3.993,6.972) (-4.013,7.012) (-4.032,7.051) (-4.049,7.089) (-4.065,7.125) (-4.080,7.160) (-4.093,7.194) (-4.105,7.225) (-4.116,7.256) (-4.126,7.284) (-4.135,7.311)};
    \draw[gray!70, semithick] plot[smooth] coordinates {(-4.135,7.311) (-4.143,7.337) (-4.149,7.360) (-4.155,7.382) (-4.159,7.402) (-4.163,7.420) (-4.165,7.436) (-4.167,7.451) (-4.168,7.463) (-4.167,7.474) (-4.166,7.483) (-4.164,7.490) (-4.161,7.495) (-4.157,7.498) (-4.152,7.499) (-4.146,7.498) (-4.139,7.495) (-4.131,7.490) (-4.122,7.483) (-4.113,7.474) (-4.102,7.463) (-4.091,7.450) (-4.078,7.435) (-4.065,7.419) (-4.051,7.400) (-4.036,7.379) (-4.019,7.356) (-4.002,7.331) (-3.984,7.304) (-3.965,7.276) (-3.945,7.245) (-3.924,7.212) (-3.902,7.177) (-3.879,7.140) (-3.855,7.102) (-3.830,7.061) (-3.803,7.019) (-3.776,6.974) (-3.747,6.928) (-3.718,6.879) (-3.687,6.829) (-3.654,6.777) (-3.621,6.723) (-3.586,6.667) (-3.550,6.610) (-3.512,6.550) (-3.473,6.489) (-3.432,6.426) (-3.390,6.361) (-3.346,6.294) (-3.300,6.226) (-3.252,6.156) (-3.202,6.084) (-3.150,6.011) (-3.095,5.936) (-3.038,5.859) (-2.977,5.781) (-2.913,5.701) (-2.846,5.619) (-2.774,5.536) (-2.696,5.452) (-2.613,5.366) (-2.521,5.278) (-2.420,5.189) (-2.304,5.099) (-2.166,5.007) (-1.986,4.914) (-1.550,4.820) (-1.114,4.914) (-0.933,5.007) (-0.794,5.099) (-0.677,5.189) (-0.573,5.278) (-0.479,5.366) (-0.392,5.452) (-0.311,5.536) (-0.235,5.619) (-0.163,5.701) (-0.094,5.781) (-0.028,5.859) (0.035,5.936) (0.096,6.011) (0.155,6.084) (0.213,6.156) (0.269,6.226) (0.323,6.294) (0.376,6.361) (0.427,6.426) (0.478,6.489) (0.527,6.550) (0.576,6.610) (0.623,6.667) (0.670,6.723) (0.715,6.777) (0.760,6.829) (0.805,6.879) (0.848,6.928) (0.891,6.974) (0.933,7.019) (0.975,7.061) (1.016,7.102) (1.056,7.140) (1.096,7.177) (1.136,7.212) (1.175,7.245) (1.213,7.276) (1.251,7.304) (1.289,7.331) (1.326,7.356) (1.363,7.379) (1.400,7.400) (1.436,7.419) (1.472,7.435) (1.507,7.450) (1.542,7.463) (1.577,7.474) (1.611,7.483) (1.646,7.490) (1.680,7.495) (1.713,7.498) (1.747,7.499) (1.780,7.498) (1.813,7.495) (1.846,7.490) (1.878,7.483) (1.911,7.474) (1.943,7.463) (1.975,7.451) (2.007,7.436) (2.039,7.420) (2.071,7.402) (2.102,7.382) (2.134,7.360) (2.166,7.337) (2.197,7.311) (2.229,7.284) (2.261,7.256) (2.292,7.225) (2.324,7.194) (2.357,7.160) (2.389,7.125) (2.422,7.089) (2.455,7.051) (2.488,7.012) (2.522,6.972) (2.557,6.930) (2.592,6.888) (2.629,6.845) (2.667,6.800) (2.706,6.756) (2.747,6.711) (2.790,6.667) (2.836,6.623) (2.887,6.580) (2.943,6.541) (3.009,6.507) (3.095,6.485) (3.300,6.545) (3.505,6.771) (3.589,6.912) (3.652,7.037) (3.705,7.153) (3.751,7.263) (3.792,7.368) (3.829,7.469) (3.862,7.566) (3.893,7.659) (3.921,7.749) (3.947,7.836) (3.971,7.920) (3.993,8.001) (4.013,8.079) (4.032,8.155) (4.049,8.228) (4.065,8.299) (4.080,8.366) (4.093,8.432) (4.105,8.495) (4.116,8.555) (4.126,8.613) (4.135,8.668) (4.143,8.721) (4.149,8.771) (4.155,8.819) (4.159,8.864) (4.163,8.907) (4.165,8.947) (4.167,8.985) (4.168,9.021) (4.167,9.054) (4.166,9.084) (4.164,9.112) (4.161,9.138) (4.157,9.161) (4.152,9.182) (4.146,9.200) (4.139,9.216) (4.131,9.230) (4.122,9.241) (4.113,9.249) (4.102,9.255) (4.091,9.259) (4.078,9.260) (4.065,9.259) (4.051,9.256) (4.036,9.250) (4.019,9.241) (4.002,9.230) (3.984,9.217) (3.965,9.202) (3.945,9.184) (3.924,9.164) (3.902,9.141) (3.879,9.116) (3.855,9.089) (3.830,9.059) (3.803,9.028) (3.776,8.993) (3.747,8.957) (3.718,8.918) (3.687,8.877) (3.654,8.834) (3.621,8.789) (3.586,8.741) (3.550,8.692) (3.512,8.640) (3.473,8.586) (3.432,8.529) (3.390,8.471) (3.346,8.411) (3.300,8.348) (3.252,8.283) (3.202,8.217) (3.150,8.148) (3.095,8.078) (3.038,8.005) (2.977,7.931) (2.913,7.854) (2.846,7.776) (2.774,7.696) (2.696,7.614) (2.613,7.530) (2.521,7.444) (2.420,7.357) (2.304,7.268) (2.166,7.177)};
    \draw[gray!85, semithick] plot[smooth] coordinates {(-4.121,7.268) (-4.107,7.225) (-4.094,7.181) (-4.080,7.138) (-4.066,7.095) (-4.053,7.053) (-4.039,7.010) (-4.025,6.968) (-4.011,6.925) (-3.997,6.883) (-3.984,6.841) (-3.970,6.800) (-3.956,6.758) (-3.942,6.716) (-3.929,6.675) (-3.915,6.634) (-3.901,6.593) (-3.887,6.552) (-3.874,6.511) (-3.860,6.470) (-3.846,6.430) (-3.832,6.390) (-3.819,6.350) (-3.805,6.310) (-3.791,6.270) (-3.777,6.230) (-3.764,6.190) (-3.750,6.151) (-3.736,6.112) (-3.723,6.073) (-3.709,6.034) (-3.695,5.995) (-3.681,5.956) (-3.667,5.918) (-3.654,5.879) (-3.640,5.841) (-3.626,5.803) (-3.612,5.765) (-3.599,5.728) (-3.585,5.690) (-3.571,5.653) (-3.558,5.615) (-3.544,5.578) (-3.530,5.541) (-3.516,5.504) (-3.502,5.468) (-3.489,5.431) (-3.475,5.395) (-3.461,5.359) (-3.447,5.322) (-3.434,5.287) (-3.420,5.251) (-3.406,5.215) (-3.393,5.180) (-3.379,5.144) (-3.365,5.109) (-3.351,5.074) (-3.337,5.039) (-3.324,5.005) (-3.310,4.970) (-3.296,4.936) (-3.282,4.901) (-3.269,4.867) (-3.255,4.833) (-3.241,4.800) (-3.227,4.766) (-3.214,4.732) (-3.200,4.699) (-3.186,4.666) (-3.172,4.633) (-3.159,4.600) (-3.145,4.567) (-3.131,4.535) (-3.118,4.502) (-3.104,4.470) (-3.090,4.438) (-3.076,4.406) (-3.062,4.374) (-3.049,4.342) (-3.035,4.311) (-3.021,4.279) (-3.008,4.248) (-2.994,4.217) (-2.980,4.186) (-2.966,4.155) (-2.952,4.125) (-2.939,4.094) (-2.925,4.064) (-2.911,4.034) (-2.897,4.004) (-2.884,3.974) (-2.870,3.944) (-2.856,3.914) (-2.842,3.885) (-2.829,3.856) (-2.815,3.827) (-2.801,3.798) (-2.787,3.769) (-2.774,3.740) (-2.760,3.712) (-2.746,3.683) (-2.732,3.655) (-2.719,3.627) (-2.705,3.599) (-2.691,3.571) (-2.677,3.544) (-2.664,3.516) (-2.650,3.489) (-2.636,3.462) (-2.623,3.435) (-2.609,3.408) (-2.595,3.381) (-2.581,3.355) (-2.567,3.328) (-2.554,3.302) (-2.540,3.276) (-2.526,3.250) (-2.512,3.224) (-2.499,3.199) (-2.485,3.173) (-2.471,3.148) (-2.457,3.123) (-2.444,3.098) (-2.430,3.073) (-2.416,3.048) (-2.402,3.023) (-2.389,2.999) (-2.375,2.975) (-2.361,2.951) (-2.348,2.927) (-2.334,2.903) (-2.320,2.879) (-2.306,2.856) (-2.292,2.832) (-2.279,2.809) (-2.265,2.786) (-2.251,2.763) (-2.237,2.740) (-2.224,2.718) (-2.210,2.695) (-2.196,2.673) (-2.183,2.651) (-2.169,2.629) (-2.155,2.607) (-2.141,2.585) (-2.127,2.564) (-2.114,2.542) (-2.100,2.521) (-2.086,2.500) (-2.072,2.479) (-2.059,2.458) (-2.045,2.438) (-2.031,2.417) (-2.018,2.397) (-2.004,2.376) (-1.990,2.356) (-1.976,2.337) (-1.962,2.317) (-1.949,2.297) (-1.935,2.278) (-1.921,2.259) (-1.907,2.239) (-1.894,2.220) (-1.880,2.202) (-1.866,2.183) (-1.852,2.164) (-1.839,2.146) (-1.825,2.128) (-1.811,2.110) (-1.798,2.092) (-1.784,2.074) (-1.770,2.056) (-1.756,2.039) (-1.742,2.022) (-1.729,2.004) (-1.715,1.987) (-1.701,1.971) (-1.688,1.954) (-1.674,1.937) (-1.660,1.921) (-1.646,1.905) (-1.632,1.888) (-1.619,1.873) (-1.605,1.857) (-1.591,1.841) (-1.577,1.826) (-1.564,1.810) (-1.550,1.795) (-1.536,1.780) (-1.522,1.765) (-1.509,1.750) (-1.495,1.736) (-1.481,1.721) (-1.467,1.707) (-1.454,1.693) (-1.440,1.679) (-1.426,1.665) (-1.412,1.651) (-1.399,1.638) (-1.385,1.624) (-1.371,1.611) (-1.358,1.598) (-1.344,1.585) (-1.330,1.572) (-1.316,1.560) (-1.303,1.547) (-1.289,1.535) (-1.275,1.523) (-1.261,1.511) (-1.248,1.499) (-1.234,1.487) (-1.220,1.476) (-1.206,1.464) (-1.193,1.453) (-1.179,1.442) (-1.165,1.431) (-1.151,1.420) (-1.137,1.409) (-1.124,1.399) (-1.110,1.388) (-1.096,1.378) (-1.083,1.368) (-1.069,1.358) (-1.055,1.349) (-1.041,1.339) (-1.027,1.329) (-1.014,1.320) (-1.000,1.311) (-0.986,1.302) (-0.972,1.293) (-0.959,1.284) (-0.945,1.276) (-0.931,1.268) (-0.917,1.259) (-0.904,1.251) (-0.890,1.243) (-0.876,1.235) (-0.862,1.228) (-0.849,1.220) (-0.835,1.213) (-0.821,1.206) (-0.807,1.199) (-0.794,1.192) (-0.780,1.185) (-0.766,1.179) (-0.753,1.172) (-0.739,1.166) (-0.725,1.160) (-0.711,1.154) (-0.698,1.148) (-0.684,1.142) (-0.670,1.137) (-0.656,1.131) (-0.643,1.126) (-0.629,1.121) (-0.615,1.116) (-0.601,1.111) (-0.588,1.107) (-0.574,1.102) (-0.560,1.098) (-0.546,1.094) (-0.532,1.090) (-0.519,1.086) (-0.505,1.082) (-0.491,1.079) (-0.477,1.075) (-0.464,1.072) (-0.450,1.069) (-0.436,1.066) (-0.422,1.063) (-0.409,1.061) (-0.395,1.058) (-0.381,1.056) (-0.367,1.054) (-0.354,1.052) (-0.340,1.050) (-0.326,1.048) (-0.312,1.046) (-0.299,1.045) (-0.285,1.044) (-0.271,1.042) (-0.257,1.041) (-0.244,1.041) (-0.230,1.040) (-0.216,1.039) (-0.202,1.039) (-0.189,1.039) (-0.175,1.039) (-0.161,1.039) (-0.148,1.039) (-0.134,1.039) (-0.120,1.040) (-0.106,1.041) (-0.093,1.041) (-0.079,1.042) (-0.065,1.044) (-0.051,1.045) (-0.038,1.046) (-0.024,1.048) (-0.010,1.050) (0.004,1.052) (0.017,1.054) (0.031,1.056) (0.045,1.058) (0.059,1.061) (0.072,1.063) (0.086,1.066) (0.100,1.069) (0.114,1.072) (0.128,1.075) (0.141,1.079) (0.155,1.082) (0.169,1.086) (0.183,1.090) (0.196,1.094) (0.210,1.098) (0.224,1.102) (0.238,1.107) (0.251,1.111) (0.265,1.116) (0.279,1.121) (0.293,1.126) (0.306,1.131) (0.320,1.137) (0.334,1.142) (0.348,1.148) (0.361,1.154) (0.375,1.160) (0.389,1.166) (0.402,1.172) (0.416,1.179) (0.430,1.185) (0.444,1.192) (0.457,1.199) (0.471,1.206) (0.485,1.213) (0.499,1.220) (0.512,1.228) (0.526,1.235) (0.540,1.243) (0.554,1.251) (0.567,1.259) (0.581,1.268) (0.595,1.276) (0.609,1.284) (0.622,1.293) (0.636,1.302) (0.650,1.311) (0.664,1.320) (0.677,1.329) (0.691,1.339) (0.705,1.349) (0.719,1.358) (0.732,1.368) (0.746,1.378) (0.760,1.388) (0.774,1.399) (0.787,1.409) (0.801,1.420) (0.815,1.431) (0.829,1.442) (0.843,1.453) (0.856,1.464) (0.870,1.476) (0.884,1.487) (0.897,1.499) (0.911,1.511) (0.925,1.523) (0.939,1.535) (0.953,1.547) (0.966,1.560) (0.980,1.572) (0.994,1.585) (1.007,1.598) (1.021,1.611) (1.035,1.624) (1.049,1.638) (1.063,1.651) (1.076,1.665) (1.090,1.679) (1.104,1.693) (1.117,1.707) (1.131,1.721) (1.145,1.736) (1.159,1.750) (1.173,1.765) (1.186,1.780) (1.200,1.795) (1.214,1.810) (1.227,1.826) (1.241,1.841) (1.255,1.857) (1.269,1.873) (1.282,1.888) (1.296,1.905) (1.310,1.921) (1.324,1.937) (1.337,1.954) (1.351,1.971) (1.365,1.987) (1.379,2.004) (1.392,2.022) (1.406,2.039) (1.420,2.056) (1.434,2.074) (1.448,2.092) (1.461,2.110) (1.475,2.128) (1.489,2.146) (1.502,2.164) (1.516,2.183) (1.530,2.202) (1.544,2.220) (1.558,2.239) (1.571,2.259) (1.585,2.278) (1.599,2.297) (1.612,2.317) (1.626,2.337) (1.640,2.356) (1.654,2.376) (1.668,2.397) (1.681,2.417) (1.695,2.438) (1.709,2.458) (1.722,2.479) (1.736,2.500) (1.750,2.521) (1.764,2.542) (1.778,2.564) (1.791,2.585) (1.805,2.607) (1.819,2.629) (1.832,2.651) (1.846,2.673) (1.860,2.695) (1.874,2.718) (1.887,2.740) (1.901,2.763) (1.915,2.786) (1.929,2.809) (1.942,2.832) (1.956,2.856) (1.970,2.879) (1.984,2.903) (1.997,2.927) (2.011,2.951) (2.025,2.975) (2.039,2.999) (2.053,3.023) (2.066,3.048) (2.080,3.073) (2.094,3.098) (2.107,3.123) (2.121,3.148) (2.135,3.173) (2.149,3.199) (2.163,3.224) (2.176,3.250) (2.190,3.276) (2.204,3.302) (2.217,3.328) (2.231,3.355) (2.245,3.381) (2.259,3.408) (2.273,3.435) (2.286,3.462) (2.300,3.489) (2.314,3.516) (2.327,3.544) (2.341,3.571) (2.355,3.599) (2.369,3.627) (2.383,3.655) (2.396,3.683) (2.410,3.712) (2.424,3.740) (2.438,3.769) (2.451,3.798) (2.465,3.827) (2.479,3.856) (2.492,3.885) (2.506,3.914) (2.520,3.944) (2.534,3.974) (2.547,4.004) (2.561,4.034) (2.575,4.064) (2.589,4.094) (2.602,4.125) (2.616,4.155) (2.630,4.186) (2.644,4.217) (2.658,4.248) (2.671,4.279) (2.685,4.311) (2.699,4.342) (2.712,4.374) (2.726,4.406) (2.740,4.438) (2.754,4.470) (2.768,4.502) (2.781,4.535) (2.795,4.567) (2.809,4.600) (2.822,4.633) (2.836,4.666) (2.850,4.699) (2.864,4.732) (2.878,4.766) (2.891,4.800) (2.905,4.833) (2.919,4.867) (2.932,4.901) (2.946,4.936) (2.960,4.970) (2.974,5.005) (2.988,5.039) (3.001,5.074) (3.015,5.109) (3.029,5.144) (3.042,5.180) (3.056,5.215) (3.070,5.251) (3.084,5.287) (3.097,5.322) (3.111,5.359) (3.125,5.395) (3.139,5.431) (3.152,5.468) (3.166,5.504) (3.180,5.541) (3.194,5.578) (3.207,5.615) (3.221,5.653) (3.235,5.690) (3.249,5.728) (3.263,5.765) (3.276,5.803) (3.290,5.841) (3.304,5.879) (3.317,5.918) (3.331,5.956) (3.345,5.995) (3.359,6.034) (3.373,6.073) (3.386,6.112) (3.400,6.151) (3.414,6.190) (3.427,6.230) (3.441,6.270) (3.455,6.310) (3.469,6.350) (3.482,6.390) (3.496,6.430) (3.510,6.470) (3.524,6.511) (3.537,6.552) (3.551,6.593) (3.565,6.634) (3.579,6.675) (3.593,6.716) (3.606,6.758) (3.620,6.800) (3.634,6.841) (3.647,6.883) (3.661,6.925) (3.675,6.968) (3.689,7.010) (3.702,7.053) (3.716,7.095) (3.730,7.138) (3.744,7.181) (3.757,7.225) (3.771,7.268) (3.785,7.311) (3.799,7.355) (3.812,7.399) (3.826,7.443) (3.840,7.487) (3.854,7.531) (3.867,7.575) (3.881,7.620) (3.895,7.665) (3.909,7.710)};
    \filldraw[fill=blue!6, draw=blue!40, dashed] (0.900,2.870) -- (2.750,3.394) -- (-0.125,4.010) -- cycle;
    \filldraw[fill=blue!20, draw=blue!70, very thick, fill opacity=0.85] (-2.600,3.400) -- (0.900,2.870) -- (-0.125,4.010) -- cycle;
    \draw[black, dashed, thin] (-2.600,-0.420) -- (-2.600,3.400);
    \draw[black, dashed, thin] (0.900,-0.490) -- (0.900,2.870);
    \draw[black, dashed, thin] (-0.125,0.647) -- (-0.125,1.040);
    \draw[gray, dashed, thin] (-0.125,1.040) -- (-0.125,4.010);
    \draw[black, dashed, thin] (-0.850,0.105) -- (-0.850,1.750);
    \draw[red, very thick, dotted] (-0.850,1.750) -- (-0.850,3.629);
    \node[right, red, xshift=1pt, yshift=-6pt] at (-0.850,2.689) {\footnotesize $g_V(\bm w_A)$};
    \draw[black, dashed, thin] (2.750,0.385) -- (2.750,3.394);
    \draw[red, very thick] (2.750,3.394) -- (2.750,4.510);
    \node[right, red] at (2.750,3.952) {\footnotesize $q_5^A(V)$};
    \filldraw[blue] (-2.600,-0.420) circle (1.4pt);
    \filldraw[blue!70!black] (-2.600,3.400) circle (1.4pt);
    \filldraw[blue] (0.900,-0.490) circle (1.4pt);
    \filldraw[blue!70!black] (0.900,2.870) circle (1.4pt);
    \filldraw[blue] (-0.125,0.647) circle (1.4pt);
    \filldraw[blue!70!black] (-0.125,4.010) circle (1.4pt);
    \filldraw[blue] (2.750,0.385) circle (1.4pt);
    \filldraw[blue!70!black] (2.750,4.510) circle (1.4pt);
    \filldraw[red] (-0.850,0.105) ++(0,2.4pt) -- ++(2.4pt,-2.4pt) -- ++(-2.4pt,-2.4pt) -- ++(-2.4pt,2.4pt) -- cycle;
    \draw[red, thick] (-0.850,1.750) ++(0,2.4pt) -- ++(2.4pt,-2.4pt) -- ++(-2.4pt,-2.4pt) -- ++(-2.4pt,2.4pt) -- cycle;
    \filldraw[orange] (-0.850,3.629) ++(-1.6pt,-1.6pt) rectangle ++(3.2pt,3.2pt);
    \draw[blue!70!black] (2.750,3.394) circle (1.4pt);
    \node[below] at (-2.600,-0.420) {$X_2$};
    \node[below] at (0.900,-0.490) {$X_3$};
    \node[right] at (-0.125,0.647) {$X_4$};
    \node[right] at (-0.850,0.105) {$X_1$};
    \node[below right] at (2.750,0.385) {$X_5$};
    \node[right, xshift=1pt, yshift=2pt] at (2.750,4.510) {\footnotesize $(X_5,\ell_V(X_5))$};
    \end{tikzpicture}
    \vspace{6pt}
    \tabnotes{The shaded triangle in the predictor plane is the admissible cell $A=\{2,3,4\}$, which contains the treated unit $X_1$. Each point $x$ is lifted to the height $\ell_V(x)$ above it. The gray outline is the graph of $\ell_V$, and every lifted point lies on it. The lifted donors of $A$ span the plane $H_{A,V}$, and the exact match $\bm w_A$ places the treated unit at the point $(X_1,H_{A,V}(X_1))$ on this plane. The lifted treated point lies below the plane, and the dotted red segment between the two is the penalty $g_V(\bm w_A)$ of Section~\ref{sec:int_penalty}. For the outside donor $X_5$, the plane is extended (dashed) to the point above $X_5$. The solid red segment is the gap $q_5^A(V)$ between the lifted donor and the plane. It is positive, so $(X_5, \ell_V(X_5))$ lies above the plane.}
\end{figure}
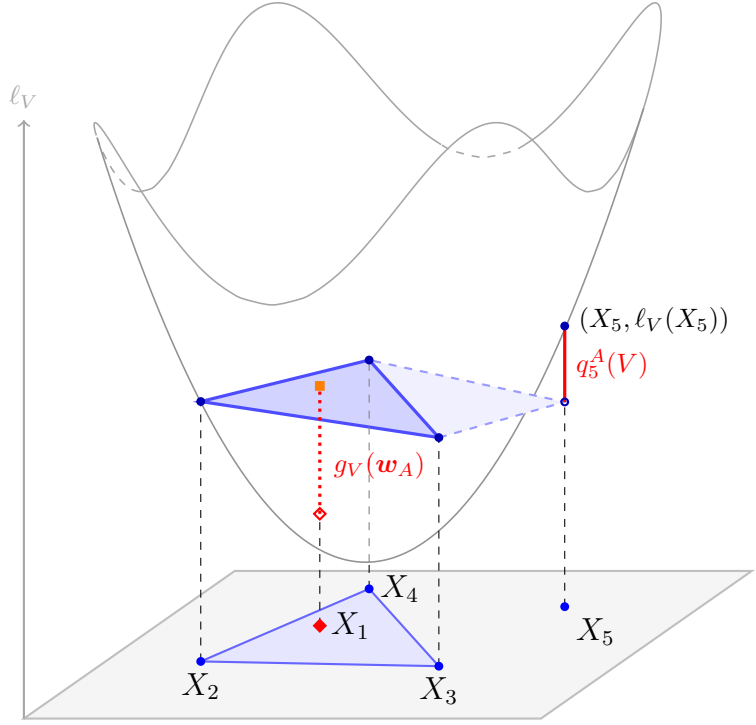

Figure~\ref{fig:lift} illustrates the geometry behind the selection among admissible cells. Fix an admissible cell $A$. Its $K+1$ lifted donors $(X_i,\ell_V(X_i))$, $i\in A$, span a unique hyperplane in $\R^{K+1}$, the plane that includes the two upper triangles in Figure~\ref{fig:lift}. Let $H_{A,V}(x)$ be the height of this hyperplane above any point $x$, so that $H_{A,V}(X_i)=\ell_V(X_i)$ for $i\in A$. For a donor $j\notin A$, the signed vertical distance between its lifted point and the hyperplane $q_j^A(V)=\ell_V(X_j)-H_{A,V}(X_j)$ is positive when the lifted donor lies above the hyperplane.

To express $q_j^A(V)$ as a function of $V$, write $X_j$ as an affine combination of the donors in $A$, $X_j=\sum_{i\in A}\beta_i^jX_i$ with $\sum_{i\in A}\beta_i^j=1$. These coefficients are unique by part (iii) of Assumption \ref{ass:gp_int} and they need not be positive. Because $H_{A,V}$ is affine, $H_{A,V}(X_j)=\sum_{i\in A}\beta_i^j\ell_V(X_i)$. 
Let $\inner{Q}{V}=\tr(Q'V)$ and 
\[
Q_j^A=X_jX_j'-\sum_{i\in A}\beta_i^jX_iX_i'.
\]
Then,
\begin{align}
q_j^A(V)&=\ell_V(X_j)-\sum_{i\in A}\beta_i^j\ell_V(X_i)\notag\\
&=X_j'VX_j-\sum_{i\in A}\beta_i^jX_i'VX_i\notag\\
&=\tr(X_jX_j'V)-\sum_{i\in A}\beta_i^j\tr(X_iX_i'V)\notag\\
&=\inner{Q_j^A}{V}.
\label{eq:Qmatrix}
\end{align}
As a result, $q_j^A(V)$ is linear in $V$.

Recall from Section~\ref{sec:int_penalty} that the penalty minimizer is the lowest point of the lifted donor hull directly above $X_1$. The weights
$\bm w_A$ produce the point $(X_1,H_{A,V}(X_1))$ in the simplex formed by
the lifted donors in $A$. If every other lifted donor lies strictly above
the plane $H_{A,V}$, so does every convex combination of lifted donors that puts weight outside $A$, and $(X_1,H_{A,V}(X_1))$ is the lowest point in the lifted
donor hull above $X_1$.
Thus $\bm w_A$ is uniquely selected exactly when $q_j^A(V)>0$ for every
$j\notin A$. The predictor weight matrices $V$
for which $\bm w_A$ is the unique solution of~\eqref{eq:int_limit} form
the set
\begin{equation}
\cV(A)=\{V\in\mathbb S^K_{++}:\inner{Q_j^A}{V}>0
\text{ for every }j\notin A\}.
\label{eq:VA}
\end{equation}
Lemma \ref{lem:selection} in
Appendix~\ref{app:geom} gives the formal proof.

Condition~\eqref{eq:VA} has a Delaunay interpretation. Let $\widetilde X_j=V^{1/2}X_j$, so that $\ell_V(X_j)=\|\widetilde X_j\|^2$, and let $\widetilde H_A(x)$ be the height above $x$ of the hyperplane through the points $(\widetilde X_i,\|\widetilde X_i\|^2)$, $i\in A$, so that $H_{A,V}(X_j)=\widetilde H_A(\widetilde X_j)$. The function $\norm{x}^2-\widetilde H_A(x)$ is zero on the unique sphere through $\widetilde X_i$, $i\in A$, negative inside it, and positive outside it. Condition~\eqref{eq:VA} thus says that no other transformed donor $\widetilde X_j$ lies inside or on this sphere, which is the defining property of a Delaunay simplex. That is, $\bm w_A$ is uniquely selected exactly when $A$ is the Delaunay simplex of the transformed donors that contains $\widetilde X_1$. For $V=I$ this is Theorem~2 of \citet{abadie2021penalized}.

\subsection{Attainability Tests and Finite Enumeration}
\label{sec:int_computation}

An admissible cell $A$ is \emph{attainable} for a given class of matrices if $V\in \mathcal V(A)$ for some $V$ in the class.
To test attainability in the positive definite class, we  solve the program
\begin{equation}
\begin{aligned}
\theta^*(A)=\max_{V,\theta}\quad &\theta\\
\text{subject to}\quad
&\inner{Q_j^A}{V}\geq\theta,\quad j\notin A,\\
&V\succeq0,\quad \tr(V)=1.
\end{aligned}
\label{eq:thetastar}
\end{equation}
The objective maximizes the smallest selection slack. The cell $A$ is attainable if and only if $\theta^*(A)$ is positive. The optimal $V$ may be singular but, for small enough $\varepsilon>0$, $V+\varepsilon I$ satisfies every inequality. Rescaling it to unit trace gives a positive definite matrix in $\cV(A)$.

For a fixed cell $A$, let $\bm\rho_j=\diag(Q_j^A)$. Because $\inner{Q}{\Diag(\bm d)}=\sum_{k=1}^K Q_{kk}d_k$, under a diagonal
metric $V=\Diag(\bm d)$, the selection inequalities reduce to
$\bm\rho_j'\bm d>0$, and the attainability test becomes 
\begin{equation}
\begin{aligned}
\theta_+^*(A)=\max_{\bm d,\theta}\quad &\theta\\
\text{subject to}\quad
&\bm\rho_j'\bm d\geq\theta,\quad j\notin A,\\
&\bm d\geq0,\quad\mathbf{1}'\bm d=1.
\end{aligned}
\label{eq:int_diagonal_margin}
\end{equation}
The cell is attainable by a positive diagonal metric if and only
if $\theta_+^*(A)>0$. As in the full class, the optimizer may have zero entries, but $\bm d+\varepsilon\mathbf 1$ satisfies every inequality for small $\varepsilon>0$, and rescaling it to sum to one gives a strictly positive diagonal matrix in $\cV(A)$.

Only the admissible cells need to be examined for attainability. They are the vertices of the polytope $\cF$ in~\eqref{eq:int_exact}. 
The optimum is the attainable vector $\bm w_A$ with the smallest validation loss.

Let $\cN$ be the set of matrices $V\in\mathbb{S}^K_{++}$ for which some $K+2$ transformed donors $\widetilde X_j$ lie on a common sphere. For such $V$, the Delaunay simplex of the transformed donors that contains $\widetilde X_1$ need not be unique. This set (and its intersection with the diagonal class) has measure zero (Proposition~\ref{prop:coverage}). 

\begin{proposition}
\label{prop:enumeration}
Let Assumption~\ref{ass:gp_int} hold, and let $\mathcal{C}$ be either $\mathbb{S}^K_{++}$ or the class of strictly positive diagonal matrices. Let $\mathcal{A}_{\mathcal{C}}$ be the set of admissible cells attainable by a metric in $\mathcal{C}$. Outside $\cN$,  the set $\{\bm w(0^+,V) : V \in \mathcal{C}\setminus\cN\}$ equals $\{\bm w_A : A \in \mathcal{A}_{\mathcal{C}}\}$. As a result,
\[
\inf_{V \in \mathcal{C} \setminus \cN} \norm{Z_1-Z_0\bm{w}(0^+, V)}^2 = \min_{A \in \mathcal{A}_{\mathcal{C}}} \norm{Z_1-Z_0\bm{w}_A}^2.
\]
\end{proposition}

Once a cell $A^*$ is selected,  $\cV(A^*)$
characterizes the metrics that rationalize the chosen donor weights. 
A representative metric (e.g., the analytic center) can be selected as in Section~\ref{sec:selection}.

\section{Empirical Application}\label{sec:numerical}
%------------------------------------------------------------------

Proposition 99, enacted in California in 1988, was a comprehensive tobacco-control intervention that combined cigarette tax increases with antismoking media campaigns, clean indoor-air restrictions, and earmarked funding for tobacco-related education and research. \citet{Abadie2010California} estimate the effect of Proposition 99 on per-capita cigarette sales in California using the synthetic control method, with other U.S.~states serving as the donor pool. This section applies the estimators proposed in this article to the Proposition 99 data, compares the results with those from the standard bi-level optimization, and discusses practical implications.

The data are a state-level panel covering 1970-2000, with 1970-1988 as the pre-treatment period and 1989-2000 as the post-treatment period. The outcome is annual per-capita cigarette sales (in packs). California is the only treated unit, and the donor pool comprises 38 states after excluding those that adopted large-scale tobacco-control programs or substantially raised cigarette taxes during the sample period. The predictors used to construct synthetic California are log per-capita state income, the retail cigarette price, and the share of the population aged 15-24, each averaged over 1980-1988; per-capita beer consumption, averaged over 1984-1988; and, as lagged outcomes, per-capita cigarette sales in 1975, 1980, and 1988. The validation sample includes all pre-treatment periods.

Table~\ref{tab:donor-weights} compares the donor weights of \citet{Abadie2010California} with those from our exact procedure. The two are nearly identical, confirming that the common bi-level optimization works well in this application. Figure~\ref{fig:exact-cali} displays the estimated treatment effect of Proposition 99 on per-capita cigarette sales in California, which is almost identical to the result in \citet{Abadie2010California}. The corresponding predictor weights appear in Table \ref{tab:predictor_weights}, and the solution coincides with the one delivered by the method of \citet{BECKER20181}.
Table~\ref{tab:faces} reports the number of faces of the donor hull that allow attainable projections. Under diagonal $V$, $226$ of the $3{,}839$ visible faces are attainable. 
A branch-and-bound algorithm reaches the optimum in under one second. Existing packages may take several minutes on the same data and do not guarantee a global optimum. 
Appendix~\ref{app:pd} compares these results with those over the full positive definite class. The optimum over the full positive definite class lowers the validation loss but keeps the same five donors and leaves the estimated
effect essentially unchanged.

\begin{table}[t]
    \centering
    \caption{Donor weights for synthetic California}
    \label{tab:donor-weights}
    \begin{tabular}{lcc}
  \hline
  State & \citet{Abadie2010California} & Exact (this article) \\
  \hline
  Colorado    & 0.164 & 0.160 \\
  Connecticut & 0.069 & 0.068 \\
  Montana     & 0.199 & 0.202 \\
  Nevada      & 0.234 & 0.236 \\
  Utah        & 0.334 & 0.335 \\
  \hline
\end{tabular}
    \tabnotes{The first column reports the donor weights of \citet{Abadie2010California}. The second column reports the exact global solution. The validation sample includes all pre-treatment years (1970-1988). The remaining 33 donor states receive zero weight in both columns and are omitted.}
\end{table}

\begin{figure}[t]
    \centering
    \caption{Estimated effect of California's Proposition 99}
    \label{fig:exact-cali}
    \includegraphics[width=0.95\linewidth]{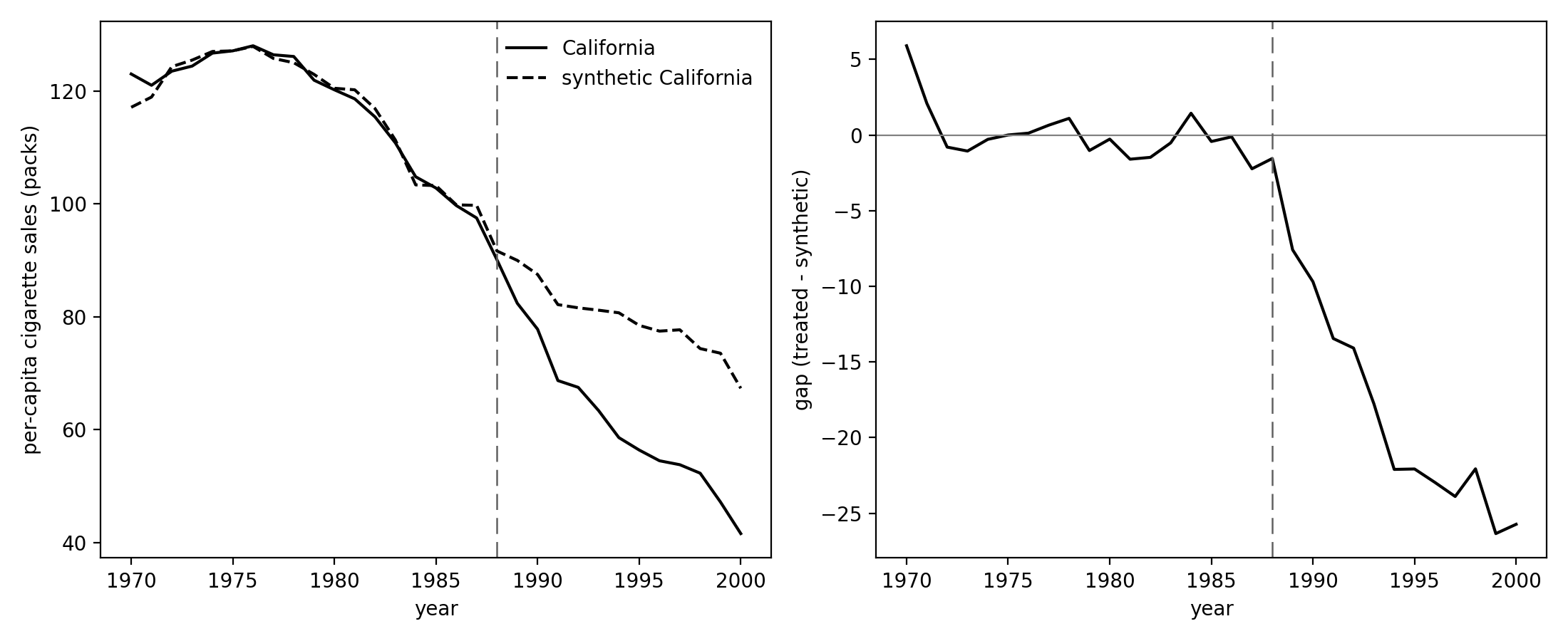}
    \tabnotes{Per-capita cigarette sales in packs. Synthetic California is constructed with the exact donor weights of Table~\ref{tab:donor-weights}. Proposition 99 takes effect in 1989. The vertical line at 1988 marks the last pre-treatment year.}
\end{figure}

\begin{table}[t]
\centering
\caption{Predictor weights for synthetic California: diagonal class}
\label{tab:predictor_weights}
\begin{tabular}{lc}
\hline
Predictor & Weight \\
\hline
Log per-capita state income          & 0.0000 \\
Share of population aged 15-24      & 0.3333 \\
Retail cigarette price               & 0.3333 \\
Per-capita beer consumption          & 0.0000 \\
Cigarette sales per capita (1988)    & 0.0000 \\
Cigarette sales per capita (1980)    & 0.0000 \\
Cigarette sales per capita (1975)    & 0.3333 \\
\hline
\end{tabular}
\tabnotes{Log-determinant maximizer over the three retained predictors, on standardized predictors. The match on these predictors is exact, so the maximizer is uniform (Section~\ref{sec:selection}).}
\end{table}

\begin{table}[t]
\centering
\caption{Attainable faces of the donor hull, by dimension}
\label{tab:faces}

\begin{tabular}{cccc}
\hline
Dimensions & \#Attainable Faces & \#Visible Faces & \#Faces \\
\hline
0 & 0 & 25 & 34 \\
1 & 6 & 183 &  416\\
2 & 29 & 620 & 2271 \\
3 & 59 & 1129 & 6235 \\
4 & 69 & 1143 & 8983 \\
5 & 49 & 607 & 6489\\
6 & 14 & 132 & 1854 \\
\hline
Total & 226 & 3839 & 26282 \\
\hline
\end{tabular}
\tabnotes{Number of faces of the convex hull of the 38 donor predictor vectors ($K = 7$) that contain attainable projections under positive diagonal $V$ in the California application.}
\end{table}

The predictor weights in Table~\ref{tab:predictor_weights} illustrate the distinction between forced and incidental sparsity drawn in Section~\ref{sec:selection}. Four of the seven predictors receive exactly zero weight. The feasibility conditions~\eqref{eq:diag_lp}, evaluated at the optimal donor weights $\bm{w}^*$, show that, within the diagonal class, this sparsity is forced. Over the set of diagonal rationalizers of $\bm{w}^*$, the maximum of $d_k$ is zero for log income, beer consumption, and the 1980 and 1988 lagged outcomes, so every diagonal matrix that rationalizes $\bm{w}^*$ is supported on the share of population aged 15-24, the retail cigarette price, and the 1975 lagged outcome. That is, $\mathcal{F}(\bm{w}^*) \cap \mathcal{V}_+ = \varnothing$, and $\bm{w}^*$ is attainable only as a limit of interior solutions. 

Table~\ref{tab:predictor_weights} reports the log-determinant maximizer over the three retained predictors (Section~\ref{sec:selection}). The match is exact on the three predictors, so every diagonal matrix supported on them is consistent with $\bm{w}^*$. At the reported $V$, the inner problem need not return $\bm w^*$, because other solutions exist (Section~\ref{sec:boundary}). By Proposition~\ref{prop:visible_facets}, however, positive definite rationalizers exist, and at any of them the inner optimum is unique, so the matrix determines $\bm{w}^*$.
Table~\ref{tab:vstar} reports two representatives of the full positive definite rationalizing cone $\mathcal{F}(\bm{w}^*) \cap \{\tr(V) = 1\}$ on standardized predictors, the trace-normalized diagonal of the log-determinant maximizer $V^*$ and that of the analytic center of Proposition~\ref{prop:extcenter}. The solution is positive definite (smallest eigenvalue $3.7 \times 10^{-3}$), and the solution of the inner problem at $V^*$ reproduces $\bm{w}^*$ up to numerical precision. As Proposition~\ref{prop:dense_V} guarantees, all seven predictors receive strictly positive shares. That is, once the metric is allowed to reward linear combinations of predictors, no predictor direction needs to be discarded. 

\begin{table}[t]
\centering
\caption{Predictor weights for synthetic California: full positive definite class}
\label{tab:vstar}
\begin{tabular}{lcc}
\hline
Predictor & Log-det maximizer & Analytic center \\
\hline
Log per-capita state income          & 0.006 & 0.030 \\
Share of population aged 15-24      & 0.143 & 0.026 \\
Retail cigarette price               & 0.144 & 0.036 \\
Per-capita beer consumption          & 0.145 & 0.034 \\
Cigarette sales per capita (1988)    & 0.224 & 0.563 \\
Cigarette sales per capita (1980)    & 0.147 & 0.031 \\
Cigarette sales per capita (1975)    & 0.191 & 0.279 \\
\hline
\end{tabular}
\tabnotes{Two representatives of the full positive definite cone of rationalizers of the donor weights of Table~\ref{tab:donor-weights}. The first column reports the trace-normalized diagonal $V^*_{kk}/\tr(V^*)$ of the log-determinant maximizer. The second column reports $V^*_{kk}/\tr(V^*)$ for the analytic center. Predictors are standardized.}
\end{table}

Like \citet{Abadie2010California}, we include in the validation sample the
three pre-treatment years whose outcomes also serve as lagged predictors. As a
check, we recompute the estimators with the validation loss restricted to the
sixteen pre-treatment years not used as predictors. The leakage-free design selects the same five donors with nearly
identical weights. The three excluded years are strongly serially correlated
with the remaining sixteen, so any donor combination that fits the sixteen
fits the three almost automatically, and the two validation criteria select
nearly the same result. 

%------------------------------------------------------------------
\section{A Leakage-Free Validation Design}\label{sec:temporal_split}
%------------------------------------------------------------------

The California application highlights a tension in the standard validation design: lagged outcomes often serve as predictors while the validation sample consists of pre-treatment outcomes, so the same observations can appear on both sides of the bi-level problem. Predictors that determine the outcome path in the validation sample have a mechanical advantage. Although the California estimates in Section~\ref{sec:numerical} are robust to this overlap, a design that removes the overlap is a cleaner alternative.

Following the cross-validation approach of \citet{Abadie2015Germany}, we partition the pre-treatment span into a \emph{training window} and a later \emph{validation window}. All predictors, including lagged outcomes, are measured in the training window, and the validation outcomes $Z$ are drawn from the validation window. By construction, no coordinate of $Z$ appears in $X$. As a result, a lagged outcome can earn validation fit only by predicting \emph{future} outcomes on the same footing as any other predictor. The geometric machinery of Sections~\ref{sec:geometric}-\ref{sec:numerical} applies without change.

We illustrate the leakage-free design with the German reunification
application of \citet{Abadie2015Germany}, which estimates the effect of the
1990 reunification on the West German economy. West Germany is the treated
unit, with sixteen OECD countries as donors. The training window is
1971-1980, with six predictors: per-capita GDP, trade openness, the inflation
rate, and the industry share averaged over the window, schooling in 1970 and
1975, and the investment rate over the 1970s. The validation outcomes are
per-capita GDP in 1981-1990, so the predictor and validation windows are
disjoint.

Table~\ref{tab:germany_weights} reports the optimal donor weights next to those of \citet{Abadie2015Germany}. The exact enumeration selects five donors, four of which coincide with the donors that receive positive weight in the published cross-validated specification, with Austria the dominant donor in both solutions. Relative to the weights in \citet{Abadie2015Germany}, the exact solution shifts weight from Japan, the United States, and the Netherlands toward Australia, Austria, and Switzerland. 

\begin{table}[t]
\centering
\caption{Donor weights for synthetic West Germany}
\label{tab:germany_weights}
\begin{tabular}{lcc}
\hline
Donor & \citet{Abadie2015Germany} & Exact (this article) \\
\hline
Austria      & 0.42 & 0.507 \\
Switzerland  & 0.11 & 0.166 \\
Japan        & 0.16 & 0.146 \\
USA          & 0.22 & 0.135 \\
Australia    & 0.00 & 0.045 \\
Netherlands  & 0.09 & 0.000 \\
\hline
\end{tabular}
\tabnotes{Temporal-split design: predictors measured in the 1971-1980 training window, validation outcomes per-capita GDP in 1981-1990 (PPP-adjusted current U.S.~dollars). The first column reports the weights of \citet{Abadie2015Germany}. The second column reports the exact global solution. Donors with zero weights in both solutions are omitted.}
\end{table}

Figure~\ref{fig:exact-germany} compares the result with \citet{Abadie2015Germany}. The left panel plots per-capita GDP for West Germany against the exact synthetic control and the synthetic control in \citet{Abadie2015Germany}. The two are nearly indistinguishable before 1990, with the exact solution fitting the validation window slightly better. The post-1990 gaps tell the same substantive story, a large and growing negative effect of reunification on West German per-capita GDP, with the exact synthetic control implying a somewhat smaller average effect over 1991-2003. Appendix~\ref{app:pd} compares this solution with the optimum over the full positive definite class. The full positive definite class for $V$ cuts the validation loss in half and changes the donor portfolio, but yields a similar treatment effect.

\begin{figure}[t]
    \centering
    \caption{West Germany and two synthetic controls}
    \label{fig:exact-germany}
    \includegraphics[width=0.95\linewidth]{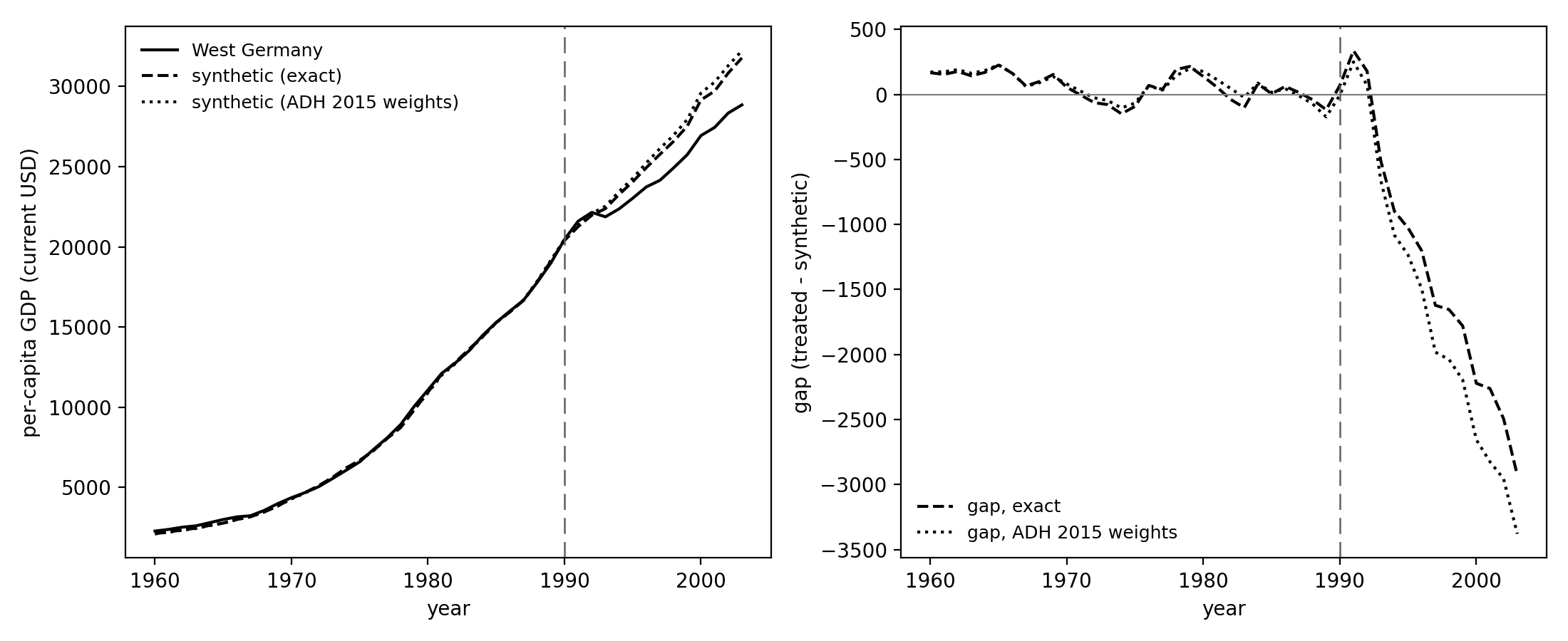}
    \tabnotes{Left panel: per-capita GDP in PPP-adjusted current U.S.~dollars for West Germany, for the exact synthetic control of the temporal-split design, and for the synthetic unit implied by the donor weights of \citet{Abadie2015Germany}; the weights are those of Table~\ref{tab:germany_weights}. Right panel: gaps between West Germany and each synthetic control. The vertical line at 1990 marks the German reunification.}
\end{figure}

Table~\ref{tab:vstar_ger} reports predictor weights for both candidate classes on standardized predictors. The first column is the analytic center of the diagonal rationalizers; the remaining two report the trace-normalized diagonals of the log-determinant maximizer (the shares $V^*_{kk}/\tr(V^*)$ of Section~\ref{sec:selection}) and of the analytic center over the full positive definite class. Every entry is strictly positive, but the diagonal and full classes allocate weight differently. The diagonal class concentrates on the lagged GDP average, with the investment rate second, while the full class spreads the weight more evenly, led by schooling and per-capita GDP.

\begin{table}[htbp]
\centering
\caption{Predictor weights for synthetic West Germany: diagonal and full positive definite classes}
\label{tab:vstar_ger}
\begin{tabular}{lccc}
\hline
 & & \multicolumn{2}{c}{Full positive definite} \\
\cline{3-4}
Predictor & Diagonal & Log-det maximizer & Analytic center \\
\hline
Per-capita GDP (1971-80)    & 0.443 & 0.205 & 0.176 \\
Trade openness (1971-80)    & 0.134 & 0.173 & 0.102 \\
Inflation rate (1971-80)    & 0.071 & 0.150 & 0.069 \\
Industry share (1971-80)    & 0.002 & 0.084 & 0.125 \\
Schooling (1970, 1975)      & 0.107 & 0.221 & 0.462 \\
Investment rate (1970s)     & 0.243 & 0.167 & 0.067 \\
\hline
\end{tabular}
\tabnotes{Standardized predictors. The first column is the analytic center of the diagonal matrices that rationalize the donor weights of Table~\ref{tab:germany_weights}. The remaining columns report the trace-normalized diagonals $V^*_{kk}/\tr(V^*)$ of the log-determinant maximizer and of the analytic center.}
\end{table}

%------------------------------------------------------------------
\section{Conclusion}\label{sec:conclusion}
%------------------------------------------------------------------

This article recasts synthetic control estimation as a geometric problem that involves projections onto the visible faces of the donor convex hull. We characterize the attainable set of donor weights under diagonal and under general positive definite predictor weight matrices, and replace the bi-level optimization with a single-level minimization of the validation loss over that set. A finite enumeration over faces solves the single-level problem exactly, in place of a nested optimization that is fragile in practice. When the treated unit lies inside the donor hull, the same approach reduces the choice among exact matches to a search over finitely many candidate simplices. The article also characterizes the set of predictor weight matrices that rationalize any given donor weight vector. Two applications, the California tobacco control program and German reunification under a leakage-free validation design, show that the exact solutions validate the previously published results while adding a global optimality guarantee.

%------------------------------------------------------------------
\clearpage
\appendix
%------------------------------------------------------------------

\numberwithin{equation}{section}
\numberwithin{theorem}{section}
\numberwithin{definition}{section}
\numberwithin{lemma}{section}
\numberwithin{proposition}{section}
\numberwithin{corollary}{section}
\numberwithin{remark}{section}
\numberwithin{assumption}{section}
\numberwithin{table}{section}
\numberwithin{figure}{section}

\section*{Appendix}

\section{Algorithms}
\label{app:algorithms}

This appendix turns the sign-compatibility criterion of Section~\ref{sec:diag_application} into an algorithm that enumerates the attainable set.

Section~\ref{sec:app_facewise} characterizes the closure of the attainable set as a union of sign-compatible regions. Algorithm~\ref{alg:vfacets}, presented in Section~\ref{sec:app_visible}, enumerates the facets of the donor hull visible from $X_1$. Algorithm~\ref{alg:pieces}, presented in Section~\ref{sec:app_pieces}, identifies the allowed sign patterns~\eqref{eq:ray_sign_test} for a given face. Finally, Section~\ref{sec:app_bnb} describes the search over faces that solves the validation-loss minimization problem~\eqref{eq:single_level_W0}.

\subsection{The Attainable Set for a Given Face}\label{sec:app_facewise}

We identify a face of $\mathcal{P}$ with the set of indices of its vertices, $\sigma \subseteq \{2, \dots, J+1\}$. Under Assumption~\ref{ass:gp}(ii) every proper face is a simplex, so $\sigma$ determines the face $\conv(\{X_j : j \in \sigma\})$ and every boundary point of $\mathcal{P}$ lies in the relative interior of exactly one face. Let
\[
F_\sigma=\conv\{X_j:j\in\sigma\},\qquad
 \Delta_\sigma=\{\bm w\in\Delta^J:w_j=0\text{ for }j\notin\sigma\}.
\]
Thus $\Delta_\sigma$ contains the weights supported on the face. Choose a reference $a\in\sigma$. Define
\begin{equation}\label{eq:cone_direct}
\begin{split}
 N(\sigma)=\{\bm c\in\R^K:\;&(X_j-X_a)'\bm c=0
       \quad(j\in\sigma\setminus\{a\}),\\
 &(X_a-X_i)'\bm c\geq0\quad(i\notin\sigma)\}.
\end{split}
\end{equation}
Let $Y\in F_\sigma$. Because $\mathcal P$ is the convex hull of the donors, $\bm c\in N_{\mathcal P}(Y)$ if and only if $\bm c'X_i\leq\bm c'Y$ for every donor $i$. Suppose that $\bm c\in N(\sigma)$. Because $Y$ is a convex combination of the donors in $\sigma$, $\bm c'X_j=\bm c'Y$ for each $j\in\sigma$, and $\bm c'X_i\leq\bm c'Y$ for every donor outside $\sigma$. As a result, $\bm c\in N_{\mathcal P}(Y)$. Suppose that $Y\in\relint(F_\sigma)$, so $Y=\sum_{j\in\sigma}w_jX_j$ with every $w_j>0$. Let $\bm c\in N_{\mathcal P}(Y)$. Then $0=\sum_{j\in\sigma}w_j\,\bm c'(X_j-Y)$ is a sum of nonpositive terms with positive weights, so $\bm c'X_j=\bm c'Y$ for every $j\in\sigma$, and $\bm c\in N(\sigma)$. Therefore, $N(\sigma)\subseteq N_{\mathcal P}(Y)$ for every $Y\in F_\sigma$, and $N_{\mathcal P}(Y)=N(\sigma)$ when $Y\in\relint(F_\sigma)$.

\begin{lemma}\label{lem:exact-face}
Let $Y\in\relint(F_\sigma)$. Then $Y\in X_0\mathcal W(\mathcal V_+)$ if and only if some $\bm c\in N(\sigma)$ has exactly the same coordinate signs, including zeros, as $X_1-Y$.
\end{lemma}

\begin{proof}
Suppose $Y = X_0\bm w(V)$ for some $V=\Diag(\bm d)$ with $\bm d>0$. By~\eqref{eq:master}, $\bm c=\Diag(\bm d)(X_1-Y)$ lies in $N_{\mathcal P}(Y)=N(\sigma)$, and because every $d_k$ is positive, $c_k$ has the sign of $X_{1k}-Y_k$ for each $k$, zeros included. Conversely, let $\bm c\in N(\sigma)$ have the same signs as $X_1-Y$, zeros included. Set $d_k=c_k/(X_{1k}-Y_k)$ when $X_{1k}\neq Y_k$ and $d_k=1$ otherwise. Every $d_k$ is positive, and $\Diag(\bm d)(X_1-Y)=\bm c$. Because $Y\in\relint(F_\sigma)$, $\bm c\in N(\sigma)=N_{\mathcal P}(Y)$, so by~\eqref{eq:master} $Y$ is the projection of $X_1$ under $\Diag(\bm d)$. The normalized matrix $\Diag(\bm d)/\sum_k d_k$ belongs to $\mathcal V_+$ and has the same projection, so $Y\in X_0\mathcal W(\mathcal V_+)$.
\end{proof}

The closure of the attainable set $X_0\mathcal W(\mathcal V_+)$ can be described face by face using weak sign inequalities. For a given face $F_\sigma$, define the allowed sign patterns by
\begin{equation}\label{eq:ray_sign_test}
 \mathcal S(\sigma)=\{\bm s\in\{-1,+1\}^K:
    \exists\,\bm c\in N(\sigma)\setminus\{0\}
    \text{ with }s_kc_k\geq0\text{ for every }k\}.
\end{equation}
For each $\bm s\in\mathcal S(\sigma)$, define its \emph{sign-compatible region} and corresponding donor weights by
\begin{align}
 R_{\bm s}(\sigma)&=\{Y\in F_\sigma:s_k(X_{1k}-Y_k)\geq0
                       \text{ for every }k\},\label{eq:ray_region}\\
 W_{\bm s}(\sigma)&=\{\bm w\in\Delta_\sigma:X_0\bm w\in R_{\bm s}(\sigma)\}.\label{eq:ray_weight_region}
\end{align}
Both sets are closed.

\begin{proposition}\label{prop:master}
Suppose Assumption~\ref{ass:gp} holds. Let $\Sigma$ denote the set of index sets $\sigma$ of the faces $F_\sigma$ contained in a facet visible from $X_1$. Then,
\begin{enumerate}
\item[(i)] $X_0\overline{\mathcal W(\mathcal V_+)}
 =\bigcup_{\sigma\in\Sigma}\bigcup_{\bm s\in\mathcal S(\sigma)}R_{\bm s}(\sigma)$.
\item[(ii)] $\overline{\mathcal W(\mathcal V_+)}
 =\bigcup_{\sigma\in\Sigma}\bigcup_{\bm s\in\mathcal S(\sigma)}W_{\bm s}(\sigma)$.
\item[(iii)] The value of~\eqref{eq:single_level_W0} is the least of the attained minima
\begin{equation}\label{eq:piece_qp}
 \min_{\bm w\in W_{\bm s}(\sigma)}\norm{Z_1-Z_0\bm w}^2,
 \quad W_{\bm s}(\sigma)\ne\varnothing,
\end{equation}
 over $\sigma\in\Sigma$ and $\bm s\in\mathcal S(\sigma)$.
\end{enumerate}
\end{proposition}

\begin{proof}
Fix $\sigma\in\Sigma$ and $\bm s\in\mathcal S(\sigma)$. Some nonzero $\bm c\in N(\sigma)$ satisfies $s_kc_k\geq0$ for every $k$. The proof has three steps. First, we replace $\bm c$ by $\widetilde{\bm c}\in N(\sigma)$ whose coordinates are all nonzero with the same signs as $\bm s$. Second, we use $\widetilde{\bm c}$ to show that every $\bm w\in W_{\bm s}(\sigma)$ is attained by a positive diagonal matrix $V$ or is a limit of weights attained by positive diagonal matrices, so that $W_{\bm s}(\sigma)\subseteq\overline{\mathcal W(\mathcal V_+)}$. Third, we show that every weight vector in $\mathcal W(\mathcal V_+)$ lies in some  $W_{\bm s}(\sigma)$ with $\sigma\in\Sigma$, and that the union of these finitely many closed regions is closed, so that it also contains $\overline{\mathcal W(\mathcal V_+)}$.

\emph{First step:} Let $F$ be the face of $\mathcal P$ supported by $\bm c$, that is, the intersection of $\mathcal P$ with the supporting hyperplane with normal $\bm c$. By~\eqref{eq:cone_direct}, $F$ contains the donors in $\sigma$, and possibly other donors $X_i$ with $i\notin\sigma$ and $\bm c'X_i=\bm c'X_a$. The idea is to perturb $\bm c$ to $\widetilde{\bm c}=\bm c+\varepsilon\bm h$, with $\bm h$ chosen so that $\widetilde{\bm c}'X_j$ remains constant on the donors in $F$, while $\widetilde c_k$ has the sign $s_k$ at every coordinate where $c_k=0$. 

Let $Q=\{k:c_k\neq 0\}$ be the set of coordinates where $\bm c$ is nonzero and let $p=|Q|$. All donors $X_j$ in $F$ attain the same value of $\bm c'X_j$. Because the coordinates of $\bm c$ outside $Q$ are equal to zero, the donors in $F$, restricted to the coordinates in $Q$, lie on a hyperplane of $\R^p$. By Assumption~\ref{ass:gp}(ii), in the coordinates of $Q$, $F$ has at most $p$ donors and they are affinely independent.

Let $E$ be the matrix whose rows are the differences $(X_j-X_a)'$ over the donors $X_j$ in $F$ other than the anchor $X_a$ of~\eqref{eq:cone_direct}. Let $E_Q$ and $E_Z$ be the submatrices of $E$ formed by the columns in $Q$ and in its complement $Z$. Because the donors in $F$ are affinely independent in the coordinates in $Q$, the rows of $E_Q$ are linearly independent. For $\bm h\in\R^K$, write $\bm h_Q$ and $\bm h_Z$ for its coordinates in $Q$ and in $Z$, so that $E\bm h=E_Q\bm h_Q+E_Z\bm h_Z$. 
We choose $\bm h$ so that coordinates outside $Q$ have the same signs as in $\bm s$ and its coordinates in $Q$ are such that $E\bm h=0$, which keeps $(\bm c+\varepsilon\bm h)'X_j$ constant on the donors of $F$ for any $\varepsilon$.
In particular, we set $h_k=s_k$ for $k\in Z$ and solve
\[
E_Q\bm h_Q=-E_Z\bm h_Z
\]
for $\bm h_Q$, which is possible because the rows of $E_Q$ are linearly independent.
Choose $\varepsilon>0$ small enough so that $(\bm c+\varepsilon\bm h)'X_i<(\bm c+\varepsilon\bm h)'X_a$ for every donor $X_i$ outside $F$ and $c_k+\varepsilon h_k$ has the sign of $c_k$ for every $k\in Q$. Since $F$ contains the donors of $\sigma$, the construction guarantees $\widetilde{\bm c}=\bm c+\varepsilon\bm h\in N(\sigma)$. Moreover, $s_k\widetilde c_k>0$ for every $k$, because $\widetilde c_k=\varepsilon s_k$ for $k\in Z$, while $\widetilde c_k$ has the sign $s_k$ of $c_k$ for $k\in Q$.

\emph{Second step:} We now show that $W_{\bm s}(\sigma)\subseteq\overline{\mathcal W(\mathcal V_+)}$. Take $\bm w\in W_{\bm s}(\sigma)$, let $Y=X_0\bm w$, and let $\bm r=X_1-Y$. Because $\widetilde{\bm c}\in N(\sigma)\subseteq N_{\mathcal P}(Y)$, $\widetilde{\bm c}'(z-Y)\leq0$ for every $z\in\mathcal P$. For each coordinate $k$ with $r_k\neq0$, set $d_k=\widetilde c_k/r_k$, which is positive because $\widetilde c_k$ and $r_k$ both have the sign $s_k$. If no $r_k$ is zero, the positive diagonal matrix $V=\Diag(d_1,\ldots,d_K)$ satisfies $V(X_1-Y)=\widetilde{\bm c}\in N_{\mathcal P}(Y)$, so $Y=X_0\bm w(V)$ by~\eqref{eq:master} and $\bm w=\bm w(V)$. Otherwise, we show that $Y$ is a limit of projections of $X_1$ under strictly positive diagonal metrics. Let
\[
g(z)=\sum_{k:r_k\ne0}d_k(X_{1k}-z_k)^2,
\quad h(z)=\sum_{k:r_k=0}(X_{1k}-z_k)^2.
\]
That is, $h(z)=0$ exactly when $z$ agrees with $X_1$ on the coordinates where $Y$ does, and $g(z)$ is the weighted squared distance between $z$ and $X_1$ on the remaining coordinates. We first show that $Y$ is the point of $\mathcal P$ that minimizes $g$ among those with $h=0$. Then, we show that projections under metrics that put increasing weight on $h$ converge to $Y$.

For $z\in\mathcal P$ with $h(z)=0$, let $\bm\delta=z-Y$. If $r_k=0$, then $\delta_k=0$, because $h(z)=h(Y)=0$. If $r_k\neq0$, then $d_kr_k=\widetilde c_k$. Hence $\sum_{k:r_k\ne0}\widetilde c_k\delta_k=\sum_{k}\widetilde c_k\delta_k=\widetilde{\bm c}'(z-Y)$, and
\[
\begin{aligned}
g(z)-g(Y)&=\sum_{k:r_k\ne0}d_k\bigl[(r_k-\delta_k)^2-r_k^2\bigr]
=-2\sum_{k:r_k\ne0}d_kr_k\delta_k+\sum_{k:r_k\ne0}d_k\delta_k^2\\
&=-2\sum_{k:r_k\ne0}\widetilde c_k\delta_k+\sum_{k:r_k\ne0}d_k\delta_k^2
=-2\,\widetilde{\bm c}'(z-Y)+\sum_{k:r_k\ne0}d_k\delta_k^2
\geq\sum_{k:r_k\ne0}d_k\delta_k^2,
\end{aligned}
\]
where the last step uses $\widetilde{\bm c}'(z-Y)\leq0$.
If $z\neq Y$, some $\delta_k$ is nonzero with $r_k\neq0$. Hence $Y$ uniquely minimizes $g$ subject to $h=0$.

We next introduce a metric that puts weight $t$ on the coordinates where the residual is zero and lets $t\to\infty$, which forces the projection to match those coordinates exactly in the limit.
For $t>0$, let $V_t$ be the diagonal metric with entries $d_k$ where $r_k\ne0$ and $t$ where $r_k=0$. Rescaling a metric does not change the projection, so $V_t/\tr(V_t)\in\mathcal V_+$ gives the same projection as $V_t$. Because $V_t$ is diagonal, for every $z\in\R^K$,
\begin{align*}
\norm{X_1-z}_{V_t}^2&=(X_1-z)'V_t(X_1-z)\\&=\sum_{k:r_k\ne0}d_k(X_{1k}-z_k)^2+t\sum_{k:r_k=0}(X_{1k}-z_k)^2\\&=g(z)+th(z).
\end{align*}
Let $Y_t=X_0\bm w(V_t)$ be the projection of $X_1$ under $V_t$. Since $Y_t$ minimizes $g(z)+th(z)$ over $z\in\mathcal P$, $Y\in\mathcal P$, and $h(Y)=0$, we obtain
\[
g(Y_t)+t h(Y_t)\leq g(Y)+t h(Y)=g(Y).
\]
Both $g$ and $h$ are nonnegative, so $h(Y_t)\leq g(Y)/t$ and $g(Y_t)\leq g(Y)$. Any limit point $\bar Y$ of $Y_t$ as $t\to\infty$, which exists because $\mathcal P$ is compact, satisfies $h(\bar Y)=0$ and $g(\bar Y)\leq g(Y)$, so $\bar Y=Y$, because $Y$ is the unique minimizer of $g$ over $\{z\in\mathcal P: h(z)=0\}$, and $Y_t\to Y$. Because $\Delta^J$ is compact and $Y$ has a unique donor-weight representation, every limit point of $\bm w(V_t)$ equals $\bm w$, so $\bm w(V_t)\to\bm w$. Hence $\bm w\in\overline{\mathcal W(\mathcal V_+)}$.

\emph{Third step:} 
Let $V=\Diag(\bm d)$ be a positive diagonal metric, let $Y=X_0\bm w(V)$, and let $F_\sigma$ be the face with $Y\in\relint(F_\sigma)$. By~\eqref{eq:master}, $\bm c=V(X_1-Y)$ is a nonzero vector in $N_{\mathcal P}(Y)=N(\sigma)$. 
Choose $\bm s\in\{-1,+1\}^K$ by setting $s_k=\sgn(c_k)$ when $c_k\ne0$ and choosing either $+1$ or $-1$ when $c_k=0$. Then $\bm s\in\mathcal S(\sigma)$, and $Y$ satisfies~\eqref{eq:ray_region} because $X_1-Y=V^{-1}\bm c$ has the signs of $\bm c$. Lemma~\ref{lem:visible_projection}, proved in Appendix~\ref{app:geom_proofs}, implies that $F_\sigma$ is in a visible facet, so $\sigma\in\Sigma$. Thus the union $\bigcup_{\sigma\in\Sigma}\bigcup_{\bm s\in\mathcal S(\sigma)}W_{\bm s}(\sigma)$ contains $\mathcal W(\mathcal V_+)$. The union is finite, so it is closed and contains $\overline{\mathcal W(\mathcal V_+)}$. This proves (ii), and (i) and (iii) follow directly.
\end{proof}

\subsection{Enumerating the Visible Facets}\label{sec:app_visible}

By Lemma~\ref{lem:visible_projection}, every attainable projection lies on a visible facet. This section describes an algorithm that enumerates the visible facets of the donor hull. For a facet $F$ of $\mathcal{P}$, write its supporting hyperplane as $\bm a_F'Y+b_F=0$, oriented so that $\bm a_F'X_j+b_F\leq0$ for every donor. By Lemma~\ref{lem:visibility_facet}, it is visible if and only if $\bm a_F'X_1+b_F>0$. 

Algorithm \ref{alg:vfacets} relies on a classical fact from convex geometry. The facets visible from an exterior point form a connected subgraph of the \emph{ridge-adjacency graph} of $\mathcal{P}$, the graph whose nodes are the facets of $\mathcal{P}$, with two facets joined when they share a ridge. 

\begin{algorithm}[t]
\caption{\textsc{VisibleFacets}($X_0$, $X_1$): Visible Facet Enumeration}
\label{alg:vfacets}
\begin{algorithmic}[1]
\Require donor predictors $X_0$, treated predictors $X_1\notin\mathcal P$
\Ensure the set $\mathcal F_{\mathrm{vis}}$ of facets visible from $X_1$
\State $p \gets J^{-1}\sum_{j=2}^{J+1}X_j$
\State $F_0 \gets$ a visible facet containing the exit point of $[p,X_1]$
\State $\mathcal F_{\mathrm{vis}}\gets\{F_0\}$; \ $Q\gets[F_0]$
\While{$Q\neq\varnothing$}
 \State $F\gets\mathrm{pop}(Q)$
 \ForAll{ridges $\rho\subset F$}
  \State $G\gets$ neighbor facet across $\rho$
  \If{$G\notin\mathcal F_{\mathrm{vis}}$ and $\bm a_G'X_1+b_G>0$}
   \State $\mathcal F_{\mathrm{vis}}\gets\mathcal F_{\mathrm{vis}}\cup\{G\}$; \ $\mathrm{push}(Q,G)$
  \EndIf
 \EndFor
\EndWhile
\State \Return $\mathcal F_{\mathrm{vis}}$
\end{algorithmic}
\algnotes{Curly braces denote a set and square brackets an ordered list. $\mathrm{pop}(Q)$ removes and returns the first facet; $\mathrm{push}(Q,G)$ appends $G$.}
\end{algorithm}

\begin{lemma}\label{lem:vis-conn}
For a full-dimensional polytope and an exterior point $X_1$, any two visible facets can be joined by a sequence of visible facets in which successive facets share a ridge.
\end{lemma}
\begin{proof}
 The result of the lemma is implied by Theorem 8.12 in \citet{Ziegler1995}. 
\end{proof}

\begin{proposition}\label{prop:vfacets}
Algorithm~\ref{alg:vfacets} terminates after finitely many steps and returns exactly the visible facets.
\end{proposition}
\begin{proof}
The starting facet and every added facet are visible. Each facet is added at most once. The walk therefore terminates. By Lemma~\ref{lem:vis-conn}, every visible facet is connected to the start through visible neighbors, and the algorithm examines every ridge of every facet it adds. It consequently reaches every visible facet.
\end{proof}

\subsection{Enumerating Sign Patterns}\label{sec:app_pieces}
This section presents an algorithm that returns $\mathcal S(\sigma)$ for a face indexed by $\sigma$. Let $m$ be the number of donors in $\sigma$ and let $a\in\sigma$ be the anchor in~\eqref{eq:cone_direct}. By~\eqref{eq:ray_sign_test}, a sign pattern $\bm s$ is in $\mathcal S(\sigma)$ when some nonzero $\bm c\in N(\sigma)$ has $s_kc_k\geq0$ for every $k$. Algorithm~\ref{alg:pieces} produces a list of nonzero vectors in $N(\sigma)$ and, for each vector $\bm c$ in the list, it records every sign pattern $\bm s$ with $s_kc_k\geq0$ for all $k$. Every $\bm c\in N(\sigma)$ satisfies the $m-1$ equations $\bm c'X_j=\bm c'X_a$ for $j\in\sigma\setminus\{a\}$. The algorithm adds $K-m$ equations, each of the form $(X_a-X_i)'\bm c=0$ for a donor $i$ outside $\sigma$ or $c_k=0$ for a coordinate $k$, and solves the resulting $K-1$ equations for a direction $\bm c_T$. It keeps $\bm c_T$ or $-\bm c_T$ when the vector satisfies $(X_a-X_i)'\bm c\geq0$ for every donor $i$ outside $\sigma$.

\begin{algorithm}[t]
\caption{\textsc{SignPatterns}($\sigma$, $X_0$): Sign Enumeration}
\label{alg:pieces}
\begin{algorithmic}[1]
\Require face $\sigma$, donor predictors $X_0$
\Ensure the set $\mathcal S(\sigma)$ of allowed sign patterns in~\eqref{eq:ray_sign_test}
\State Choose anchor $a\in\sigma$; \ $m\gets|\sigma|$; \ $\mathcal S(\sigma)\gets\varnothing$
\State $E\gets[(X_j-X_a)']_{j\in\sigma\setminus\{a\}}$
\State $A\gets[(X_a-X_i)']_{i\in\{2,\ldots,J+1\}\setminus\sigma}$
\State $M\gets[A;I_K]$
\ForAll{$(K-m)$-row submatrices $T$ of $M$}
 \State Solve $E\bm c_T=0$, $T\bm c_T=0$, $\norm{\bm c_T}=1$
 \ForAll{$\bm v\in\{+\bm c_T,-\bm c_T\}$}
  \If{$A\bm v\geq0$}
   \State $\mathcal S(\sigma)\gets\mathcal S(\sigma)\cup\{\bm s\in\{-1,+1\}^K:s_kv_k\geq0\ \forall k\}$
  \EndIf
 \EndFor
\EndFor
\State \Return $\mathcal S(\sigma)$
\end{algorithmic}
\algnotes{Assumption~\ref{ass:gp}(ii) gives $\operatorname{rank}[E;T]=K-1$ for every selection. If $m=K$, $T$ has no rows and $E\bm c_T=0$ defines the facet-normal line. The rows of $A$ in $T$ can be restricted to the vertices of the facets that contain  $F_\sigma$, because only those can satisfy $(X_a-X_i)'\bm c=0$ for some nonzero $\bm c\in N(\sigma)$. Likewise, the rows of $I_K$, which bound the orthants, can be restricted to the coordinates $k$ such that $c_k=0$ for some nonzero $\bm c\in N(\sigma)$.}
\end{algorithm}

\begin{lemma}\label{lem:ray_complete}
Under Assumption~\ref{ass:gp}(ii), Algorithm~\ref{alg:pieces} returns exactly $\mathcal S(\sigma)$ in~\eqref{eq:ray_sign_test}.
\end{lemma}

\begin{proof}
For a sign pattern $\bm s$, let $O_{\bm s}=\{\bm c\in\R^K:s_kc_k\geq0\text{ for all }k\}$ be its closed orthant. By~\eqref{eq:ray_sign_test}, $\bm s\in\mathcal S(\sigma)$ when $N(\sigma)\cap O_{\bm s}$ contains a nonzero vector. The orthant is bounded by the coordinate hyperplanes $c_k=0$, which are the rows of $I_K$ in the algorithm.

We first show that the $m-1$ equations $E\bm c=0$ and the $K-m$ equations $T\bm c=0$ are linearly independent, so that their solutions form a line through the origin and the unit vector $\bm c_T$ on it is determined up to sign. The $q$ rows of $T$ from $I_K$ involve only the coordinates they set to zero. On the other $K-q$ coordinates, the remaining $K-q-1$ equations have coefficient vectors $X_j-X_a$ restricted to those coordinates, which are linearly independent by Assumption~\ref{ass:gp}(ii), so, together with the $q$ coordinate equations, the $K-1$ equations are linearly independent. The rest of the proof has two steps. First, we show that every sign pattern the algorithm returns is in $\mathcal S(\sigma)$. Second, we show that the algorithm returns every sign pattern in $\mathcal S(\sigma)$.

\emph{First step:} Each direction $\bm v$ kept by the algorithm satisfies $E\bm v=0$ and $A\bm v\geq0$, so $\bm v$ is a nonzero vector in $N(\sigma)$ by~\eqref{eq:cone_direct}. Every pattern $\bm s$ added for $\bm v$ satisfies $s_kv_k\geq0$ for all $k$, so $\bm s\in\mathcal S(\sigma)$ by~\eqref{eq:ray_sign_test}.

\emph{Second step:} The coordinate hyperplanes $c_k=0$ divide $N(\sigma)$ into pieces, one per orthant. The piece in the orthant of $\bm s$ is $C=\{\bm c\in N(\sigma):s_kc_k\geq0\text{ for every }k\}$, and by~\eqref{eq:ray_sign_test}, $\bm s\in\mathcal S(\sigma)$ means that $C$ contains a nonzero vector. Because $C$ lies in an orthant, if it contains a nonzero vector it has an extreme ray, that is, a direction $\bm c$ at which $K-1$ linearly independent equations among $E\bm c=0$ and $M\bm c=0$ hold. The equations $E\bm c=0$ account for $m-1$ of them, so the ray is fixed by $K-m$ rows of $M$. These rows form a selection $T$ that the algorithm enumerates. One of $\bm c_T=\pm\bm c/\norm{\bm c}$ is the unit vector on the extreme ray. For $\bm s\in\mathcal S(\sigma)$, this ray lies in $N(\sigma)$ and in the orthant of $\bm s$. Therefore, the orientation of $\bm c_T$ on the ray passes the test $A\bm v\geq0$, and $\bm s$ is added.
\end{proof}

\subsection{Search over Faces}\label{sec:app_bnb}
Let $L(\bm w)=\norm{Z_1-Z_0\bm w}^2$ be the validation loss. By Proposition~\ref{prop:master}(iii), the minimum of $L$ over $\overline{\mathcal W(\mathcal V_+)}$ is the least of the minima of $L$ over the $W_{\bm s}(\sigma)$, where $\sigma$ ranges over the faces of the visible facets returned by Algorithm~\ref{alg:vfacets} and $\bm s$ over the sign patterns returned by Algorithm~\ref{alg:pieces}. Our implementation optimizes over the faces in a branch-and-bound order. Every face spanned by vertices in $\sigma$ is generated by weights in $\Delta_\sigma$. Therefore, the minimum of $L$ over $\Delta_\sigma$ is a lower bound for the minimum over the weights of any of these faces. When this minimum is not below the best loss found so far, neither $\sigma$ nor its subfaces can improve on it, and the search discards them all.

\section{Proofs for Section~\ref{sec:geometric}}
\label{app:geom_proofs}

\begin{proof}[Proof of Proposition~\ref{prop:general_diag}]
By Berge's maximum theorem, the correspondence $V\mapsto\mathcal W(V)$ is upper hemicontinuous on $\mathcal V_\Delta$. Let $V_n\in\mathcal V_+$. Because $\mathcal V_\Delta$ is compact, $V_n\to V\in\mathcal V_\Delta$ along a subsequence, and by upper hemicontinuity every limit point of $\bm w(V_n)$ along it lies in $\mathcal W(V)$. Therefore $X_0\overline{\mathcal W(\mathcal V_+)}\subseteq X_0\mathcal W(\mathcal V_\Delta)$.

It remains to show that $X_0\mathcal{W}(\mathcal{V}_{\Delta}) \subseteq X_{0}\overline{\mathcal{W}(\mathcal{V}_{+})} \cup X_{0}\mathcal{W}(\mathcal{V}_{1})$. Fix $V \in \mathcal{V}_{\Delta} \setminus \mathcal{V}_{+}$, let $S = \{k : d_k > 0\}$, and let $P_S$ denote the projection onto the coordinates in $S$. Consider the following two cases.

\emph{Case 1: $\mathcal{W}(V)$ is a singleton.} Take any sequence $V_n \in \mathcal{V}_+$ with $V_n \to V$. By upper hemicontinuity, every limit point of $\bm{w}(V_n)$ lies in $\mathcal{W}(V)$, and $\Delta^J$ is compact, so $\bm{w}(V_n)\to\bm{w}(V)$. Hence $\bm{w}(V) \in \overline{\mathcal{W}(\mathcal{V}_{+})}$ and $X_0\bm{w}(V) \in X_{0}\overline{\mathcal{W}(\mathcal{V}_{+})}$.

\emph{Case 2: $\mathcal{W}(V)$ is not a singleton.} Let $\mathcal{P}_S=\conv(\{P_S X_j:j=2,\ldots,J+1\})$. Since $d_k>0$ on $S$, the inner loss is strictly convex in the fitted coordinates $y=P_S X_0\bm{w}$. Thus every optimal weight vector represents the same point $y^*\in\mathcal{P}_S$.

Suppose $P_S X_1\notin\mathcal{P}_S$. Then $y^*$ lies on the boundary of $\mathcal P_S$. Choose a facet $G$ containing $y^*$, with supporting inequality $\bm{a}'y\leq\beta$ and equality on $G$. Every $\bm{w}\in\mathcal{W}(V)$ satisfies
\[
 0=\beta-\bm{a}'y^*
 =\sum_{j=2}^{J+1}w_j\bigl(\beta-\bm{a}'P_S X_j\bigr).
\]
All summands are nonnegative, so $w_j=0$ whenever $P_S X_j\notin G$. By Assumption~\ref{ass:gp}(ii), the donor projections on $G$ are affinely independent. Their weights representing $y^*$ are therefore unique, contradicting non-uniqueness of $\mathcal{W}(V)$.

We must then conclude that $P_S X_1\in\mathcal{P}_S$, the minimum loss is zero, and every $\bm{w}\in\mathcal{W}(V)$ matches all coordinates in $S$ exactly. Choose any $k\in S$ and set $V'=\Diag(\bm{e}_k)\in\mathcal{V}_1$, where $\bm{e}_k$ is the $k$-th standard basis vector. These weights also attain zero loss under $V'$, so $\mathcal{W}(V)\subseteq\mathcal{W}(V')$ and therefore $X_0\mathcal{W}(V)\subseteq X_0\mathcal{W}(\mathcal{V}_1)$.

In both cases, $X_0\mathcal{W}(V)\subseteq X_{0}\overline{\mathcal{W}(\mathcal{V}_{+})}\cup X_{0}\mathcal{W}(\mathcal{V}_{1})$, and the same inclusion holds for $V\in\mathcal{V}_+$. 
\end{proof}

\begin{proof}[Proof of Lemma~\ref{lem:visibility_facet}]
Every $Z\in\mathcal P$ satisfies $\bm a'Z+b\leq0$, with equality on $F$. Suppose $\bm a'X_1+b>0$. For $Y\in F$ and $\lambda\in(0,1)$,
\[
\bm a'[\lambda X_1+(1-\lambda)Y]+b=\lambda(\bm a'X_1+b)>0,
\]
so no point of the open segment from $X_1$ to $Y$ lies in $\mathcal P$, and $F$ is visible.

Suppose now that $\bm a'X_1+b\leq0$, so that $X_1$ lies on the inner side of the hyperplane of $F$, or on it. Take $Y$ in the relative interior of $F$. A short step from $Y$ toward $X_1$, to $(1-\lambda)Y+\lambda X_1$ with small $\lambda>0$, does not cross the hyperplane of $F$, since $\bm a'[(1-\lambda)Y+\lambda X_1]+b=\lambda(\bm a'X_1+b)\leq0$, and it crosses no other facet. Therefore, $(1-\lambda)Y+\lambda X_1\in\mathcal P$, so $Y$ is not visible.
\end{proof}

\begin{proof}[Proof of Lemma~\ref{lem:visible_projection}]
    Consider first the Euclidean case $V = I_K$. Let $Y = X_0\bm{w}(I_K)$ and $\bm r = X_1 - Y$, so that $\bm r \neq 0$ by Assumption~\ref{ass:gp}(i). By Lemma~\ref{lem:projection_normal}, $\bm r \in N_{\mathcal{P}}(Y)$, and by \eqref{eq:normal_cone_face} there exist coefficients $\alpha_i \geq 0$ such that $\bm r = \sum_{i \in I} \alpha_i \bm{a}_i$, where $I$ indexes the facets of $\mathcal{P}$ containing $Y$. Then
    \[
        0 < \norm{\bm r}^2 = \sum_{i \in I} \alpha_i \, \bm{a}_i'\bm r,
    \]
    so $\bm{a}_i'\bm r > 0$ for at least one $i \in I$. For that facet, let $\{Z : \bm{a}_i'Z + b_i = 0\}$ be its supporting hyperplane. Because $\bm{a}_i'Y + b_i = 0$,
    \[
        \bm{a}_i'X_1 + b_i = \bm{a}_i'(Y + \bm r) + b_i = \bm{a}_i'\bm r > 0,
    \]
    so by Lemma~\ref{lem:visibility_facet} the facet is visible from $X_1$, and it contains $Y$. 

The case of general $V\in\mathbb S^K_{++}$ follows by the change of coordinates $\widetilde X=V^{1/2}X$ and the fact that linear bijections preserve visibility.
\end{proof}

\begin{proof}[Proof of Proposition~\ref{prop:visible_facets}]
    The inclusion $X_0\mathcal{W}(\mathbb{S}^K_{++}) \subseteq \bigcup\{F : F \text{ is visible}\}$ follows from Lemma \ref{lem:visible_projection}.
    For the reverse inclusion, fix a visible facet $F$ with supporting hyperplane $\{Z : \bm{a}'Z + b = 0\}$, where $\bm{a}$ is the outward normal, and let $Y \in F$ be an arbitrary point of the facet. Since $F$ contains $Y$, $\bm{a} \in N_{\mathcal{P}}(Y)$. Let $\bm r = X_1 - Y$. Visibility of $F$ and Lemma~\ref{lem:visibility_facet} give
    \[
        \bm{a}'\bm r = \bm{a}'X_1 - \bm{a}'Y = \bm{a}'X_1 + b > 0.
    \]
    By the criterion~\eqref{eq:master}, it suffices to construct $V \in \mathbb{S}^K_{++}$ with $V\bm r = \bm{a}$. Take
    \begin{equation}\label{eq:V_construction}
        V = \frac{\bm{a}\bm{a}'}{\bm{a}'\bm r} + \left(I - \frac{\bm r\bm r'}{\norm{\bm r}^2}\right).
    \end{equation}
    The matrix $V$ is symmetric, and since $(I - \bm r\bm r'/\norm{\bm r}^2)\,\bm r = 0$,
    \[
        V\bm r = \frac{\bm{a}(\bm{a}'\bm r)}{\bm{a}'\bm r} = \bm{a}.
    \]
    To prove that $V$ is positive definite, note that for any $x \neq 0$,
    \[
        x'Vx = \frac{(\bm{a}'x)^2}{\bm{a}'\bm r} +  x'\left(I - \frac{\bm r\bm r'}{\norm{\bm r}^2}\right)x \geq 0,
    \]
     where the second term is non-negative because $I - \bm r\bm r'/\norm{\bm r}^2$ is symmetric and idempotent, hence positive semidefinite, and it vanishes only if $x$ is a multiple of $\bm r$. Equality therefore requires $\bm{a}'x = 0$ and $x$ proportional to $\bm r$, which together contradict $\bm{a}'\bm r > 0$. Hence $V \in \mathbb{S}^K_{++}$ and $V(X_1 - Y) = \bm{a} \in N_{\mathcal{P}}(Y)$, so $Y = X_0\bm{w}(V) \in X_0\mathcal{W}(\mathbb{S}^K_{++})$. Since $Y$ was an arbitrary point of an arbitrary visible facet, every point of a visible facet is $X_0\bm{w}(V)$ for some $V \in \mathbb{S}^K_{++}$. Because boundary points of $\mathcal{P}$ have a unique donor-weight representation under Assumption~\ref{ass:gp}(ii), every $\bm{w} \in \Delta^J$ with $X_0\bm{w}$ on a visible facet belongs to $\mathcal{W}(\mathbb{S}^K_{++})$. Finally, $\mathcal{W}(\mathbb{S}^K_{++})$ is closed because $\bm w\mapsto X_0\bm w$ is continuous and the union of the visible facets is a finite union of closed sets.
\end{proof}

\begin{proof}[Proof of Proposition~\ref{prop:dense_V}]
    Since $\bm{w}^* \in \mathcal{W}(\mathbb{S}^K_{++})$, $\bm{w}(V_0) = \bm{w}^*$ for some $V_0 \in \mathbb{S}^K_{++}$, so $V_0 \in \mathcal{F}(\bm{w}^*)$. The constraints defining $\mathcal{F}(\bm{w}^*)$ are homogeneous in $V$, so $V_0/\tr(V_0)$ is feasible for~\eqref{eq:V_logdet} with $\log\det(V_0/\tr(V_0)) > -\infty$. The feasible set of~\eqref{eq:V_logdet} is compact, being a closed subset of $\{V \in \mathbb{S}^K_+ : \tr(V) = 1\}$, and $\log\det$ is upper semicontinuous on $\mathbb{S}^K_+$, so a maximizer exists and attains a value of at least $\log\det(V_0/\tr(V_0))$. Since $\log\det(V) = -\infty$ for any singular $V$, every maximizer is positive definite. Strict concavity of the objective function gives uniqueness. 
\end{proof}

\begin{lemma}\label{lem:strict_slack}
Under Assumption \ref{ass:gp}, for every $\bm{w}^* \in \mathcal{W}(\mathbb{S}^K_{++})$, there exists $V \in \mathcal{F}(\bm{w}^*) \cap \mathbb{S}^K_{++}$ with $s_j(V) > 0$ for every $j \notin \supp(\bm{w}^*)$. 
\end{lemma}

\begin{proof}[Proof of Lemma~\ref{lem:strict_slack}]
Let $Y = X_0\bm{w}^*$ and $\bm r = X_1 - Y$, so $\bm r \neq 0$ by Assumption~\ref{ass:gp}(i). We construct $V$ by choosing a normal $\bm{\nu}$ in the relative interior of $N_{\mathcal{P}}(Y)$ with $\bm{\nu}'\bm r>0$ and solving $V\bm r=\bm{\nu}$ with the $V$ in~\eqref{eq:V_construction}. The slack of the $j$-th inequality is then $-\bm{\nu}'(X_j-Y)$, which we show to be positive for every $j\notin\supp(\bm w^*)$ because $\bm{\nu}$ lies in the relative interior of the cone.

Because $\bm{w}^* \in \mathcal{W}(\mathbb{S}^K_{++})$, there is $V_1 \in \mathbb{S}^K_{++}$ with $\bm{w}(V_1)=\bm{w}^*$. By~\eqref{eq:master}, $\bm{c}_1 = V_1\bm r \in N_{\mathcal{P}}(Y)$, and $\bm{c}_1'\bm r = \bm r'V_1\bm r > 0$ because $V_1$ is positive definite. If $\bm c_1$ is in the relative interior of $N_{\mathcal{P}}(Y)$, take $\bm\nu = \bm c_1$. Otherwise, let $\bm{\nu}_0$ be any vector in the relative interior of $N_{\mathcal{P}}(Y)$ and set $\bm{\nu} = (1-\epsilon)\bm{c}_1 + \epsilon\bm{\nu}_0$ with $\epsilon \in (0,1]$. Then $\bm{\nu}$ lies in the relative interior of $N_{\mathcal{P}}(Y)$, and $\bm{\nu}'\bm r = (1-\epsilon)\,\bm r'V_1\bm r + \epsilon\,\bm{\nu}_0'\bm r > 0$ for $\epsilon$ small enough.

The matrix $V$ now comes from~\eqref{eq:V_construction} with $\bm{\nu}$ replacing $\bm{a}$, so $V \in \mathbb{S}^K_{++}$ and $V\bm r = \bm{\nu}$. Since $\bm{\nu} \in N_{\mathcal{P}}(Y)$, $Y = X_0\bm{w}(V)$ by~\eqref{eq:master}, so $V \in \mathcal{F}(\bm{w}^*)$.

Let $\sigma=\supp(\bm w^*)$ and take $j_0$ as the anchor in~\eqref{eq:cone_direct}. Since $Y=X_0\bm w^*$ has positive weights exactly on $\sigma$, $N_{\mathcal P}(Y)=N(\sigma)$, as shown in Appendix~\ref{sec:app_facewise}. For an inactive donor $j$, the slack $s_j(V)=(X_{j_0}-X_j)'\bm\nu$ is the left side of the inequality in~\eqref{eq:cone_direct} for the donor $j$. The normal $\bm{\nu}$ is in the relative interior of $N(\sigma)$, and by Assumption~\ref{ass:gp}(ii) no donor outside $\sigma$ lies in the affine span of the donors in $\sigma$, so $\bm{\nu}$ satisfies these inequalities strictly. As a result, $s_j(V)>0$ for every $j\notin\sigma$.
\end{proof}

\begin{proof}[Proof of Proposition~\ref{prop:extcenter}]
Let $\bm r = X_1 - X_0\bm{w}^*$, and let $\mathcal{F}_{>} = \{V \succ 0 : \tr(V) = 1, \ (X_j - X_{j_0})'V\bm r = 0 \ \forall j \in \supp(\bm{w}^*), \ s_j(V) > 0 \ \forall j \notin \supp(\bm{w}^*)\}$ be the feasible set with strict inequalities. Let $\Phi(V) = \log\det V + \sum_{j \notin \supp(\bm{w}^*)} \log s_j(V)$.

By Lemma~\ref{lem:strict_slack}, there is $V_0\in\mathcal F(\bm w^*)\cap\mathbb S^K_{++}$ with $s_j(V_0)>0$ for every $j\notin\supp(\bm w^*)$. The constraints that define $\mathcal F(\bm w^*)$ and the slacks $s_j(V)$ are homogeneous in $V$. We choose the scale of this matrix so $\tr(V_0)=1$. Because $V_0\in \mathcal F_{>}$, it follows that $\Phi(V_0)>-\infty$.

Any $V\in\mathcal{F}_{>}$ is symmetric and positive definite with trace equal to one, so its eigenvalues are positive and sum to one. Hence $\mathcal F_{>}$ is bounded and $\log\det V\leq0$ on it. Each slack is linear in $V$, so it is bounded on $\mathcal F_{>}$. Therefore $\Phi^*=\sup_{V\in\mathcal{F}_{>}}\Phi(V)$ is finite.

Choose a sequence $V_n\in\mathcal{F}_{>}$ with $\Phi(V_n)\to\Phi^*$. Boundedness of $\mathcal F_{>}$ implies that there is a convergent subsequence of $V_n$ with some limit $\bar V$. The limit $\bar V$ is positive semidefinite with $\tr(\bar V)=1$, and satisfies $(X_j-X_{j_0})'\bar V\bm r=0$ for every $j\in\supp(\bm w^*)$ and $s_j(\bar V)\geq0$ for every $j\notin\supp(\bm w^*)$ (by continuity). If $\bar V$ were singular or some $s_j(\bar V)$ were zero, then $\log\det V_n$ or $\log s_j(V_n)$ would tend to $-\infty$ while the other terms of $\Phi$ stay bounded above, so $\Phi(V_n)\to-\infty$, contradicting $\Phi^*\geq\Phi(V_0)>-\infty$. It follows that $\bar V\in\mathcal F_{>}$ attains the maximum.

The set $\mathcal F_{>}$ is convex, because its constraints are linear and the positive definite matrices form a convex set. Because $\Phi$ is strictly concave, the maximizer of $\Phi$ over $\mathcal F_{>}$ is unique.
\end{proof}

\section[The Interior Case: Details and Proofs]{The Interior Case: Details and Proofs}
\label{app:geom}

This appendix proves interior-case results. Let
\[
\cF=\{\bm w\in\Delta^J:X_0\bm w=X_1\}
\]
be the set of exact matches. Throughout, we use the quadratic lift
$\ell_V(x)=\norm{x}_V^2=x'Vx$ and the pairwise penalty
$g_V(\bm w)=\sum_j w_j\norm{X_1-X_j}_V^2$ from~\eqref{eq:penalized}.
The lift is linear in $V$ and, for $V\succ0$, strictly convex in $x$.

Our starting point is the vanishing-penalty selection principle of
\citet[Section~2.2]{abadie2021penalized}. 
The following lemma states this principle in our notation, with a short proof for completeness.

\begin{lemma}\label{lem:selection_principle}
Let $V\succ0$ and suppose $\cF\neq\varnothing$. Every accumulation point of $\bm w(\lambda,V)$ as $\lambda\downarrow 0$ minimizes $g_V$ over $\cF$. If this minimizer is unique, then $\bm w(\lambda,V)$ converges to it.
\end{lemma}

\begin{proof}
Let $\bm w_\lambda=\bm w(\lambda,V)$ and fix any $\bm w\in\cF$. Because $\Delta^J$ is compact, accumulation points of $\bm w_\lambda$ always exist. Because $\bm w_\lambda$ minimizes the objective in~\eqref{eq:penalized} and $X_0\bm w=X_1$,
\[
\norm{X_1-X_0\bm w_\lambda}_V^2+\lambda g_V(\bm w_\lambda)\leq\lambda g_V(\bm w).
\]
The term $\|X_1-X_0\bm w_\lambda\|^2_V$ on the left-hand side is nonnegative and bounded by $\lambda g_V(\bm w)$, which goes to zero as $\lambda\downarrow 0$. As a result, every accumulation point $\bar{\bm w}$ of $\bm w_\lambda$ satisfies $X_0\bar{\bm w}=X_1$.  The second term gives $g_V(\bm w_\lambda)\leq g_V(\bm w)$. By continuity, $g_V(\bar{\bm w})\leq g_V(\bm w)$. Since $\bm w\in\cF$ was arbitrary, $\bar{\bm w}$ minimizes $g_V$ over $\cF$. If the minimizer is unique, $\bar{\bm w}$ is the only accumulation point of $\bm w_\lambda$, and $\bm w_\lambda$ converges to it.
\end{proof}

We next characterize the metrics for which the
exact-match weights $\bm w_A$ supported on an admissible cell
$A\subseteq\cD$, $|A|=K+1$, uniquely minimize $g_V$ over $\cF$. Recall from
Section~\ref{sec:int_cells} that
\[
\begin{aligned}
H_{A,V}(x)&=\sum_{i\in A}\beta_i(x)\ell_V(X_i),\\
x&=\sum_{i\in A}\beta_i(x)X_i,\quad
\sum_{i\in A}\beta_i(x)=1,
\end{aligned}
\]
where the affine coordinates $\beta_i(x)$ are unique by
Assumption~\ref{ass:gp_int}(iii). Thus $H_{A,V}(x)$ is the height above $x$ of the hyperplane that goes through the lifted donors $(X_i,\ell_V(X_i))$, $i\in A$. For each donor $j\in\cD$, 
\[
q_j^A(V)=\ell_V(X_j)-H_{A,V}(X_j)=\inner{Q_j^A}{V}
\]
 is zero on $A$ and positive, zero, or negative for other donors depending on whether their lifted points lie above, on, or below the hyperplane $H_{A,V}$. The matrix $V$ belongs to $\cV(A)$ when $q_j^A(V)>0$ for every donor outside $A$. 

\begin{lemma}
\label{lem:selection}
Let Assumption~\ref{ass:gp_int} hold, $A$ be admissible, and $V$ belong to $\cV(A)$. Then $\bm{w}_A$ is the unique minimizer of $g_V$ over $\cF$, and $\bm{w}(\lambda, V) \to \bm{w}_A$ as $\lambda\downarrow 0$. Conversely, if $q_j^A(V)<0$ for some $j\notin A$, then $\bm{w}_A$ does not minimize $g_V$ over $\cF$, and if $q_j^A(V)=0$ for some $j\notin A$, then $\bm{w}_A$ is not the unique minimizer.
\end{lemma}

\begin{proof}
By~\eqref{eq:int_height}, $g_V$ and $\sum_j w_j\ell_V(X_j)$ differ only by the constant $\ell_V(X_1)$ on $\cF$, so minimizing $g_V$ over $\cF$ is the same as minimizing the lifted height $\sum_j w_j\ell_V(X_j)$ over $\cF$. By the definition of $q_j^A(V)$, $\ell_V(X_j)=H_{A,V}(X_j)+q_j^A(V)$ for every donor, and $q_j^A(V)=0$ for donors in $A$. Because $H_{A,V}$ is affine, for any $\bm w\in\cF$,
\[
\sum_j w_j\ell_V(X_j)=H_{A,V}(X_1)+\sum_{j\notin A}w_jq_j^A(V).
\]
If $V\in\cV(A)$, then $q_j^A(V)>0$ for every donor $j\notin A$. Since all donor weights are nonnegative, the sum $\sum_{j\notin A}w_jq_j^A(V)$ is nonnegative and equals zero if and only if $w_j=0$ for every $j\notin A$. Therefore, $\sum_j w_j\ell_V(X_j)\geq H_{A,V}(X_1)$ for every $\bm w\in\cF$, with equality only if $w_j=0$ for every $j\notin A$. Since $\bm w_A$ is the only exact match supported on $A$, it is the unique minimizer of $\sum_j w_j\ell_V(X_j)$ over $\cF$ and, therefore, the unique minimizer of $g_V$. Convergence of $\bm w(\lambda,V)$ to $\bm w_A$ then follows from Lemma~\ref{lem:selection_principle}.

Conversely, suppose $q_j^A(V)\leq0$ for some $j\notin A$. We construct another exact match by giving donor $j$ a small positive weight and adjusting the weights on $A$. Since $A$ is admissible and $X_1$ lies in the interior of its simplex by Assumption~\ref{ass:gp_int}(ii), $(\bm w_A)_i > 0$ for every $i \in A$. We can therefore choose $t\in(0,1)$ small enough that all the following weights are nonnegative:
\[
\widetilde w_k=
\begin{cases}
(\bm w_A)_k-t\beta_k(X_j), & k\in A,\\
t, & k=j,\\
0, & k\notin A\cup\{j\}.
\end{cases}
\]
Then,
\[
\begin{aligned}
\sum_k\widetilde w_k
&=1-t\sum_{i\in A}\beta_i(X_j)+t=1,\\
X_0\widetilde{\bm w}
&=X_1-t\sum_{i\in A}\beta_i(X_j)X_i+tX_j=X_1.
\end{aligned}
\]
As a result, $\widetilde{\bm w}\in\cF$, and it differs from $\bm w_A$ because $\widetilde w_j=t>0$. Since both are exact matches,
\[
\begin{aligned}
g_V(\widetilde{\bm w})-g_V(\bm w_A)
&=\sum_k (\widetilde w_k -(\bm w_A)_k)\ell_V(X_k)\\
&=t\left[\ell_V(X_j)-\sum_{i\in A}\beta_i(X_j)\ell_V(X_i)\right]\\
&=tq_j^A(V)\leq0.
\end{aligned}
\]
Hence $g_V(\widetilde{\bm w})<g_V(\bm w_A)$ when $q_j^A(V)<0$, and $g_V(\widetilde{\bm w})=g_V(\bm w_A)$ when $q_j^A(V)=0$.
\end{proof}

Condition~\eqref{eq:VA} is the Delaunay empty-circumsphere criterion for $A$ in the coordinates $\widetilde X_j=V^{1/2}X_j$.

The next result uses Lemma~\ref{lem:selection} to show that, outside a set of metrics of measure zero, the limit $\bm w(0^+,V)$ \citep[i.e., the pure synthetic control of][]{abadie2021penalized} exists and is unique. For a subset $B\subseteq\cD$ with $|B|=K+2$, define 
\[
\varphi_B(V)=\det\begin{pmatrix}1 & X_j' & \ell_V(X_j)\end{pmatrix}_{j\in B}.
\]
The function $\varphi_B(V)$ is linear in $V$, and equal to zero if and only if the lifted points of $B$ are affinely dependent, that is, when the transformed donors $\widetilde X_j$, $j\in B$, lie on a common sphere.

\begin{proposition}
\label{prop:coverage}
Let Assumption~\ref{ass:gp_int} hold and 
\[
\cN = \bigcup_{B \subseteq \cD, \ |B| = K+2} \left\{ V \in \mathbb{S}^K_{++} : \varphi_B(V) = 0 \right\}.
\]
The set $\cN$ has measure zero in $\mathbb{S}^K_{++}$, and $\{\bm d>0 : \Diag(\bm d)\in\cN\}$ has measure zero in $\R^K_{++}$. For every $V \in \mathbb{S}^K_{++} \setminus \cN$ the program~\eqref{eq:int_limit} has a unique solution, equal to $\bm w_A$ for an admissible cell $A$ with $V \in \cV(A)$, and $\bm{w}(\lambda, V) \to \bm{w}_A$ as $\lambda \downarrow 0$.
\end{proposition}

\begin{proof}
By Assumption~\ref{ass:gp_int}(iv), for each $B$, $\varphi_B$ is nonzero at some positive definite $V$ and at some positive diagonal $V$. Because $\varphi_B$ is linear in $V$, the equation $\varphi_B(V)=0$ is a single linear equation in the entries of $V$ with some nonzero coefficient. Its solutions determine one entry of $V$ as a function of the others, so they form a set of measure zero. Restricted to diagonal $V=\Diag(\bm d)$, the same holds for a single linear equation in $\bm d$. The set $\cN$ is the union of these solution sets over the finitely many $B$, and so is its diagonal counterpart, and a finite union of measure-zero sets has measure zero.

Fix $V\notin\cN$. For any $\bm w\in\cF$, the point $\sum_j w_j\,(X_j,\ell_V(X_j))$ of the lifted donor hull lies directly above $X_1$, at height $\sum_j w_j\ell_V(X_j)$. By~\eqref{eq:int_height}, $\bm w$ minimizes $g_V$ over $\cF$ when this point is the lowest point of the lifted donor hull above $X_1$. The lowest point lies on the boundary of the hull, hence on a face of dimension at most $K$, so by Carath\'eodory's theorem it is a convex combination of at most $K+1$ lifted donors. Let $\bm w^*\in\cF$ be its weights, which minimize $g_V$. By Assumption~\ref{ass:gp_int}(ii) and (iii), $|\supp(\bm w^*)|=K+1$ and $A=\supp(\bm w^*)$ is affinely independent, so $\bm w^*=\bm w_A$ with all weights positive, and $A$ is admissible.

We next show that $V\in\cV(A)$. Take $j\notin A$. If $q_j^A(V)<0$, then $\bm w_A$ does not minimize $g_V$ by Lemma~\ref{lem:selection}. If $q_j^A(V)=0$, then $\ell_V(X_j)=H_{A,V}(X_j)$, so the lifted points of $A\cup\{j\}$ lie on the graph of $H_{A,V}$ and are affinely dependent, which gives $\varphi_{A\cup\{j\}}(V)=0$ and $V\in\cN$. Hence $q_j^A(V)>0$ for every $j\notin A$, that is, $V\in\cV(A)$. By Lemma~\ref{lem:selection}, $\bm w_A$ is the unique minimizer of $g_V$ over $\cF$ and $\bm w(\lambda,V)\to\bm w_A$ as $\lambda\downarrow 0$.
\end{proof}

\begin{proof}[Proof of Proposition~\ref{prop:enumeration}]
First, notice that if $V\in\mathcal{C}\setminus\cN$, then by Proposition~\ref{prop:coverage}, $\bm w(0^+,V)=\bm w_A$ for an admissible $A$ with $V\in\cV(A)$. As a result, $A\in\mathcal{A}_{\mathcal{C}}$. In addition, if $A\in\mathcal{A}_{\mathcal{C}}$, then $\cV(A)\cap\mathcal{C}$ is a nonempty open subset of $\mathcal{C}$ and $\cN$ has measure zero in $\mathcal{C}$, so $(\cV(A)\cap\mathcal{C})\setminus\cN$ is not empty, and for $V$ in this set $\bm w(0^+,V)=\bm w_A$ by Lemma~\ref{lem:selection}. Hence $\bm w(0^+,V)$ takes exactly the finitely many values $\bm w_A$, $A\in\mathcal{A}_{\mathcal{C}}$, and the smallest validation loss among them is the global minimum.
\end{proof}

\section{The Full Positive Definite Class in the Applications}
\label{app:pd}

This appendix reports results under the full positive definite class for the two applications of Sections~\ref{sec:numerical} and~\ref{sec:temporal_split}. By Proposition~\ref{prop:visible_facets}, the attainable synthetic controls under $\mathbb{S}^K_{++}$ are the closed visible facets. There are 132 visible facets in the California application, and 39 in the West Germany application. Tables~\ref{tab:pd_ca} and~\ref{tab:pd_ger} compare the resulting donor weights, validation losses, and estimated effects with the diagonal-class solutions of the main text. Figures~\ref{fig:pd_ca} and~\ref{fig:pd_ger} display the corresponding synthetic controls.

For the California example, the full class improves the validation loss from $58.5$ to $54.7$. The selected donors stay unchanged with some changes in the weights. The estimated effect of Proposition 99 is essentially unchanged. 

For the West Germany application, the validation loss falls by more than half, and the selected donors change. Austria's weight falls from $0.51$ to $0.19$, the United States and Australia drop out, and the Netherlands, Norway, and Belgium enter the mix. The estimated effect, however, remains substantially unchanged.

\begin{table}[t]
\centering
\caption{Donor weights for synthetic California: diagonal and full positive definite classes}
\label{tab:pd_ca}
\begin{tabular}{lcc}
\hline
State & Diagonal class & Full positive definite class \\
\hline
Utah        & 0.335 & 0.319 \\
Nevada      & 0.236 & 0.248 \\
Montana     & 0.202 & 0.302 \\
Colorado    & 0.160 & 0.028 \\
Connecticut & 0.068 & 0.102 \\
\hline
Validation loss (1970-1988) & 58.46 & 54.70 \\
Mean post-treatment gap, 1989-2000 (packs) & $-19.0$ & $-19.8$ \\
\hline
\end{tabular}
\tabnotes{The first column reports the exact global solution over the closure of the attainable set of the positive diagonal class (Table~\ref{tab:donor-weights}). The second column reports the exact global solution over the attainable set of the full positive definite class. Both optima are supported on the same five states. The remaining 33 donors receive zero weight in both columns and are omitted.}
\end{table}

\begin{figure}[t]
    \centering
    \caption{California under the diagonal and full positive definite classes}
    \label{fig:pd_ca}
    \includegraphics[width=0.95\linewidth]{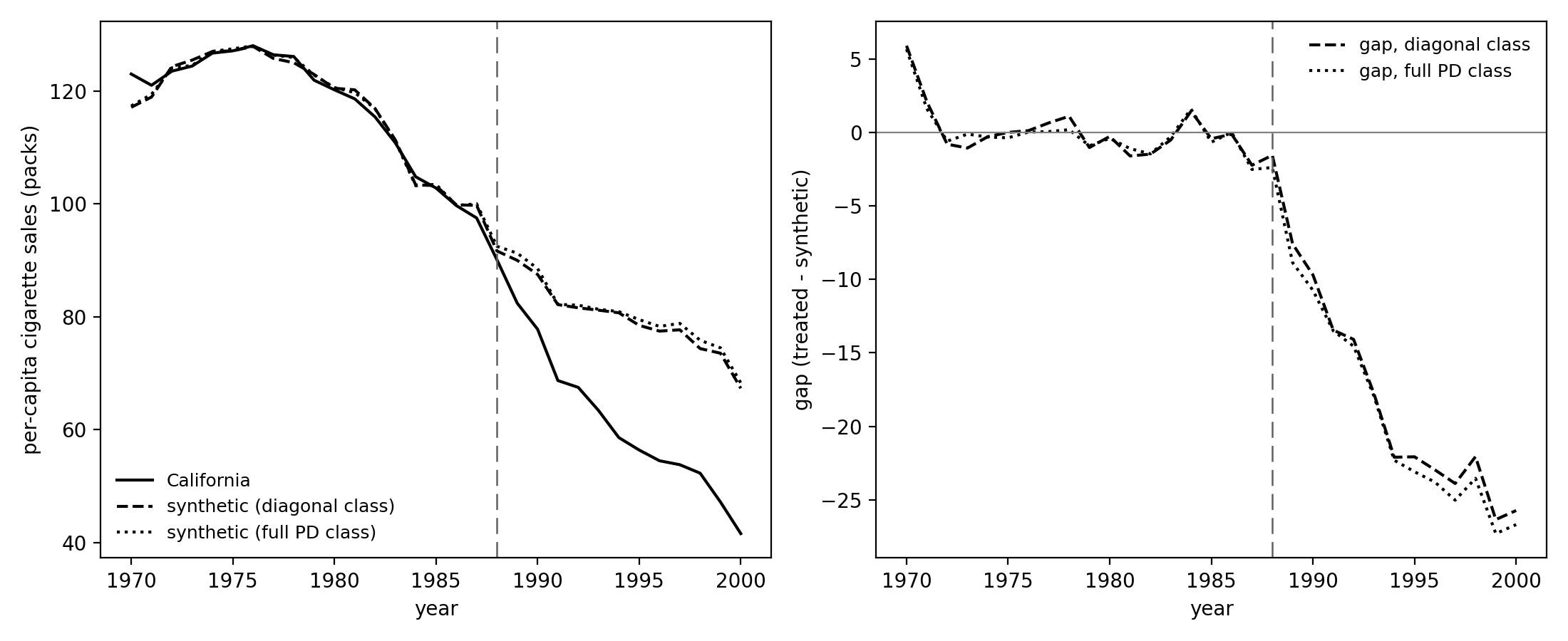}
    \tabnotes{Left panel: per-capita cigarette sales in packs for California, and for the exact synthetic controls of the diagonal and full positive definite classes. Right panel: gaps between California and each synthetic control. The vertical line at 1988 marks the last pre-treatment year.}
\end{figure}

\begin{table}[t]
\centering
\caption{Donor weights for synthetic West Germany: diagonal and full positive definite classes}
\label{tab:pd_ger}
\begin{tabular}{lcc}
\hline
Donor & Diagonal class & Full positive definite class \\
\hline
Austria      & 0.507 & 0.191 \\
Switzerland  & 0.166 & 0.301 \\
Japan        & 0.146 & 0.206 \\
USA          & 0.135 & 0.000 \\
Australia    & 0.045 & 0.000 \\
Netherlands  & 0.000 & 0.133 \\
Norway       & 0.000 & 0.161 \\
Belgium      & 0.000 & 0.009 \\
\hline
Validation loss (1981-1990) & 45{,}804 & 20{,}287 \\
Mean post-treatment gap, 1991-2003 (USD) & $-1{,}391$ & $-1{,}582$ \\
\hline
\end{tabular}
\tabnotes{Temporal-split design of Section~\ref{sec:temporal_split}. The first column reports the exact global solution over the closure of the attainable set of the positive diagonal class (Table~\ref{tab:germany_weights}). The second column reports the exact global solution over the attainable set of the full positive definite class. Donors receiving zero weight in both columns are omitted.}
\end{table}

\begin{figure}[t]
    \centering
    \caption{West Germany under the diagonal and full positive definite classes}
    \label{fig:pd_ger}
    \includegraphics[width=0.95\linewidth]{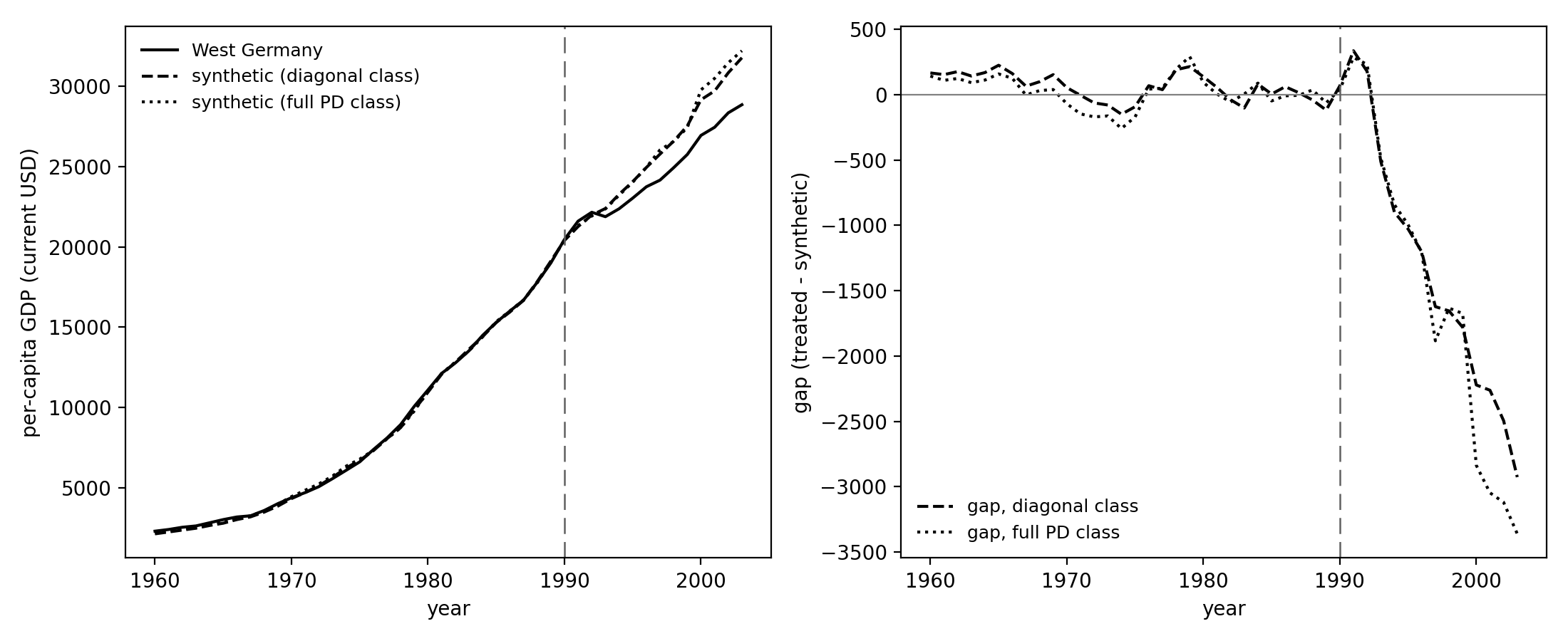}
    \tabnotes{Left panel: per-capita GDP in PPP-adjusted current U.S. dollars for West Germany, and for the exact synthetic controls of the diagonal and full positive definite classes. Right panel: gaps between West Germany and each synthetic control. The vertical line at 1990 marks the German reunification.}
\end{figure}

\makeatletter\setlength{\@fptop}{0pt}\setlength{\@fpsep}{\floatsep}\makeatother
\clearpage
\bibliographystyle{econ}
\bibliography{Bibliography}
\addcontentsline{toc}{section}{Reference}
\end{document}